\documentclass[10pt,twocolumn,letterpaper]{article}

\usepackage{iccv}              % To produce the CAMERA-READY version
\usepackage{amsfonts, amssymb, amsthm, amsmath}
\usepackage{multicol, multirow}
\usepackage{algorithm, algorithmic}
\usepackage{comment}
\usepackage{wrapfig}

\newcommand{\db}{\texttt{DB}}

\newtheorem{definition}{Definition}
\newtheorem{proposition}{Proposition}
\newtheorem{lemma}{Lemma}
\newtheorem{fact}{Fact}
\newtheorem{assumption}{Assumption}

\usepackage{kotex}

\definecolor{iccvblue}{rgb}{0.21,0.49,0.74}
\usepackage[pagebackref,breaklinks,colorlinks,allcolors=iccvblue]{hyperref}

\def\paperID{15047} % *** Enter the Paper ID here
\def\confName{ICCV}
\def\confYear{2025}

\newcommand{\Hquad}{\hspace{0.75em}}

\title{IDFace: Face Template Protection for Efficient and Secure Identification}

\author{Sunpill Kim$^{1\dagger}$\Hquad Seunghun Paik$^{1\dagger}$\Hquad Chanwoo Hwang$^{1}$\Hquad Dongsoo Kim$^{1}$\Hquad Junbum Shin$^{2}$\Hquad Jae Hong Seo$^{1}$\thanks{Corresponding author}\\
$^{1}$Department of Mathematics \& Research Institute for Natural Sciences, Hanyang University\\
$^{2}$CryptoLab Inc.\\
{\tt\small \{ksp0352, whitesoonguh, aa5568, frds37, jaehongseo\}@hanyang.ac.kr, junbum.shin@cryptolab.co.kr}
}

\begin{document}
\maketitle

\def\thefootnote{$\dagger$}\footnotetext{These co-first authors contributed equally to this work}\def\thefootnote{\arabic{footnote}}

\begin{abstract}
As face recognition systems (FRS) become more widely used, user privacy becomes more important. 
A key privacy issue in FRS is protecting the user’s face template, as the characteristics of the user’s face image can be recovered from the template. Although recent advances in cryptographic tools such as homomorphic encryption (HE) have provided opportunities for securing the FRS, HE cannot be used directly with FRS in an efficient plug-and-play manner. 
In particular, although HE is functionally complete for arbitrary programs, it is basically designed for algebraic operations on encrypted data of predetermined shape, such as a polynomial ring. 
Thus, a non-tailored combination of HE and the system can yield very inefficient performance, and many previous HE-based face template protection methods are hundreds of times slower than plain systems without protection. 
In this study, we propose $\mathsf{IDFace}$, a new HE-based secure and efficient face identification method with template protection. 
$\mathsf{IDFace}$ is designed on the basis of two novel techniques for efficient searching on a (homomorphically encrypted) biometric database with an angular metric. 
The first technique is a template representation transformation that sharply reduces the unit cost for the matching test. 
The second is a space-efficient encoding that reduces wasted space from the encryption algorithm, thus saving the number of operations on encrypted templates. 
Through experiments, we show that $\mathsf{IDFace}$ can identify a face template from among a database of 1M encrypted templates in 126ms, showing only \textcolor{red}{$2\times$} overhead compared to the identification over plaintexts.
\vspace{-4mm}
\end{abstract}

\section{Introduction}
Recent advancements of the metric learning on the hypersphere have greatly improved the performance of face recognition systems (FRS)~\cite{deng2019arcface, meng2021magface, wen2022sphereface2, boutros2022elasticface, kim2022adaface, zhou2023uniface, jia2023unitsface, kim2024keypoint}, enabling various real-world applications of them. However, as the advanced FRS extracts more discriminative characteristics from a user's face image to his/her face template, privacy concerns regarding the secure storage of these templates become intensified. The advances on FRSs have focused on improving performance without protecting users' face templates, and using them without appropriate template protection methods can cause severe privacy problems. For example, several studies~\cite{mai2018reconstruction,duong2020vec2face, kansy2023controllable, shahreza2023face, shahreza2023template, jung2024face, shahreza2024face, kim2024scores} showed that face images can be reconstructed from corresponding unprotected face templates, even without seeing the internal parameters of the target FRS. Even worse, the reconstructed faces can be exploited for further malicious purposes, \textit{e.g.}, impersonating the identity of the recovered faces in other FRSs, including commercial ones~\cite{wenger2022assessing, kim2024scores}. Hence, for real-world applications such as identification at the airport or entrance of a building, as well-documented in GDPR~\cite{voigt2017eu}, secure storage of templates becomes one of the most important security goals.

\begin{figure*}[t]
    \vspace{-0.5mm}
    \centering
    \includegraphics[width=\textwidth]{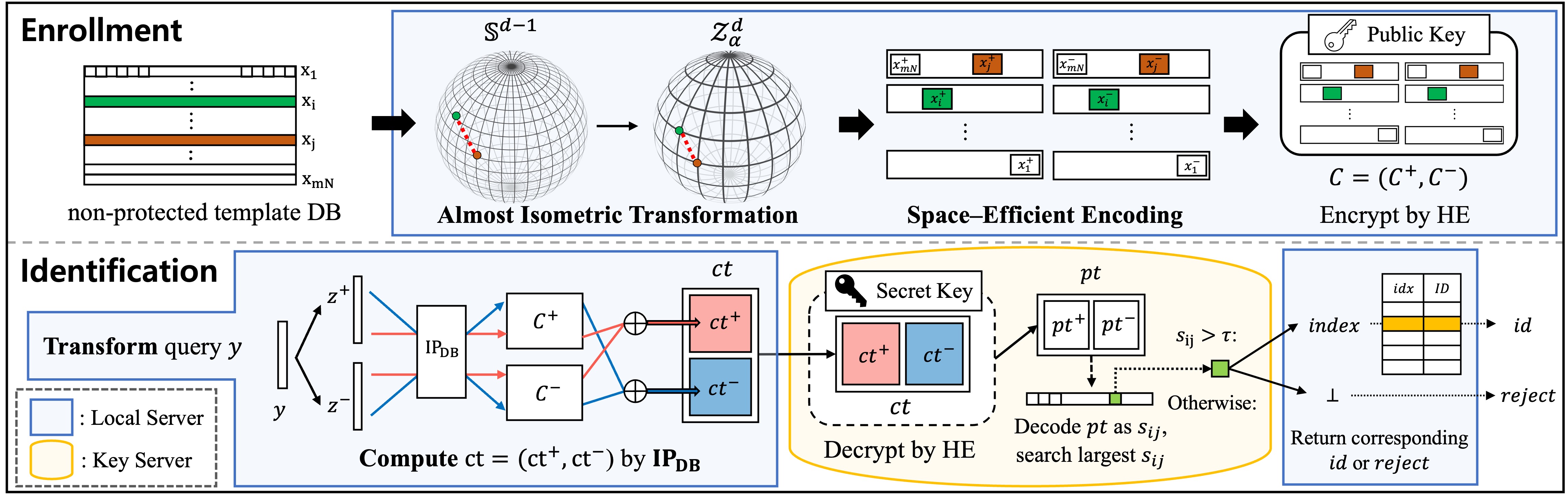}
    \vspace{-7mm}
    \caption{Overview of $\mathsf{IDFace}$. Detailed explanations about each component will be presented in Section~\ref{sec:4}. Best viewed in color.\label{fig1}
    }     
    \vspace{-5mm}
\end{figure*}

Faces are a form of biometrics, and biometric template protection (BTP) has been researched for a long time in various types of biometrics, including irises~\cite{hao2006combining, yang2007secure, lai2017cancellable}, fingerprints~\cite{nandakumar2007fingerprint, tuyls2005practical, jin2017ranking, choi2024blindtouch}, and recently faces~\cite{boddeti2018secure, kim2021ironmask, lai2021efficient, engelsma2022hers, choi2024blind}. The security notions of BTP are well established and have three standard requirements according to ISO/IEC 24745~\cite{ISO24745}: irreversibility, revocability, and unlinkability. For pursuing concrete and provable security, cryptographic tools such as cryptographic hash functions and homomorphic encryptions (HE)~\cite{gentry2009fully} are widely employed on BTP designs. Despite strong security guarantees from these cryptographic tools, simultaneously achieving an appropriate level of security without sacrificing performance is still challenging. One obstacle is that cryptographic tools are designed to deal with specific types of data, \textit{e.g.}, bit-strings or polynomials. In particular, contrary to classical biometrics such as iris and fingerprint, whose templates are usually represented as bit-strings~\cite{daugman2004iris, maltoni2009handbook}, most FRSs deal with real-valued templates. This makes designing BTPs for faces more challenging than other classical biometrics. Moreover, for applying BTPs in closed-set face recognition scenarios where the user's images can be utilized for training or fine-tuning the FRS, additional training for template protections is allowed to improve performance~\cite{kumar2016deep, talreja2019zero, dang2020fehash, osorio2022stable}. However, this approach is not plausible in open-set face recognition scenarios where the train dataset is completely independent from users of a FRS.
 
FRSs are generally deployed in two different scenarios: verification and identification. While one identity is enrolled and verified in the verification scenario, multiple identities are enrolled and the closest one is identified in the identification scenario. In practice, rapid verification/identification procedures are usually more important than rapid enrollment. Contrary to the verification scenario, using heavy cryptographic tools for protection can be a burden in large-scale identification scenarios such as in airports, because their computational cost increases linearly with the number of identities. Recently, there have been several BTP proposals for identification~\cite{drozdowski2019application, Bausp2022improved, jindal2020secure, engelsma2022hers, huang2023efficient, osorio2022stable, drozdowski2021feature, bauspiess2023hebi, choi2024blind} using HE, which enables arbitrary computation on encrypted data. However, despite several optimization techniques through SIMD (Single Instruction, Multiple Data) operations, the resulting performance of theirs remains impractical for large-scale applications.

\subsection{Contribution}

We find that the main cause of inefficiency in all previous approaches is that their use of HE was not tailored to coordinate with face templates. To overcome this limitation, we introduce two novel ideas to properly process templates so that they can be efficiently combined with HE. With these ideas, we propose $\mathsf{IDFace}$, a \textit{plug-and-play} HE-based face template protection method specialized for large-scale identification. Our $\mathsf{IDFace}$ achieves a drastic improvement in efficiency without significant accuracy degradation.
In Fig.~\ref{fig1}, we illustrated the brief overview of the proposed $\mathsf{IDFace}$.

Our first idea is employing a novel transformation from a real-valued unit vector to a ternary vector comprised of $\{0, \pm1\}$, while largely preserving the distance relationship. This facilitates harmonizing the aforementioned cryptographic tools with FRSs without significant accuracy degradation. In particular, our transformation enables computing the inner product from only additions. Since multiplications over encrypted data are much more expensive than additions, it massively accelerates the calculation of the inner product between templates when combined with HE.

Our second idea is introducing a space-efficient encoding for HE-based template protection schemes. We found that, for security reasons, each message slot of HE should occupy a considerably larger number of bits than that necessary for representing feature vectors. This gap becomes larger after applying the proposed transformation. To exploit the full capacity of each message slot, we propose a new encoding of feature templates to process much more of them simultaneously. This maximizes the advantage of SIMD operations while saving a large amount of storage.

We conducted extensive analyses on $\mathsf{IDFace}$. First, we showed that our transformation does preserve the distance relationship to some extent under the suggested parameter settings. In addition, we experimentally verified that the accuracy drop by $\mathsf{IDFace}$ for various FRSs is less than $1\%$ on famous benchmark datasets, including LFW~\cite{huang2008labeled}, CFP-FP~\cite{sengupta2016frontal}, AgeDB~\cite{moschoglou2017agedb}, and IJB-C~\cite{maze2018iarpa}. 
The time spent for identification in $\mathsf{IDFace}$ is 126ms-753ms for a database with 1 million identities, depending on the amount of accuracy degradation, which only takes $\mathbf{2\times}$-$\mathbf{12\times}$ overhead compared to the setting without protection.

\subsection{Related Works}

We briefly review previous HE-based BTP proposals. In this paper, we will focus on HE-based BTPs only; surveys on other types of BTPs, \textit{e.g.}, FC-based~\cite{talreja2019zero, jindal2019securing, kumar2018face, mohan2019significant, kim2021ironmask} or obfuscation-based~\cite{mi2024privacy, jin2024faceobfuscator, wang2023privacy}, are deferred to Appendix~\textcolor{iccvblue}{A}.

HE is a famous tool for privacy-preserving computation, and lots of HE-based BTPs have been proposed~\cite{boddeti2018secure, drozdowski2019application, Bausp2022improved, jindal2020secure, engelsma2022hers, huang2023efficient, osorio2022stable, drozdowski2021feature, bauspiess2023hebi, choi2024blind}. The main advantage of utilizing HE is two-fold: since HE enables the computation of the matching score on an encrypted domain, this approach not only stores data and computes scores securely but also does not result in accuracy degradation. Moreover, thanks to SIMD operations, HE-based BTPs can leverage parallel computation in large-scale identification scenarios.

However, as \cite{drozdowski2019application} reported, a naive application of HE is quite far from the practical applications. Recent studies on designing HE-based BTPs have focused on optimization techniques to reduce the enormous computational cost of computing the matching score. \cite{boddeti2018secure} proposed a method to reduce the number of ciphertext rotations when computing the inner product. \cite{Bausp2022improved, jindal2020secure, choi2024blind} focused on the waste of message slots in each ciphertext.
\cite{engelsma2022hers} proposed a novel idea to encrypt the database for computing matching scores of a bunch of enrolled templates at the same time.
Recently, a series of works~\cite{bassit2022multiplication, bassit2023improved, bassit2025practical} proposed a look-up table-based method during the score calculation, thus eliminating the need for homomorphic multiplications.
In other directions, some proposals attempted to reduce the number of calculated matching scores~\cite{osorio2022stable, drozdowski2021feature, bauspiess2023hebi} or propose an efficient matching algorithm~\cite{ibarrondo2023grote} using scores.
Many of them, including~\cite{boddeti2018secure, engelsma2022hers, Bausp2022improved}, considered dimensionality reduction for further speed-up on the identification process while sacrificing accuracy. Nevertheless, we found that all the previous approaches still suffer from expensive computational costs when a large number of identities, \emph{e.g.} 1 million, are enrolled into the system.
More detailed comparisons between theirs and our $\mathsf{IDFace}$ will be given in Tab.~\ref{tab:comp_res}, Section~\ref{sec:5_1}.

\section{Preliminaries}
\noindent \textbf{Notation.} Throughout this paper, we will use the following notation. A vector is denoted by a lowercase letter shown in bold, such as $\mathbf{x} \in \mathbb{R}^{d}$, the components of which are denoted as a tuple $(x_{1}, x_{2}, \dots, x_{d})$. The standard inner product is denoted by $\langle\cdot,\cdot\rangle$. We denote $\mathbb{S}^{d-1}$ as a unit $d$-dimensional sphere embedded on $\mathbb{R}^{d}$. We denote a set of non-negative numbers less than $n\in \mathbb{N}$ as $[n]$.

\subsection{Biometric Identification System}

We first formalize the biometric identification system with a feature extractor, such as FRSs. For the set $\mathcal{I}$ of (preprocessed) biometrics and a metric space $(\mathcal{X}, d_\mathcal{X})$, the feature extractor is a function from $\mathcal{I}$ to $\mathcal{X}$. The output of the feature extractor is called (unprotected) biometric template. The feature extractor is expected to be a \textit{similarity-preserving mapping} from the biometrics to the templates: two biometric templates from the same identity should be close enough, or vice versa. We remark that many recently proposed face recognition models~\cite{deng2019arcface, meng2021magface, wen2022sphereface2, boutros2022elasticface, kim2022adaface, zhou2023uniface, jia2023unitsface, kim2024keypoint} adopt $\mathcal{X} = \mathbb{S}^{d-1}$ for $d = 512$ with angular metric as $d_\mathcal{X}$. 

The operation of a biometric identification system consists of two phases, enrollment and identification, between two parties: the server and the user. In the enrollment, the user requests to enroll in the system by presenting his/her biometrics and an identifier. Then the server extracts a template from the user's biometrics through the feature extractor and stores it in the server's database with the corresponding identifier. In the identification, the user requests to identify him/herself by presenting the biometrics. The server first extracts a biometric template of the queried biometrics through the feature extractor and finds the identifier in the server's database whose corresponding template is closest to the extracted one. If the distance between these templates is under the predetermined threshold, then the server returns the found identifier; otherwise, the server returns ``reject".

\subsection{Homomorphic Encryption}

HE is a cryptographic tool that enables executing arithmetic operations on an encrypted domain. Although there exist HE schemes supporting arbitrary arithmetic, called fully HE (FHE)~\cite{fan2012somewhat, brakerski2014leveled,cheon2017homomorphic}, an additive homomorphic property suffices for our purpose. 
Our proposed method is not limited to a specific HE scheme; it supports both addition-only homomorphic encryption (AHE) schemes, such as the Paillier cryptosystem~\cite{paillier1999public}, and FHE schemes.

Many FHE schemes support SIMD additions for each message slot, which is denoted by $\oplus$. For any plaintexts $\mathbf{x}$ and $\mathbf{y}$ of dimension $N$, the following equality holds:
\begin{align}
\mathsf{DEC}_{\mathsf{sk}}(\mathsf{ENC}_{\mathsf{pk}}(\mathbf{x+y})) = \mathsf{DEC}_{\mathsf{sk}}(\mathsf{ENC}_{\mathsf{pk}}(\mathbf{x}) \oplus \mathsf{ENC}_{\mathsf{pk}}(\mathbf{y})), \nonumber
\end{align}
where $\mathsf{ENC}_{\mathsf{pk}}$ and $\mathsf{DEC}_{\mathsf{sk}}$ are encryption and decryption algorithms of the FHE scheme, respectively.

Naturally, additions can be extended to scalar multiplications. For the sake of simplicity, we use a special notation $\odot$ for SIMD scalar multiplications on an encrypted domain satisfying the following relation for any scalar $c$:
\begin{align}
    \mathsf{DEC}_{\mathsf{sk}}(c \odot \mathsf{ENC}_{\mathsf{pk}}(\mathbf{x}))
    = \mathsf{DEC}_{\mathsf{sk}}(\mathsf{ENC}_{\mathsf{pk}}(c \cdot \mathbf{x})) \nonumber.
\end{align}

\section{BTP from Database Encryption Scheme}
Our approach for protecting face templates is basically to encrypt the database by AHE. We first formalize the notion of \textit{database encryption} and give a generic construction of the biometric identification scheme from it.

\subsection{Database Encryption Scheme}

We begin with introducing the concept of \textit{database encryption scheme}, whose required functionalities are to encrypt a bulk of templates simultaneously and calculate the matching score between the queried template and each enrolled template. 
The formal definition is given as follows:
\begin{definition}[Database Encryption Scheme]
    The database encryption scheme is a pair of algorithms ($\mathsf{ENC}_{\mathtt{DB}}$, $\mathsf{IP}_{\mathtt{DB}}$), each of whose functionality is as follows:
    \begin{itemize}
        \item $\mathsf{ENC}_{\mathtt{DB}}$ takes a bulk of biometric templates $\mathbf{x}_{1},\dots,\mathbf{x}_{N}$ and a public key $\mathsf{pk}$, returning a ciphertext $\mathcal{C}$.
        \item $\mathsf{IP}_{\mathtt{DB}}$ takes a biometric template $\mathbf{y}$ and a ciphertext $\mathcal{C}$, returning a ciphertext of $\langle \mathbf{x}_{i}, \mathbf{y} \rangle$ for all plaintexts $\mathbf{x}_{i}$ of $\mathcal{C}$.
    \end{itemize}    
\end{definition}

We can view the construction of HERS~\cite{engelsma2022hers}, the previous state-of-the-art, as a database encryption scheme.
The key idea of HERS was to view matrix-vector multiplication as a weighted sum of column vectors.
If we denote $\mathbf{X}\in \mathbb{R}^{N \times d}$ as a matrix of which the row vectors are enrolled templates and $\mathbf{y} = (y_{1}, \dots, y_{d}) \in \mathbb{R}^{d}$ as a vector requested to identify, the matching score between enrolled templates and $\mathbf{y}$ can be computed by a matrix-vector multiplication $\mathbf{X} \cdot \mathbf{y}$. If we denote $\mathbf{c}_{i} \in \mathbb{R}^{N}$ as the $i$'th column vector of $\mathbf{X}$ for $i\in[d]$, then we have that $\mathbf{X} \cdot \mathbf{y} = \sum_{i=1}^{d} y_{i} \cdot \mathbf{c}_{i}$. Hence, by encrypting each $\mathbf{c}_{i}$ during enrollment, we can compute the matching scores for the queried vector $\mathbf{y}_{i}$ in an encrypted domain via additions and scalar multiplications only.

Given an AHE scheme, we present a generic description of HERS as a database encryption $(\mathsf{ENC}_{\mathtt{DB}}, \mathsf{IP}_{\mathtt{DB}})$.
$\mathsf{ENC}_{\mathtt{DB}}$ encrypts each $N$-dimensional column vector of $\mathbf{X}$ to an AHE-ciphertext $\mathbf{ct}_{i}$, and a final ciphertext is $\mathcal{C}=(\mathbf{ct}_{1},\ldots,\mathbf{ct}_{d})$.
For given $\mathbf{ct}_{i}$'s and $\mathbf{y}=(y_1,\ldots,y_d)$, $\mathsf{IP}_{\mathtt{DB}}$ computes a ciphertexts of $(\langle\mathbf{x}_i,\mathbf{y}\rangle)_{i=1}^N$ by homomorphically computing a weighted sum of $\mathbf{ct}_{i}$'s with weights $y_{i}$'s. 
The descriptions of both algorithms are given in Algorithm~\ref{alg1:enrollelt} and \ref{alg2:identifyelt}, respectively. 
We will consider $\mathsf{ENC}_\mathtt{DB}$ and $\mathsf{IP}_\mathtt{DB}$ as the baseline, and utilize them as subroutines in our method.

\begin{figure}[t]
\vspace{-4mm}
\begin{minipage}[t]{\linewidth}
\begin{algorithm}[H]
\caption{$\mathsf{ENC}_\mathtt{DB}$ (\textbf{Base})}\label{alg1:enrollelt}
\begin{algorithmic}[1]
\REQUIRE $\mathbf{X}=[\mathbf{x}_{1},\dots,\mathbf{x}_{N}]^{T} \in \mathbb{R}^{N \times d}$ and $\mathsf{pk}$
\STATE Parse $\mathbf{X}$ as $d$ column vectors ${[ \mathbf{v}_{1},\dots,\mathbf{v}_{d} ]}$
\STATE Compute $ \mathbf{ct}_{i} \gets \mathsf{ENC}_{\mathsf{pk}}(\mathbf{v}_{i}) $ for $\forall i\in[d]$
\RETURN $\mathcal{C}=( \mathbf{ct}_{1},\dots,\mathbf{ct}_{d} )$
\end{algorithmic}
\end{algorithm}
\vspace{-8mm}
\end{minipage}
\begin{minipage}[t]{\linewidth}
\begin{algorithm}[H]
\caption{$\mathsf{IP}_\mathtt{DB}$ (\textbf{Base})}\label{alg2:identifyelt}
\begin{algorithmic}[1]
\REQUIRE $\mathbf{y}=(y_{1},\dots,y_{d})$ and a ciphertext $\mathcal{C}$
\STATE Parse $\mathcal{C}$ as $( \mathbf{ct}_{1},\dots,\mathbf{ct}_{d} )$
\STATE $\mathbf{ct}_{i} = y_{i} \odot \mathbf{ct}_{i}  $ for $\forall i\in[d]$
\RETURN $\mathbf{ct}_{1}\oplus \cdots \oplus\mathbf{ct}_{d} $
\end{algorithmic}
\end{algorithm}
\end{minipage}
\vspace{-5mm}
\end{figure}

\subsection{Secure Biometric Identification System from Database Encryption Scheme}\label{sec:3_2}

We present a generic construction of a biometric identification system with template protection by a database encryption scheme. 
We first note that, to instantiate the database encryption scheme with an AHE, the secret key of the AHE should be kept secure. For this reason, we consider two servers for managing the encrypted database and the secret key separately: the \emph{local server} ($\mathcal{S}_{local}$) that stores the public key $\mathsf{pk}$ and the encrypted database, and the \emph{key server} ($\mathcal{S}_{key}$) that only stores the secret key $\mathsf{sk}$. In this setting, $\mathcal{S}_{key}$ finds the identity of the largest score or returns \textit{reject}, and $\mathcal{S}_{local}$ receives templates for enrollment and identification. 

With the database encryption scheme ($\mathsf{ENC}_{\mathtt{DB}}$, $\mathsf{IP}_{\mathtt{DB}}$), $\mathcal{S}_{local}$ can encrypt a bunch of templates from enrollment by $\mathsf{ENC}_{\mathtt{DB}}$.
In addition, by using $\mathsf{IP}_{\mathtt{DB}}$, $\mathcal{S}_{local}$ can compute matching scores between the queried template and stored ones in an encrypted domain. The identifier of the largest matching score can be found by $\mathcal{S}_{key}$ after decryption. From these idea, we construct the biometric identification system with two servers, $\mathcal{S}_{local}$ and $\mathcal{S}_{key}$, as follows.

In the enrollment, $\mathcal{S}_{local}$ is requested to enroll a set of biometric templates. It first separates the set of biometric templates into subsets of size $N$. For each subset $\mathbf{X}$ with the corresponding identities $\mathtt{ID}$, $\mathcal{S}_{local}$ runs $\mathsf{ENC}_{\mathtt{DB}}(\mathbf{X}, \mathsf{pk})\rightarrow\mathcal{C}$, and adds $(\mathtt{ID}, \mathcal{C})$ into the bottom row of the database $\db$.

In the identification, $\mathcal{S}_{local}$ receives a biometric template $\mathbf{y}$ and returns the corresponding identity, if there exists an enrolled template $\mathbf{x}$ such that $\langle\mathbf{x},\mathbf{y}\rangle>\tau$ for a threshold $\tau$. The precise process is as follows: assume that $\db$ consists of $\{(\mathtt{ID}_i,\mathcal{C}_{i})\}_{i=1}^D$. For each $i\in[D]$, $\mathcal{S}_{local}$ runs $\mathsf{IP}_{\mathtt{DB}}(\mathbf{y}, \mathcal{C}_{i})\rightarrow \widehat{\mathcal{C}_{i}}$, whose output is the ciphertext containing inner products between $\mathbf{y}$ and plaintexts $\mathbf{x}$ of $\mathcal{C}_{i}$. 
Now, $\mathcal{S}_{key}$ obtains the decryption of $\{\widehat{\mathcal{C}_{i}}\}_{i=1}^D$ received by $\mathcal{S}_{local}$, which contains every matching score between the query and all enrolled identities in $\bigcup_{i=1}^D\mathtt{ID}_i$. Finally, $\mathcal{S}_{key}$ finds the index of the maximum matching score $\tau'$ and sends it to $\mathcal{S}_{local}$ if and only if $\tau'>\tau$, and $\mathcal{S}_{local}$ returns the corresponding $id$.

Later, we will follow this approach to construct the proposed BTP, $\mathsf{IDFace}$, with an improved database encryption scheme from our novel techniques.

\section{Proposed Method: \textsf{IDFace}}\label{sec:4}

We now present two main ingredients of $\mathsf{IDFace}$: almost isometric transformation that drastically improves computational efficiency with a small accuracy degradation, and space-efficient encoding for a bulk of templates that minimizes waste of plaintext space of HE. Using them, we construct our database encryption scheme and the proposed secure face identification system $\mathsf{IDFace}$.

\subsection{Almost Isometric Transformation}\label{sec:4_1}

Face templates usually lie in the hypersphere $\mathbb{S}^{d-1}$ and the predominant operation in identification is taking the inner product between two unit vectors. We observed that an inner product value is relatively more influenced by components with large absolute values than those with small such values. Our main strategy is to transform the biometric template into a vector in $\{-1, 0, 1\}^d$ according to the magnitude of each component in the template\footnote{In fact, $\{-1, 0, 1\}^d$ is not a subspace of $\mathbb{S}^{d-1}$, so final normalization is required. Without loss of generality, we omit this because the number of nonzero components is fixed after the transformation.}. 
More precisely, given a template, the transformation $T_{\alpha}$ with a parameter $\alpha$ first takes the top $\alpha$ components with larger absolute values. Then, those components are replaced with either $1$ or $-1$ according to their signs, and the other components are set to $0$. That is, the range of the transformation is a subset of $\{-1,0,1\}^d$, the vectors of which have only $\alpha$ nonzero components. We denote this set by $\mathcal{Z}^{d}_{\alpha}$.
With our transformation, taking an inner product can be done by a lookup operation, with addition and subtraction determined by the sign of each component. In particular, when combined with $\mathsf{IP_{DB}}$, we can replace homomorphic scalar multiplications with look-ups, hence requiring much lower computational cost than the ordinary method.

\vspace{2mm}

\noindent \textbf{Analysis on the Transformation.}
% \noindent \textbf{}
%
Although the computational advantage of our transformation is obvious, analyzing the extent to which it can preserve the distance relationship is not straightforward. To clarify this, we define a new notion of \emph{almost isometry}, or $(\epsilon, \delta, \theta)$-isometry that catches the change of inner product values between two unit vectors of angle $\theta$ from the transformation. The formal definition of our almost isometry is provided in Definition~\ref{def:alm_iso}.

\begin{definition}[$(\epsilon, \delta, \theta)$-isometry]\label{def:alm_iso}
    Let $T:\mathbb{S}^{d-1} \rightarrow \mathbb{S}^{d-1}$ be a function. Then for $\epsilon > 0, \theta \in (0, \pi)$, and $\delta \in [0,1]$, $T$ is called $(\epsilon, \delta, \theta)$-isometry if for $\mathbf{x}, \mathbf{x}' \in \mathbb{S}^{d-1}$ uniformly sampled from the set $\{(\mathbf{x}, \mathbf{x}') \in \mathbb{S}^{d-1} \times \mathbb{S}^{d-1} : \langle \mathbf{x}, \mathbf{x}' \rangle = \cos\theta\}$, $\mathrm{Pr}[|\langle T(\mathbf{x}), T(\mathbf{x}') \rangle - \langle \mathbf{x}, \mathbf{x}' \rangle| < \epsilon] > 1 - \delta$ holds.
\end{definition}
Here, the restriction $\langle \mathbf{x}, \mathbf{x'} \rangle = \cos\theta$ for sampling $\mathbf{x}$ and $\mathbf{x'}$ is necessary to ensure that the function preserves the inner product value even for two close vectors; note that the inner product value of two vectors uniformly sampled from $\mathbb{S}^{d-1}$ is close to 0 with an overwhelming probability. We also introduce a failure probability $\delta$ for relaxation.

We proved that the proposed transformation $T_{\alpha}$ satisfies the condition of $(\epsilon, \delta, \theta)$-isometry. However, due to space constraints, we only provide the result on the parameter ($\alpha = 341$) that minimizes $\epsilon$ for $d=512$. The full statement, its proof, and further analyses on the choice of the transformation parameter are given in Appendix~\textcolor{iccvblue}{B} and \textcolor{iccvblue}{C}.
\begin{proposition}[Informal]\label{prop1_inform}
For $d = 512$ and $\alpha = 341$, $T_{\alpha}$ is $(0.111, o(1), \theta)$-isometry, $\forall$ $\theta \in (0, \pi)$.
\end{proposition}

% \paragraph{Remarks.}
\noindent\textbf{Remarks.}
Note that $T_{\alpha}$ is technically the same as the decoding algorithm of a real-valued error-correcting code by \cite{kim2021ironmask}. However, it plays a completely different role in our paper. The authors of~\cite{kim2021ironmask} tried to facilitate the collision on $\mathcal{Z}_{d}^{\alpha}$ after the transformation, whereas we utilize it as an almost isometry. The almost-isometric property has neither been used nor even identified in their work.

We also note that the proposed transformation can be viewed as a ternary quantization on unit vectors. Several ternary quantization methods have been proposed for faster matrix multiplications on the inference of the neural network~\cite{zhu2017trained, li2021trq, ma2024era, hubara2018quantized}. Since a matrix multiplication can be viewed as a batched version of inner products, one can consider replacing our transformation to one of them.
Nevertheless, we found that there are some desirable properties for adopting them to our $\mathsf{IDFace}$.
Due to space constraints, we provide discussions about them and our rationale for choosing the proposed transformation in Appendix~\textcolor{iccvblue}{E}.

\subsection{Space Efficient Encoding}\label{sec:4_2}

The message space of AHEs is determined by the security parameter, regardless of practical applications. When employing them for face template protection, we found that a huge waste on the message space occurs if we do not carefully handle this issue.
For example, Paillier cryptosystem~\cite{paillier1999public} uses a 2048-bit message size for 112-bit security, and a typical parameter setting of CKKS~\cite{cheon2017homomorphic} preserves a certain level of precision, \emph{e.g.}. 50-bit, with 4096 message slots in 128-bit security. On the other hand, 16- or 32-bit precision is enough for representing face templates with an appropriate level of accuracy. Such a waste on the message space becomes worse after our transformation because the transformed template consists of only $-1, 0$ or $1$.

We tackle this issue by further packing multiple templates into a single message slot. We first observe that the inner product value of two transformed templates by $T_{\alpha}$ ranges between $-\alpha$ and $\alpha$ since each transformed template has $\alpha$ nonzero components of $\pm 1$. Handling both negative and positive values is rather cumbersome, and thus we separate a transformed template $\mathbf{x}$ into two positive binary vectors $\mathbf{x}^{+}$ and $\mathbf{x}^{-}$ such that $\mathbf{x}=\mathbf{x}^{+}-\mathbf{x}^{-}$. Then, an inner product value of two transformed templates $\mathbf{x}$ and $\mathbf{y}$ is a sum of inner products between separations: 
\begin{align}
\langle \mathbf{x}, \mathbf{y} \rangle = \langle \mathbf{x^{+}}, \mathbf{y^{+}} \rangle + \langle \mathbf{x^{-}}, \mathbf{y^{-}} \rangle - \langle\mathbf{x^{+}}, \mathbf{y^{-}}\rangle - \langle\mathbf{x^{-}}, \mathbf{y^{+}}\rangle. \nonumber
\end{align}
Using separations instead of original templates, all the inner product values in the above lie between $0$ and $\alpha$. 
Now, we can encode $m$ transformed templates $\mathbf{x}_{1}^{\dagger}, \dots, \mathbf{x}_{m}^{\dagger}$ after separation into a single vector $\mathbf{x}^{\dagger}:=\sum_{i=1}^{m}p^{i-1} \cdot \mathbf{x}_{i}^{\dagger}$ for $p > \alpha$ and $\dagger \in \{+, -\}$. Each $\mathbf{x}_{i}^{\dagger}$ can be decoded from $\mathbf{x}^{\dagger}$ by inductively calculating $\mathbf{x}_{1}^{\dagger} = \mathbf{x}^{\dagger}\pmod{p}$ and $\mathbf{x}_{i}^{\dagger}=(\mathbf{x}^{\dagger}-\sum_{j=1}^{i-1}(p^{j-1}\cdot \mathbf{x}_{j}^{\dagger}))\cdot p^{1-i}\pmod{p}$ for $i=2,\ldots,m$. If the message space of a message slot of AHE is $n$-bit, then we can encode at most $\lfloor \frac{n}{\log_{2} \alpha} \rfloor$ templates at once. We call these two operations $\mathsf{Encode}$ and $\mathsf{Decode}$, respectively.

We now show that $\mathsf{Encode}$ is compatible with $\mathsf{IP}_{\mathtt{DB}}$. Suppose that $m$ transformed templates $\mathbf{x}_{1},\dots,\mathbf{x}_{m}$ are split into ($\mathbf{x}^+,\mathbf{x}^-$) and each component is encrypted by $\mathsf{ENC}_{\mathtt{DB}}$, say $\mathcal{C}^{+}$ and $\mathcal{C}^{-}$. If a new template $\mathbf{y}$ is queried, we first split it by two vectors $\mathbf{y}^+$ and $\mathbf{y}^-$ and then run $\mathsf{IP}_{\mathtt{DB}}(\mathbf{y}^*,\mathcal{C}^\dagger)$ obtaining ciphertexts of $\langle\mathbf{x^{\dagger}}, \mathbf{y^{*}}\rangle$ for $\dagger,*\in\{+,-\}$. 
Since we appropriately select the bound $p$, the $\mathsf{Decode}$ with input $\langle\mathbf{x^{\dagger}}, \mathbf{y^{*}}\rangle$ returns $\langle\mathbf{x}_i^\dagger,\mathbf{y}^*\rangle$ for all $i\in[m]$. Finally, we can calculate the inner product values $\langle\mathbf{x}_i,\mathbf{y}\rangle$ for all $i\in[m]$.

Note that $\mathsf{Encode}$ is compatible with only one of the additions and subtractions. Thus, instead of directly computing $\langle\mathbf{x},\mathbf{y}\rangle$, we only compute $\langle\mathbf{x}^\dagger,\mathbf{y}^*\rangle$ homomorphically for each $\dagger,*\in\{+,-\}$.
We can reduce the number of decryptions by performing homomorphic additions to obtain two ciphertexts of $\langle \mathbf{x}^+, \mathbf{y}^+ \rangle+\langle \mathbf{x}^-, \mathbf{y}^- \rangle$ and $\langle \mathbf{x}^+, \mathbf{y}^- \rangle+\langle \mathbf{x}^-, \mathbf{y}^+ \rangle$, subtracting one from another after decoding.

\subsection{Putting them All Together}\label{sub:putting}
With almost-isometric transformation and space-efficient encoding, we can improve the baseline database encryption scheme $(\mathsf{ENC}_{\mathtt{DB}}$, $\mathsf{IP}_{\mathtt{DB}})$. We denote each improved algorithm as $\mathsf{IDFace.ENC}_{\mathtt{DB}}$ and $\mathsf{IDFace.IP}_{\mathtt{DB}}$, respectively. We assume that the number of face templates is a multiple of $mN$, where both $m$ and $N$ are parameters selected by the underlying AHE scheme. These $mN$ templates are represented by rows of a matrix $\mathbf{X} \in \mathbb{R}^{mN \times d}$. To process more templates, we can repeat the same procedure.

$\mathsf{IDFace.ENC}_{\mathtt{DB}}$ first applies our almost isometric transform $T_{\alpha}$ into each row vector of $\mathbf{X}$ and then encodes each $m$ rows of $\mathbf{X}$ into two vectors $\mathbf{x}_{i}^{+}$ and $\mathbf{x}_{i}^{-}$ by our $\mathsf{Encode}$ algorithm for $i \in [N]$. Finally, it runs $\mathsf{ENC}_{\mathtt{DB}}$ for each $\mathbf{x}_{i}^{+}$ and $\mathbf{x}_{i}^{-}$ and returns a pair $\mathcal{C} = (\mathcal{C}^{+}, \mathcal{C}^{-})$.

$\mathsf{IDFace.IP}_{\mathtt{DB}}$ takes a face template query $\mathbf{y} \in \mathbb{R}^{d}$ as input, and then applies $T_{\beta}$\footnote{We can use the different parameter $\beta$ from $\alpha$ that is used for enrollment in order to take advantage of the trade-off between accuracy and efficiency. Detailed explanation and analysis will be provided in Section~\ref{sec:5}.} to obtain a transformed template $\mathbf{z}$.
Then, it splits $\mathbf{z}$ into positive part $\mathbf{z}^{+}$ and negative part $\mathbf{z}^{-}$. We parse the given $\mathcal{C}$ into ($\mathcal{C}^{+}$, $\mathcal{C}^{-}$) and compute $\mathbf{ct}^{+} = \mathsf{IP}_{\mathtt{DB}}(\mathbf{z}^{+}, \mathcal{C}^{+}) \oplus \mathsf{IP}_{\mathtt{DB}}(\mathbf{z}^{-}, \mathcal{C}^{-})$ and $\mathbf{ct}^{-} = \mathsf{IP}_{\mathtt{DB}}(\mathbf{z}^{+}, \mathcal{C}^{-}) \oplus \mathsf{IP}_{\mathtt{DB}}(\mathbf{z}^{-}, \mathcal{C}^{+})$ by considering their sign. The subtraction $\mathbf{ct}^{+} - \mathbf{ct}^{-}$ after decryption gives the desired inner product value with the encoding, and the final result can be obtained by running $\mathsf{Decode}$.

Due to space constraints, we provide the full description of $\mathsf{IDFace}$ and its subroutine algorithms in Appendix~\textcolor{iccvblue}{F}.

\section{Analyses on \textsf{IDFace}}\label{sec:5}
In this section, we provide several analyses of the proposed $\mathsf{IDFace}$, along with the implementation result. We used an i7-11700K CPU, which has 3.60GHz processor speed with 8 cores and 16 threads.
%
% For feature extractor, we used ArcFace-ResNet100~\cite{deng2019arcface} trained on the MS1MV3 dataset~\cite{deng2020retinaface}, by following their official PyTorch implementation. We used IJB-C~\cite{maze2018iarpa} and other famous face verification benchmarks, LFW \cite{huang2008labeled}, CFP-FP \cite{sengupta2016frontal}, and AgeDB \cite{moschoglou2017agedb} for evaluating the accuracy.
%
% , which were also used in other recent face recognition system proposals \cite{deng2019arcface, meng2021magface, wen2022sphereface2, boutros2022elasticface}.
%
For AHEs in $\mathsf{IDFace}$, we use two schemes: Paillier cryptosystem (PC) and CKKS. Since CKKS also satisfies the multiplicative homomorphic property, previous HE-based methods can be directly compared. For parameters of each AHE, we select a 2048-bit plaintext modulus for PC, and a 61-bit plaintext modulus, a 220-bit ciphertext modulus with 4096 message slots for CKKS. Each setting satisfies 112-bit and 128-bit security levels, respectively. We experimentally verified that CKKS preserves almost 50-bit precision in this setting, and thus we used only 50-bit for our $\mathsf{Encode}$ when using CKKS.
We implemented our $\mathsf{IDFace}$ by exploiting IPCL~\cite{ICPL} library for PC and HEaaN~\cite{HEAAN} library for CKKS. Each library is written in modern standard C++ and provides a Python wrapper.

\subsection{Efficiency Analysis}\label{sec:5_1}

% The efficiency gain of $\mathsf{IDFace}$ can be summarized as two parts: Removing $d$ multiplications by using our almost isometric transformation and dealing multiple biometric templates at the same time using our $\mathsf{Encode}$ algorithm. 

We first provide a concrete efficiency analysis on $\mathsf{IDFace.IP}_{\mathtt{DB}}$ and introduce a further acceleration technique via the trade-off between efficiency and accuracy. 
Recall that our transformation $T_{\beta}$ in $\mathsf{IDFace.IP}_{\mathtt{DB}}$ converts the given biometric template into a signed binary vector in $\mathcal{Z}^{d}_{\beta}$.
Since the number of required additions in $\mathsf{IP}_{\mathtt{DB}}$ for binary vectors is the same as the number of nonzero components, $2(\beta - 1)$ ciphertext addition suffices during four iterations of execution of $\mathsf{IP}_{\mathtt{DB}}$.
The Tab.~\ref{table1:comparison} provides the comparison on the cost for identification on the database of size $mN$ between a na\"ive method, some previous works~\cite{boddeti2018secure, engelsma2022hers, bassit2023improved, bassit2022multiplication} tried to reduce the cost for inner product, and our $\mathsf{IDFace}$.

% Some previous works tried to reduce the computational cost for computing inner product on the encrypted domain. For a comparison between theirs and ours, we calculated the cost for identifying one template in each protocol, including a naive method for computing the inner product. The result is provided in Tab.~\ref{table1:comparison}.

\begin{table}[t]%{.55\linewidth}
% \vspace{-3mm}
\centering
\resizebox{\linewidth}{!}{
\begin{tabular}{c|c|c|c|c|c}
\hline
\multicolumn{1}{c|}{Method} & Add & Mult & Rotate & DEC & \# $\mathsf{Ct}$ \\ \hline
Na\"ive & $mN(d-1)$ & $mN$ & $mN(d-1)$ & $mN$ & $mN$ \\
\cite{boddeti2018secure} & $mN\lceil\log_{2}(d) \rceil$ & $mN$ & $mN\lceil\log_{2}(d)\rceil$ & $mN$ & $mN$ \\ 
HERS~\cite{engelsma2022hers} & $m(d-1)$ & $md$ & $0$ & $m$ & $m$ \\ 
MFBR-v2~\cite{bassit2023improved} & $mN\lceil \log_{2}(N) \rceil$ & $0$ & $mN \lceil \log_{2}(N) \rceil$ & $mN$ & $mN$ \\ 
MFBR-ID~\cite{bassit2025practical} & $m(d-1)$ & $0$ & $0$ & $m$ & $8m$ \\ \hline
$\mathsf{IDFace}$ & $2(\beta-1)$ & $0$ & $0$ & $2$ & $2$ \\ \hline
\end{tabular}}
\vspace{-3mm}
\caption{Comparison of the number of operations for identification and storage. $m$: the number of transformed templates in each message slot of AHE. $mN$: the number of enrolled templates, which are $d$-dimensional vectors. $\# \mathsf{Ct}$: the number of ciphertexts.}\label{table1:comparison}
\vspace{-4mm}
\end{table}

Note that the number $m$ of templates handled by $\mathsf{Encode}$ is determined by the maximum value of inner product, which is $\min \{\alpha, \beta \}$. Hence, choosing smaller $\beta$ while fixing $\alpha$ further accelerates $\mathsf{IDFace.IP}_{\mathtt{DB}}$.
Of course, choosing smaller $\beta$ leads to inaccurate inner product values.
Since the effect of $\beta$ on $m$ depends on the length of $\beta$ as a binary representation, we select $\beta = 63$ (resp. $\beta=127$ or $\beta=341$) while fixing $\alpha$ to $341$. In this case, we can encode many templates at once up to 342 (resp. 293 or 243) in PC and 8 (resp. 7 or 5) in CKKS.

\vspace{2mm}

\noindent\textbf{Implementation Results.}
We now verify the efficiency of our $\mathsf{IDFace}$ in our experimental environment. We note that every experiment in this subsection was done by generating random vectors.
We first evaluate the time spent on the enrollment and identification for various numbers of enrolled identities. The number of identities varies from $2^{13}$ to $2^{20} \approx $1 M, and we select $\beta = 63, 127, 341$ while $\alpha$ is fixed to $341$.
The results for each algorithm are illustrated in Fig.~\ref{fig:enr_and_idf}.
Unlike PC, $\mathsf{IDFace}$ from CKKS can processes more identities at the same time because of the coefficient packing. Concretely, $\mathsf{IDFace}$ in our setting can process 20480, 28672, 32768 identities at once for $\beta = 341, 127, 63$, respectively. This is why the elapsed time for the number of identity $2^{13}$--$2^{15}$ in CKKS seems unchanged.
In addition, the elapsed time in all parameter settings grows linearly with respect to the number of enrolled identities, because $\mathsf{IDFace}$ reduces the cost of inner product rather than that of searching. Nevertheless, in all parameter settings, the $\mathsf{IDFace}$ with CKKS can perform an identification in less than a second on the database of $2^{20} \approx $ 1 M identities.

\begin{figure}[t]
    \vspace{-2mm}
    \centering
    \includegraphics[width = \linewidth]{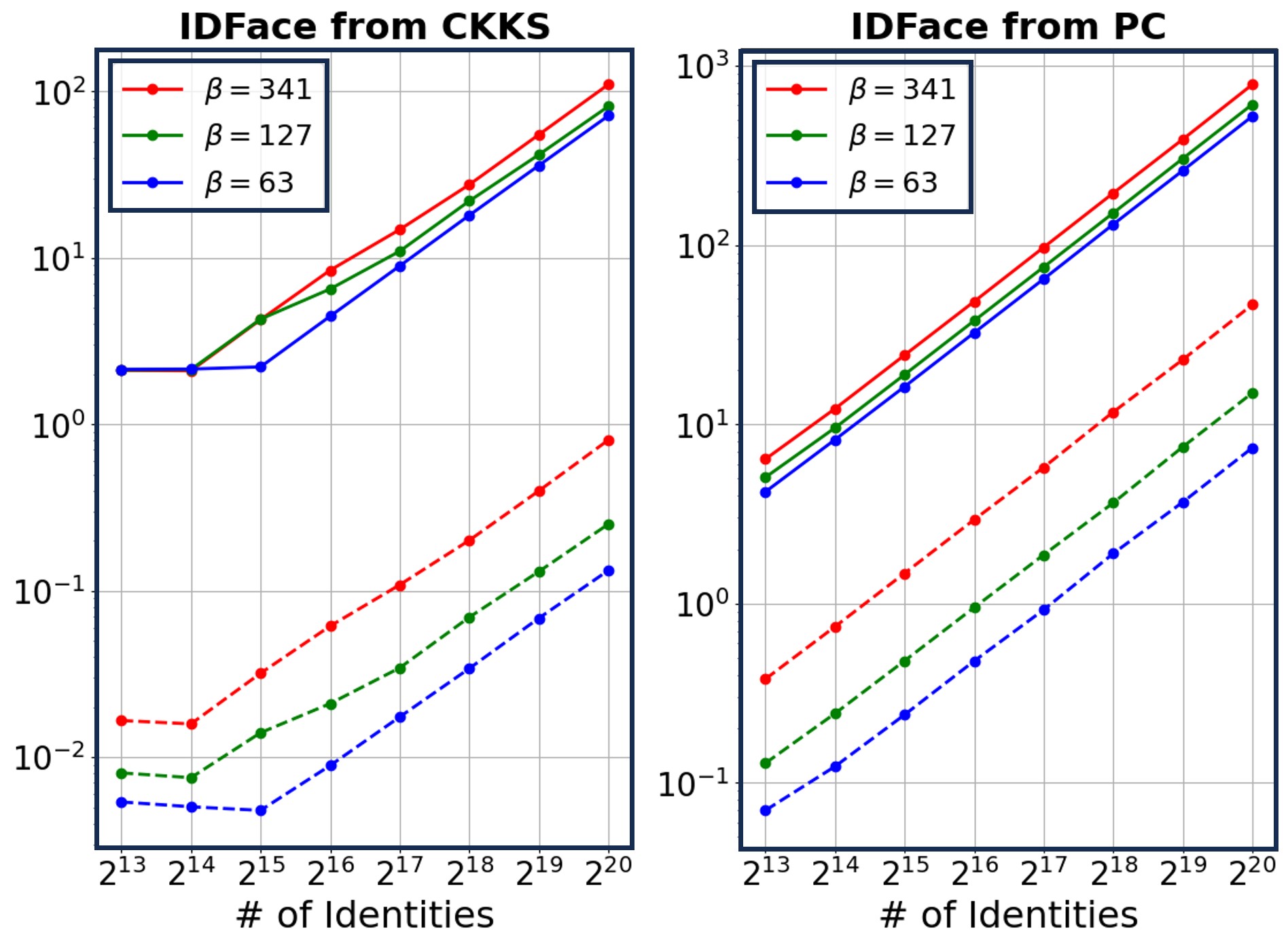}
    \vspace{-7mm}
    \caption{Elapsed time for $\mathsf{Enroll}$ and $\mathsf{Identify}$ on various numbers of enrolled identities. Solid line: elapsed time for $\mathsf{Enroll}$. Dashed line: elapsed time for $\mathsf{Identify}$.}
    \label{fig:enr_and_idf}
    \vspace{-5mm}
\end{figure}

To compare our $\mathsf{IDFace}$ with other existing BTPs, we provide the computational cost of each method for enrollment and identification on a database of size 1 M. We re-implemented IronMask~\cite{kim2021ironmask}, SecureVector~\cite{meng2022towards}, \cite{drozdowski2019application, boddeti2018secure, Bausp2022improved}, BlindMatch~\cite{choi2024blind}, HERS~\cite{engelsma2022hers}, MFBR-v2~\cite{bassit2023improved} and MFBR-ID~\cite{bassit2025practical} in our experimental environment. We followed the pseudocodes provided in each original paper. Since \cite{boddeti2018secure, Bausp2022improved} and HERS exploited a dimensionality reduction technique for faster identification, we tested both cases when the dimension of the feature vector was 512 and 128\footnote{Although their original papers provided results for lower dimensions, such as 64 or 32, we omitted them because of a notable accuracy degradation. The amount of degradation will be further discussed in Appendix~\textcolor{iccvblue}{D.5.}}. 
For BlindMatch, we select the number of partitions of the templates ($N_{in}$ in their paper) to be 8 and 16 for 128 and 512-dimension, respectively.
Since MFBR-v2 and MFBR-ID are designed for BFV~\cite{fan2012somewhat}, we implemented them by OpenFHE v1.2.4~\cite{al2022openfhe} in C++17 with the same FHE parameters according to their papers.
Moreover, we measured the time spent on identification without BTP, \textit{i.e.}, one matrix-vector multiplication without encryption.

Remark that some methods were omitted because their main efficiency gain came from other previous works~\cite{jindal2020secure, huang2023efficient}, or severe accuracy degradation is expected for large-scale databases of size 1 M~\cite{osorio2022stable, drozdowski2021feature, bauspiess2023hebi}. In addition, we reported estimated results for some works~\cite{kim2021ironmask, boddeti2018secure, drozdowski2019application} by first evaluating them on the database of 1,000 identities and multiplying 1,000 on the result. This estimation would be reasonable because the computational cost of theirs is linear to the number of identities.
Finally, unlike other HE-based methods, IronMask does not require a secret key. Although comparing it with other methods, including ours, would be unfair, we added it as a reference for the efficiency analysis.

We also note that some prior works, including~\cite{boddeti2018secure, drozdowski2021feature, Bausp2022improved, choi2024blind} and HERS, support the case where the queried vectors are encrypted. 
Such a setting corresponds to the case where the privacy of the query vector against server(s) is crucial, e.g., outsourcing the data and matching process to 3rd party cloud servers~\cite{drozdowski2019application, engelsma2022hers, huang2023efficient}.
Nevertheless, when secure storage of templates is the primary goal, such as in our target scenarios, considering an encrypted query vector would be an overkill.
Therefore, for the sake of fair comparison, we implemented those schemes as if the query vector were not encrypted.

The results are given in Tab.~\ref{tab:comp_res}. We reported the average elapsed time with a 1-$\sigma$ error for 10 runs. 
Our $\mathsf{IDFace}$ significantly outperforms previous HE-based BTPs in terms of
Notably, compared to HERS, with the parameter $\beta=341$ (resp. $\beta=127$ or $\beta=63$) achieves $16.3\times$ (resp. $53.0\times$ or $97.6\times$) faster identification speed with only $60\%$ (resp. $20\%$ or $0\%$) extra storage.
In addition, compared to MFBR-ID~\cite{bassit2025practical}, which is another multiplication-free method based on look-up tables, $\mathsf{IDFace}$ with $\beta=63$ spent $23.5\times$ and $4.5\times$ less time for enrollment and identification, respectively, while requiring $31.8\times$ smaller storage for protected templates.
Furthermore, when comparing ours to a na\"ive identification without BTP, i.e., performing a matrix-vector multiplication between the enrolled and queried template(s), the overhead from both computational and storage costs of $\mathsf{IDFace}$ with $\beta = 63$ is only about $\mathbf{2\times}$.
These results suggest that our $\mathsf{IDFace}$ could be deployed in application scenarios where real-time identification is required, such as at an airport or the entrance of a building.

\begin{table}[t]
\centering
\resizebox{\linewidth}{!}{
    \begin{tabular}{|c|c|c|c|c|}
    \hline
     Method (Param.) & Enroll & Identify & Storage & Primitive \\
     \hline \hline
     Non-protected & N/A & $0.063 \pm 0.0005$s & 2GB & N/A \\
     \hline \hline
     IronMask~\cite{kim2021ironmask}$^{\dagger}$ & $4,416\pm 3.73$s & $97.9 \pm 0.27$s & 1024GB & FC \\
     SecureVector~\cite{meng2022towards} & $350 \pm 4.09$s & $398 \pm 2.75$s & 5GB  & PC \\
     \hline \hline
     \cite{drozdowski2019application}$^{\dagger, \ddagger}$ & $917\pm3.57$s & $290,605\pm458$s & 132GB & CKKS\\
     \cite{boddeti2018secure} (512)$^{\dagger, \ddagger}$ & $909\pm3.30$s & $6,210 \pm 16.5$s & 132GB  & CKKS\\
     \cite{boddeti2018secure} (128)$^{\dagger, \ddagger}$ & $846\pm1.58$s & $5,037 \pm 8.49$s & 33GB  & CKKS\\
     % \cite{jindal2020secure} & N/A & 2,830s & 32.8GB & CKKS \\
     \cite{Bausp2022improved} (512)$^{\ddagger}$ & $198\pm0.73$s & $826 \pm 1.14$s & 16.5GB & CKKS\\
     \cite{Bausp2022improved} (128)$^{\ddagger}$ & $50\pm0.05$s & $171 \pm 0.25$s & 4.125GB  & CKKS \\
     BlindMatch (512)$^{\ddagger}$ & $253\pm 13.9$s & $233\pm 10.4$s & 16.5GB & CKKS\\
     BlindMatch (128)$^{\ddagger}$ & $68.5\pm2.60$s & $66.6\pm1.79$s & 4.125GB  & CKKS \\     
     % \cite{osorio2022stable} (TBW) & N/A & 637s & 109GB & BFV \\ %Tobewritten
     % \cite{huang2023efficient}$^{\dagger}$ & 2,283s & 5,100s & 13.76GB & BGV \\ %single core
     HERS (512)$^{\ddagger}$ & $199\pm0.68$s & $48.9 \pm 0.20$s & 16.5GB  & CKKS \\
     % \cline{2-4}
     HERS (128)$^{\ddagger}$ & $\mathbf{49\pm0.05}$s & $12.3\pm0.04$s & 4.125GB & CKKS \\
     \hline \hline
     MFBR-v2$^\dagger$~\cite{bassit2023improved} & $1348\pm 2.32$s & $1026\pm8.04$s & 131.1GB & BFV \\
     MFBR-ID~\cite{bassit2025practical} & $1641\pm5.98$s & $0.545\pm0.006$s & 132.1GB & BFV \\
     \hline \hline
     $\mathsf{IDFace}$ (341) & $751\pm2.44$s & $43.4\pm0.38$s & 2.25GB & PC \\
     $\mathsf{IDFace}$ (127) & $582\pm0.37$s & $14.2\pm0.13$s & 1.75GB & PC \\
     $\mathsf{IDFace}$ (63) & $501\pm 1.73$s & $7.08\pm0.03$s & \textbf{1.5GB} & PC \\
     \hline
     $\mathsf{IDFace}$ (341) & $109 \pm 0.10$s & $0.753 \pm 0.005$s & 6.6GB & CKKS \\
     % \cline{1-4}
     $\mathsf{IDFace}$ (127) & $80 \pm 0.13$s & $0.232 \pm 0.002$s & 4.71GB & CKKS \\
     % \cline{1-4}
     $\mathsf{IDFace}$ (63) & $72 \pm 0.11$s & $\mathbf{0.126\pm0.001}$s & 4.125GB  & CKKS \\
     \hline
    \end{tabular}  
}
\vspace{-3mm}
\caption{Time and storage cost comparison of identification protocols with 1 M enrolled identities. FC: Fuzzy Commitment. PC: Paillier Cryptosystem. $\dagger$: Estimated from 1,000 identities. $\ddagger$: This scheme supports encrypted queries.
}\label{tab:comp_res}
\vspace{-4mm}
\end{table}

\begin{table*}[t]
\centering
\resizebox{.91\linewidth}{!}{
\begin{tabular}{c|c|c|c|c|c|c||c|c}
\hline
\multirow{2}{*}{Dataset} & \multirow{2}{*}{Metric (FAR)} & \multirow{2}{*}{Plain} & \multicolumn{6}{c}{$\mathsf{IDFace}$, Parameters: $(\alpha, \beta)$} \\ \cline{4-9}
 & & & $(512,512)$ & $(341,341)$ & $(127,127)$ & $(63,63)$ & $(341,127)$ & $(341,63)$ \\ \hline
LFW & \multirow{3}{*}{Accuracy} & 
99.82\% & 99.77\% & 99.78\% & \textcolor{red}{\textbf{99.83\%}} & 99.73\% & 99.80\% & 99.82\% \\ 
CFP-FP &  & 
\textbf{99.24\%} & 98.96\% & 99.24\% & 99.09\% & 98.20\% & 99.19\% & 98.99\% \\
AgeDB &  & 
\textbf{98.00\%} & 97.67\% & 97.97\% & 97.55\% & 96.48\% & 97.80\% & 97.23\% \\ \hline
IJB-C (V) & TAR (1e-3) & 
\textbf{98.39\%} & 97.93\% & 98.27\% & 97.90\% & 96.94\% & 98.14\% & 97.78\% \\ \hline
\multirow{3}{*}{IJB-C (ID)} & TPIR (1e-2) & 
\textbf{96.42\%} & 95.15\% & 95.95\% & 95.32\% & 92.87\% & 95.72\% & 95.19\% \\ 
 & TPIR (1e-3) & 
85.81\% & \textcolor{red}{85.83\%} & 85.11\% & 85.40\% & \textcolor{red}{86.21\%} & \textcolor{red}{\textbf{86.59\%}} & \textcolor{red}{85.96\%} \\
 & TPIR (1e-4) & 
65.60\% & 63.88\% & \textcolor{red}{\textbf{68.38\%}} & 62.38\% & 61.39\% & 63.45\% & 64.66\% \\ \hline
\end{tabular}
}
\vspace{-3mm}
\caption{Benchmark results of non-protected feature extractor (Plain) and $\mathsf{IDFace}$ for various $(\alpha, \beta)$.}\label{tab3}
\vspace{-4mm}
\end{table*}

\subsection{Accuracy Analysis}\label{sec:5_2}

We provide the benchmark result of $\mathsf{IDFace}$ on various datasets, including IJB-C~\cite{maze2018iarpa}, LFW~\cite{huang2008labeled}, CFP-FP~\cite{sengupta2016frontal}, and AgeDB~\cite{moschoglou2017agedb}. For the feature extractor, we used the pre-trained AdaFace~\cite{kim2022adaface} with the IResNet101 backbone trained on the WebFace12M dataset~\cite{zhu2021webface260m}, which is available at the CVLFace library~\cite{CVLface}. Throughout this section, we refer to a non-protected feature extractor as a plain model.
We evaluated the effect of our transformation parameters $(\alpha, \beta)$ on the accuracy of the FRS. In the first 4 columns of Tab.~\ref{tab3}, we provide benchmark results when $\alpha=\beta$ with $\alpha = 63, 127, 341$ and $512$. One notable result is that the result from $\alpha = 341$ is better than that of $\alpha = 512$ in most cases.
Recall that we considered the setting with $\beta < \alpha$ for speed-up while allowing extra accuracy degradation. So, we evaluated the effect of the choice $\beta = 63, 127$ on accuracy while $\alpha$ is fixed to 341. The results for them are provided in the rightmost two columns of Tab.~\ref{tab3}. 
We can figure out that $\mathsf{IDFace}$ shows less than 1\% performance loss for LFW, CFP-FP, and AgeDB even for the setting $\beta=63$.
For our main interest, identification on IJB-C, a trade-off between query speed and performance was evident. But, for $\beta = 341$, the loss for TAR was still less than 1\%.

Due to space constraints, we defer several detailed analyses results to Appendix. For example, in Appendix~\textcolor{iccvblue}{D.3}, we provide results from other recent feature extractors, including ArcFace~\cite{deng2019arcface}, MagFace~\cite{meng2021magface}, SphereFace2~\cite{wen2022sphereface2}, ElasticFace~\cite{boutros2022elasticface}, and AdaFace-KPRPE~\cite{kim2024keypoint}, showing that the proposed $\mathsf{IDFace}$ shows consistently small accuracy degradation as in Tab.~\ref{tab3}. In Appendix~\textcolor{iccvblue}{D.6}, we provide our analyses on how our transformation affects the accuracy, along with our interpretation of data marked with red in Tab.~\ref{tab3}.

\subsection{Security Analysis}

We analyze the security of $\mathsf{IDFace}$ according to standard security requirements for BTP~\cite{ISO24745}: \emph{irreversibility}, \emph{revocablity}, and \emph{unlinkability}. Informally, each of them is stated as follows: Irreversibility requires that extracting biometric information from a protected template be computationally infeasible. Revocability requires that the protected template can be renewed by the same biometric without security loss. To satisfy unlinkability, protected templates from biometrics of the same identity should be independent; useful information cannot be extracted by comparing them. 

For the threat model, we consider an attacker who can take the public key $\mathsf{pk}$ and the encrypted database stored in the local server, without knowing the secret key $\mathsf{sk}$ stored in the key server. This threat model reflects several practical situations, such as when the key server is a secure tamper-proof hardware that stores a secret key, or when multiple local servers exist as shown in Appendix~\textcolor{iccvblue}{H}. In both cases, the key server can become more secure with additional protections. The attacker’s goal is to extract biometric information from the encrypted database. Because our database is encrypted by AHE, this goal can be regarded as breaking the IND-CPA security of AHE. Note that the AHE scheme exploited in $\mathsf{IDFace}$ satisfies IND-CPA security. That is, even when the stored ciphertexts are revealed, they are computationally indistinguishable from random strings from the view of the adversary, so $\mathsf{IDFace}$ satisfies the irreversibility. In addition, because the encryption algorithms of each AHE are randomized, two ciphertexts from the same biometric source are computationally indistinguishable. Unlinkability clearly holds for the same reason as revocability. To sum up, our $\mathsf{IDFace}$ satisfies all these security requirements.

\vspace{-1mm}

\section{Conclusion and Further Directions}
Protecting and exploiting biometric templates simultaneously is a challenging problem because many cryptographic tools are not tailored for processing biometric templates in a real-valued vector format.
In this study, we propose two techniques to tackle this challenging problem, and then show the effectiveness of our techniques by presenting $\mathsf{IDFace}$, a plug-and-play HE-based face template protection method for efficient and secure large-scale identification. 
With $\mathsf{IDFace}$, we successfully identify a queried template over encrypted template databases of 1M identities in \textit{less than a second}, 
which is only $\mathbf{2\times}$ slower compared to the setting without protection.
We believe that $\mathsf{IDFace}$ enables reliable and secure applications of face identification systems in various circumstances.

Throughout this study, we leave several interesting questions.
First, our almost-isometric transformation is of independent interest. Designing another transformation with more accurate distance relationship preservation or applying it to other applications would be interesting future directions\footnote{For example, we designed a secure two-party computation-based variant of $\mathsf{IDFace}$ using our transformation in Appendix~\textcolor{iccvblue}{G}.}$^,$\footnote{We also provided the result of applying $\mathsf{IDFace}$ on speaker verification and fingerprint verification tasks in Appendix~\textcolor{iccvblue}{I}.}. In addition, we considered only passive attackers for setting the threat model of $\mathsf{IDFace}$. Extending $\mathsf{IDFace}$ to secure against active attackers is open, and investigating this would be an interesting future work.

\section*{Acknowledgements}

We would like to thank the anonymous ICCV 2025 reviewers for their valuable feedback.

{
    \small
    \bibliographystyle{ieeenat_fullname}
    \bibliography{main}
}

\clearpage
\onecolumn
\appendix
% ICCV 2025 Paper Template

% \documentclass[10pt,letterpaper]{article}

%%%%%%%%% PAPER TYPE  - PLEASE UPDATE FOR FINAL VERSION
% \usepackage{iccv}              % To produce the CAMERA-READY version
% \usepackage[review]{iccv}      % To produce the REVIEW version
% \usepackage[pagenumbers]{iccv} % To force page numbers, e.g. for an arXiv version

% Import additional packages in the preamble file, before hyperref
% \input{preamble}

% It is strongly recommended to use hyperref, especially for the review version.
% hyperref with option pagebackref eases the reviewers' job.
% Please disable hyperref *only* if you encounter grave issues, 
% e.g. with the file validation for the camera-ready version.
%
% If you comment hyperref and then uncomment it, you should delete *.aux before re-running LaTeX.
% (Or just hit 'q' on the first LaTeX run, let it finish, and you should be clear).
% \definecolor{iccvblue}{rgb}{0.21,0.49,0.74}
% \usepackage[pagebackref,breaklinks,colorlinks,allcolors=iccvblue]{hyperref}

% \newcommand{\Hquad}{\hspace{0.75em}} 

%%%%%%%%% PAPER ID  - PLEASE UPDATE
% \def\paperID{15047} % *** Enter the Paper ID here
% \def\confName{ICCV}
% \def\confYear{2025}

\section*{\centering Supplementary Material}% of IDFace: Face Template Protection for Efficient and Secure Identification}

%%%%%%%%% AUTHORS - PLEASE UPDATE
% \author{Sunpill Kim$^{1\dagger}$\Hquad Seunghun Paik$^{1\dagger}$\Hquad Chanwoo Hwang$^{1}$\Hquad Dongsoo Kim$^{1}$\Hquad Junbum Shin$^{2}$\Hquad Jae Hong Seo$^{1}$\thanks{Corresponding author}\\
% $^{1}$Department of Mathematics \& Research Institute for Natural Sciences, Hanyang University\\
% $^{2}$CryptoLab Inc.\\
% {\tt\small \{ksp0352, whitesoonguh, aa5568, frds37, jaehongseo\}@hanyang.ac.kr, junbum.shin@cryptolab.co.kr}
% }

% \begin{document}
% \maketitle
% \thispagestyle{empty} % ← 현재 페이지 번호 없애기
% \pagestyle{empty}     % ← 이후 페이지들도 번호 없애기
% \def\thefootnote{$\dagger$}\footnotetext{These co-first authors contributed equally to this work}\def\thefootnote{\arabic{footnote}}

% \appendix

% \vspace{-5mm}

In this supplementary material, we provide some details on both theoretical and empirical analyses that were omitted in the main paper. This document provides the following contents:

% \tableofcontents

% \newpage

\section{Additional Related Works}

We briefly survey other types of biometric template protection (BTP) methods than using homomorphic encryption (HE), which were omitted in the main text. 
We divide the taxonomy of BTPs with respects to whether cryptographic tools are used or not during construction. For BTPs with cryptographic tools, we will focus on fuzzy commitment (FC)~\cite{juels1999fuzzy} constructions, which is considered one of typical solutions for constructing BTPs.

\;

\noindent \textbf{BTPs without Cryptographic Tools.} In fact, several BTPs without using cryptographic primitives have been proposed. For example, one famous approach is to employ locality-sensitive hashing (LSH)~\cite{indyk1998approximate}, which is a (non-cryptographic) hash function that makes a collision on the hashed value if two input vectors are sufficiently close enough. Starting from~\cite{jin2004biohashing}, several LSH-based BTPs~\cite{jin2017ranking, dong2022deep, lai2021efficient, li2022indexing, jiang2023cancelable} have been proposed under the assumption that finding a pre-image of a LSH is computationally infeasible. However, their security proofs are less robust than BTPs whose security relies on cryptographic assumptions, as shown in several cryptanalyses~\cite{ghammam2020cryptanalysis, wang2021interpretable, paik2023security} that successfully recover a pre-image of the protected template.

On the other hand, there are several attempts to design BTPs through obfuscating the feature templates~\cite{mi2023privacy, mi2024privacy, jin2024faceobfuscator, wang2023privacy}. More precisely, the goal of these works is to modify the face templates so that the reconstruction attack adversary cannot train and infer the reconstruction model as desired while maintaining the accuracy of the recognition model through these templates. To this end, \cite{wang2023privacy} utilized the idea of adversarial attack on the shadow face reconstruction model owned by the database server to protect face templates. On the other hand, other works attempted to delete and recover visual information of the queried facial images on frequency domain~\cite{mi2023privacy, jin2024faceobfuscator} or using an auxiliary neural network~\cite{mi2024privacy}, in order to disturb the training of reconstruction attack models. Although all these works provided their own security analyses, they are all empirical; there are no theoretical or provable security guarantees.

\;

\noindent \textbf{FC-based BTPs.}
A FC scheme is a famous solution for constructing BTP using error-correcting code (ECC). In this approach, they view biometric templates as a fuzzy data, coping with the inherent noise of the biometrics via error-correction. Since the hashed values of the noise-corrected templates are stored to the database, the security of FC-based BTPs relies on the onewayness of the cryptographic hash function.

The accuracy degradation accompanied from the FC-based BTP depends on the noise-correction capacity of the underlying ECC. In addition, since many classical ECCs, such as BCH-codes or Reed-Solomon codes. were defined over binary or finite fields, they cannot be directly adopted for correcting noises in the real-valued face templates.
For this reason, recent studies elaborated on designing ECCs suitable for face templates. For example, some studies~\cite{talreja2019zero, jindal2019securing, kumar2018face, mai2020secureface, dong2024wifakey} designed neural network-based decoders for classical ECCs. \cite{mohan2019significant} proposed a novel representation of real-valued feature vectors tailored for classical ECCs. On the other hand, \cite{kim2021ironmask} proposed a novel ECC defined over the unit hypersphere, proposing a FC-based BTP called IronMask using this ECC.

Many FC-based BTPs were proposed for face verification tasks, showing a reasonable performance with respect to efficiency.
However, we note that they are not suitable for large-scale identification because of the difficulty of leveraging parallel computation, such as a batched computation of the matching score.

\;

\noindent \textbf{Privacy-Preserving Face Recognition.}
Other than BTPs, it is worthwhile to mention that there have been studies focused on information leakage during the communication in the matching process~\cite{erkin2009privacy, evans2011efficient, lei2021privface, im2020practical, huang2023efficient}. 
In their works, they regard each process of enrollment and identification as a protocol between the database server and the client. From this setting, they utilized several cryptographic tools on secure multi-party computation to construct secure two (or multi) party protocols. Nevertheless, we note that protecting templates is not a primary goal of theirs; in fact, the stored face templates in the database are unencrypted in some works~\cite{erkin2009privacy, evans2011efficient}. For this reason, we focused on BTP methods only throughout this paper.

\section{Full Statement of Proposition 1 and its Proof}\label{sec:supp_A}

In this section, we give our theoretical analysis of the proposed almost isometric transformation. For this, we first give the full statement of Proposition~\ref{prop_1_full} in the main text. Furthermore, we give a mathematical proof on the Proposition~\ref{prop_1_full}, which demonstrates that our proposed transformation is indeed an almost isometric transformation.

\subsection{Full Statement of Proposition 1}

The following proposition is the full statement of the Proposition~\ref{prop_1_full}. Note that the result given in the main text was derived by numerically calculating the following integrals when $(d, \alpha) = (512, 341)$. Throughout this section, $o(1)$ denotes a little-oh notation with respect to the dimension $d$.

\begin{proposition}\label{prop_1_full}
For $d\in \mathbb{N}\setminus\{1\}$ and $\theta \in (0, \pi/2) \cup (\pi/2, \pi)$,
\begin{align}
   T_{\alpha}: \mathbb{S}^{d-1} \rightarrow \mathcal{Z}_{\alpha}^{d} \textit{ is a } (\epsilon_{\alpha, \theta} + o(1), \delta_{\alpha}, \theta) \textit{-isometry } 
   \textit{for } \epsilon_{\alpha, \theta} = \left| \cos{\theta}-\frac{2P(\theta)\cdot d}{\alpha} \right|, \delta_{\alpha} = o(1) \nonumber
\end{align}
\noindent \textit{where},
\begin{align}
    P(\theta)=\int_{c}^{\infty}\int_{\frac{c}{\cos{\theta}}-u}^{\infty} f_{UV}(u,v) dv du
    - \int_{c}^{\infty}\int_{-\infty}^{-\frac{c}{\cos{\theta}}-u} f_{UV}(u,v) dv du,  \nonumber
\end{align}
$f_{UV}(u,v)$ is a joint pdf of bivariate normal distribution with two independent random variables $U\sim N(0,1)$ and $V\sim N(0,(\tan{\theta})^2)$, $c=\sqrt{2}{\mathrm{erf}}^{-1}(1-\frac{\alpha}{d})$ for error function $\mathrm{erf}$.
\end{proposition}

Since the closed form of $f_{UV}(u,v)$ is well-known, we can compute $P(\theta)$ via numerical methods implemented in various libraries, \textit{e.g.}, SciPy~\cite{virtanen2020scipy}.

\subsection{Proof of Proposition 1}

We will prove this proposition by (1) replacing the condition of sampling from $\mathbb{S}^{d-1}$ with sampling from $N(0, I_{d})$ and (2) explicitly calculating $\langle T_{\alpha}(\mathbf{x}), T_{\alpha}(\mathbf{x}') \rangle$ from order statistics. To this end, we provide 6 lemmas, each of which plays a role as follows:

\begin{itemize}
    \item \textit{Lemma 1} is a somewhat weaker version of \textit{Proposition 1.}, which introduces some additional parameters related to $\epsilon$ and $\delta$. This lemma implies \textit{Proposition 1.} with the careful choice of such parameters.

    \item \textit{Lemma 2, 3} give our analysis of the relation between the uniform distribution on $\mathbb{S}^{d-1}$ and the normal distribution $N(0, I_{d})$. Thanks to these lemmas, it is enough to analyze $N(0, I_{d})$ during the proof.

    \item \textit{Lemma 4, 5, 6} give our analysis of the random variable $T_{\alpha}(\mathbf{X})$ for $\mathbf{X} \sim N(0, I_{d})$. Note that our transformation has a strong relationship with order statistics.
\end{itemize}

\noindent A simple schematic diagram of our proof is given in Figure~\ref{fig:schem}.

\begin{figure}[h]
    \centering
    \includegraphics[height=2cm]{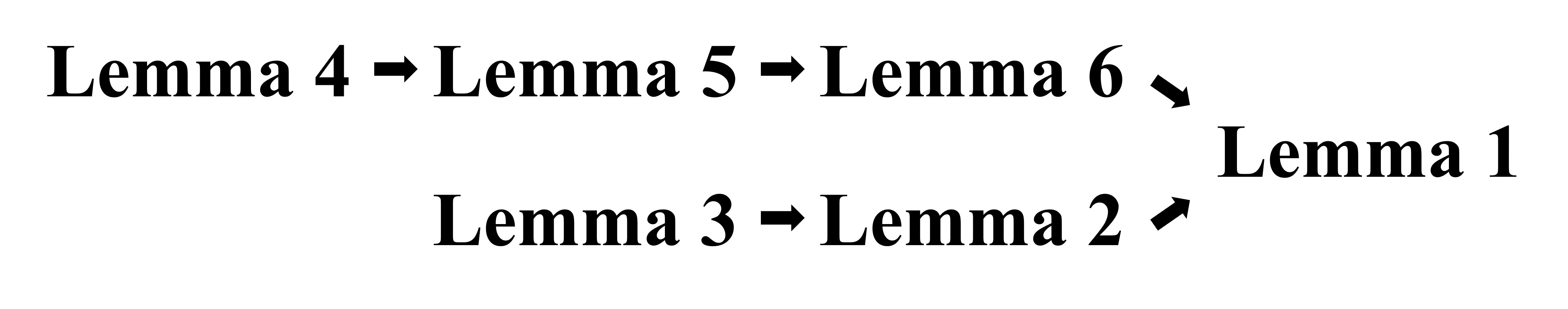}
    \vspace{-5mm}
    \caption{Schematic diagram for our strategy on the proof of \textit{Proposition 1}.}
    \label{fig:schem}
    \vspace{-3mm}
\end{figure}

% Lemma 1

Now, we first introduce our \textit{Lemma 1}. 

\begin{lemma}\label{lemma1}
    Let $\mathbf{X}, \mathbf{Y}$ be two independent random variables following $N(0, I_{d})$. Then for fixed $\theta \in (0, \pi/2) \cup (\pi/2, \pi)$ and the random variable $\mathbf{W} = \frac{\mathbf{X}+\tan\theta \mathbf{Y}}{\sqrt{1 + \tan^{2}\theta}}$, let us define the random variable $E$ as
    \begin{align}
        E = \left|\langle T_{\alpha}(\mathbf{X}), T_{\alpha}(\mathbf{W}) \rangle - \frac{\langle \mathbf{X}, \mathbf{W} \rangle}{||\mathbf{X}||_{2} \cdot ||\mathbf{W}||_{2}} \right| \nonumber
    \end{align}
    \noindent Then, for $\xi_{1}>0, \xi_{2}=d^{-1/2+\epsilon}$ for $\epsilon \in (0, 1/2)$ the following holds:
    \begin{align}
        \mathrm{Pr}\left[\left| E - |\cos\theta - \frac{2P(\theta, \xi_{2})d}{\alpha}| \right| < \xi_{1} + \xi_{2} + o(1) \right] > 1 - \delta, \nonumber
    \end{align}
    \noindent where $\delta = d^{-1}(3\xi_{1}^{-2} + 64\xi_{2}^{-2})$, and
    \begin{align}
        P(\theta, \xi_{2})=\int_{c + \xi_{2}}^{\infty}\int_{\frac{c+\xi_{2}}{\cos{\theta}}-u}^{\infty} f_{UV}(u,v) dv du
    - \int_{c+\xi_{2}}^{\infty}\int_{-\infty}^{-\frac{c + \xi_{2}}{\cos{\theta}}-u} f_{UV}(u,v) dv du, \nonumber
    \end{align}
    \noindent $f_{UV}(u,v)$ is a joint pdf of bivariate normal distribution with two independent random variables $U \sim N(0,1)$ and $V \sim N(0, (\tan\theta)^{2})$, and $c = \sqrt{2}\mathrm{erf}^{-1}(1-\frac{\alpha}{d})$ for error function $\mathrm{erf}$.
\end{lemma}
We can derive \textit{Proposition 1} directly from \textit{Lemma 1}. This is because we have that 
\begin{align}
    1-\delta &< \mathrm{Pr}\left[\left| E - \big |\cos\theta -\nonumber\frac{2P(\theta, \xi_{2})d)}{\alpha} \big | \right| < \xi_{1} + \xi_{2} + o(1) \right] \\ \nonumber
    &= \mathrm{Pr} \left[E < \epsilon_{\alpha, \theta} + (\xi_{1} + \xi_{2} + o(1)) \right]. \nonumber
\end{align}
\noindent In this case, if we choose $\xi_{1} = \xi_{2} = d^{-1/3}$, then $\xi_{1} + \xi_{2} = O(d^{-1/3}) = o(1)$ and $\delta = O(d^{-1/3}) = o(1)$, showing our $T_{\alpha}$ is $(\epsilon_{\alpha, \theta} + o(1)), o(1), \theta)$-isometry.

Our strategy for proving \textit{Lemma 1} is as follows: We first apply the triangular inequality to the involved term on the R.H.S. of the inequality in order to approximate each term contained in $E$ separately and merge them together:
\begin{align}
\hspace*{-2mm}
    \left| E - \big|\cos\theta - \frac{2P(\theta, \xi_{2})d}{\alpha} \big| \right| \le
    \left| \frac{\langle \mathbf{X}, \mathbf{W} \rangle}{||\mathbf{X}||_{2} \cdot ||\mathbf{W}||_{2}} - \cos\theta \right|
    + \left| \langle T_{\alpha}(\mathbf{X}), T_{\alpha}(\mathbf{W}) \rangle
    - \frac{2P(\theta, \xi_{2})d}{\alpha} \right|.\label{eq:5}
\end{align}
\noindent We will bound each term in Eq.~(\ref{eq:5}) with some real numbers $\xi_{1}, \xi_{2}$ respectively. Precisely, if we show that
\begin{gather}
    \mathrm{Pr}\left[\left| \frac{\langle \mathbf{X}, \mathbf{W} \rangle}{||\mathbf{X}||_{2} \cdot ||\mathbf{W}||_{2}} - \cos\theta \right| < \xi_{1} + o(1) \right] > 1 - \delta_{1}, \\
    \mathrm{Pr}\left[\left| \langle T_{\alpha}(\mathbf{X}), T_{\alpha}(\mathbf{W}) \rangle - \frac{2P(\theta, \xi_{2})d}{\alpha} \right| < \xi_{2} \right] > 1 - \delta_{2}\label{eq:7}
\end{gather}
\noindent for $\delta_{1} = O(\xi_{1}^{-2}d^{-1})$ and $\delta_{2} = O(\xi_{2}^{-2}d^{-1})$, then the R.H.S. of Eq.~(\ref{eq:5}) is bounded by $\xi_{1}+\xi_{2} + o(1)$ with probability at least $1-(\delta_{1}+\delta_{2}) = 1 - O(d^{-1}(\xi_{1}^{-2}+\xi_{2}^{-2}))$. To this end, we first give a \textit{Lemma 2} for dealing with the former term.
\begin{lemma}
    Let $\mathbf{X}, \mathbf{Y} \sim N(0, I_{d})$ be two independent random variables, and define random variables
        $\mathbf{W} = \frac{\mathbf{X}+\tan\theta\mathbf{Y}}{\sqrt{1+\tan^{2}\theta}}$, $\mathbf{C} = \frac{\langle \mathbf{X}, \mathbf{W} \rangle}{||\mathbf{X}||_{2} \cdot ||\mathbf{W}||_{2}}.$
    Then for $\xi_{1} > 0$, the following holds:
    \begin{align}
    \mathrm{Pr} \left[\left|C - \cos\theta \right| < \xi_{1} + o(1) \right] > 1 - 3\xi_{1}^{-2}d^{-1}. \nonumber
    \end{align}
\end{lemma}
\noindent Prior to proving \textbf{Lemma 2}, we provide a useful fact that gives a relation between the uniform distribution on $\mathbb{S}^{d-1}$ and the normal distribution $N(0, I_{d})$.
\begin{fact}[\cite{muller1959note}] 
Let $\mathbf{X} \sim N(0, I_{d})$ be a random variable. Then the random variable $\mathbf{Y}:=\frac{\mathbf{X}}{||\mathbf{X}||_{2}}$ follows the uniform distribution on $\mathbb{S}^{d-1}$.
\end{fact}
\noindent In addition, we give another lemma showing that two vectors uniformly sampled in $\mathbb{S}^{d-1}$ are ``almost" orthogonal. 
\begin{lemma}
    Let $\mathbf{X}, \mathbf{Y} \sim N(0, I_{d})$ be two independent random variables. For the random variable $C$ defined as $C = \frac{\langle \mathbf{X}, \mathbf{Y} \rangle}{||\mathbf{X}||_{2} \cdot ||\mathbf{Y}||_{2}}$, then $\mathbb{E}[C] = 0$ and $\mathrm{Var}[C] = \frac{1}{d}$.
\end{lemma}

\begin{proof}
We will explicitly calculate $\mathbb{E}[C]$ and $\mathrm{Var}[C]$, as follows. For simplicity, we will denote $\mathbf{X} = (X_{1}, \dots, X_{d})$, and $\mathbf{Y} = (Y_{1}, \dots, Y_{d})$. Then for $\tilde{X_{1}} = \frac{X_{1}}{\sqrt{\sum_{i=1}^{d}X_{i}^{2}}}$ and $\tilde{Y_{1}} = \frac{Y_{1}}{\sqrt{\sum_{i=1}^{d}Y_{i}^{2}}}$,
\begin{align}
    \mathbb{E}[C] = \mathbb{E}\left[\frac{\sum_{i=1}^{d}X_{i}Y_{i}}{||\mathbf{X}||_{2} \cdot ||\mathbf{Y}||_{2}} \right] = d\cdot \mathbb{E}[\tilde{X_{1}}] \cdot \mathbb{E}[\tilde{Y_{1}}]. \nonumber
\end{align}
By the \textit{Fact 1}, $\tilde{X}_{1}$ follows the distribution of the 1st component of the vector sampled from $\mathbb{S}^{d-1}$. By the radial symmetry of $\mathbb{S}^{d-1}$, $\mathbb{E}[\tilde{X}_{1}]=0$ holds, so $\mathbb{E}[C] = 0$. We also compute $\mathrm{Var}[C]$ by
\begin{align}
    \mathrm{Var}[C] &= d \cdot \mathbb{E} \nonumber[\tilde{X_{1}}^{2}\tilde{Y_{1}}^{2}] + d(d-1) \cdot \mathbb{E}[\tilde{X_{1}}\tilde{X_{2}}\tilde{Y_{1}}\tilde{Y_{2}}] \\ \nonumber
    &= d \cdot \mathbb{E}[\tilde{X_{1}}^{2}] \cdot \mathbb{E}[\tilde{Y_{1}}^{2}] + d(d-1) \cdot \mathbb{E} [\tilde{X_{1}}\tilde{X_{2}}]\mathbb{E}[\tilde{Y_{1}}\tilde{Y_{2}}]. \nonumber
\end{align}
\noindent Note that $\tilde{X_{1}}^{2}$ follows beta distribution $\mathrm{Beta}(1/2, (d-1)/2)$ and $\mathbb{E}[\tilde{X_{1}}^{2}]$ = $1/d$. For two random variables $U, V$ such that $U = \frac{X_{1} + X_{2}}{\sqrt{2}}, V= \frac{X_{1}-X_{2}}{\sqrt{2}}$, we can deduce that 
\begin{align}
    \tilde{X_{1}}\tilde{X_{2}} = \frac{1}{2}\frac{U^{2}}{U^{2} + V^{2} + \sum_{i=3}^{d}X_{i}^{2}} - \frac{1}{2}\frac{V^{2}}{U^{2} + V^{2} + \sum_{i=3}^{d}X_{i}^{2}}. \nonumber
\end{align}
\noindent Since all of both term on L.H.S. follows $\mathrm{Beta}(1/2, (d-1)/2)$, we finally obtain $\mathbb{E}[\tilde{X_{1}}\tilde{X_{2}}] = 0$, and $\mathrm{Var}[C] = 1/d$, as claimed.
\end{proof}

By using the same technique as \textit{Lemma 3}, it seems that \textit{Lemma 2} can be easily proved. However, we need to consider the dependence of two random variables $\mathbf{X}$ and $\mathbf{W}$. To deal with such dependence, we need to introduce an assumption as follows:
\begin{assumption}
Let $\mathbf{X}$ and $\mathbf{Y}$ be two independent random variables following $N(0, I_{d})$. For fixed $\theta$, let us define the random variable $\mathbf{W} = \frac{X + \tan\theta\mathbf{Y}}{\sqrt{1 + \tan^{2}\theta}}$. Then we have that $\mathbb{E}\left[\frac{\langle \mathbf{X}, \mathbf{W} \rangle}{||\mathbf{X}||_{2} \cdot ||\mathbf{W}||_{2}} \right] = \cos\theta + o(1)$.
\end{assumption}
\noindent For justification, we will experimentally verify this assumption in the next section. By admitting that the \textit{Assumption 1} holds, we can prove the \textit{Lemma 2} as follows.

\vspace{2mm}

\begin{proof}[Proof of Lemma 2]
By the same strategy as \textit{Lemma 3}, we will calculate $\mathbb{E}[C]$ and $\mathrm{Var}[C]$ directly. From the \textit{Assumption 1}, $\mathbb{E}[C]$ is approximated as $\cos\theta$ with an error term $o(1)$. Therefore, it suffices to compute the reasonable bound of $\mathrm{Var}[C]$. By using the same technique and notation as \textit{Lemma 3},
\begin{align}
    \mathrm{Var}[C] &= \mathbb{E}[C^{2}] - \mathbb{E}[C]^{2}
    \nonumber \\ \nonumber
    &=d\mathbb{E}[\tilde{X_{1}}^{2}\tilde{W_{1}}^{2}] + d(d-1) \mathbb{E} [\tilde{X_{1}}\tilde{W_{1}}]\mathbb{E} [\tilde{X_{2}}\tilde{W_{2}}] - d^{2}\mathbb{E}[\tilde{X_{1}}\tilde{W_{1}}]^{2} \\ \nonumber
    &= d(\mathbb{E}[\tilde{X_{1}}^{2}]\mathbb{E}[\tilde{W_{1}}^{2}]+\mathrm{Cov}(\tilde{X_{1}}^{2}, \tilde{W_{1}}^{2}) - \mathbb{E}[\tilde{X_{1}}\tilde{W_{1}}]^{2}), \nonumber
\end{align}
\noindent Since $\mathrm{Cov}(\tilde{X_{1}}^{2}, \tilde{W_{1}}^{2}) \le \sqrt{\mathrm{Var}[\tilde{X_{1}}^{2}]\cdot\mathrm{Var}[\tilde{W_{1}}^{2}]}$, along with the fact that $\tilde{X_{1}}^{2}$ follows beta distribution $\mathrm{Beta}(1/2, (d-1)/2)$, we obtain 
\begin{align}
    \mathrm{Var}[C]
    \le d(\mathbb{E}[\tilde{X_{1}}^{2}]^{2} + \mathrm{Var}[\tilde{X_{1}}]- \mathbb{E}[\tilde{X_{1}}\tilde{W_{1}}]^{2})  
    \le d \left(\frac{1}{d^{2}} + \frac{2(d-1)}{d^{2}(d+1)} \right)
    \le 3d^{-1}. \nonumber
\end{align}
\noindent By Chebyshev's inequality, for $\xi_{1}>0$, $\xi_{E} = \mathbb{E}[C] - \cos\theta$,
\begin{align}
    \mathrm{Pr}[|C - \cos\theta| < \xi_{1} + \xi_{E}]
    > 1 - \mathrm{Var}[C]\xi_{1}^{-2}
    > 1 - 3\xi_{1}^{-2}d^{-1}. \nonumber
\end{align}
Since $\xi_{E} = o(1)$ by \textit{Assumption 1}, we have the desired result.
\end{proof}

On the other hand, we need to handle the latter term of Eq.~(\ref{eq:5}). Our observation is the following: let $\mathbf{x}, \mathbf{y} \in \mathbb{S}^{d-1} \subset \mathbb{R}^{d}$. Then, when we denote $(I_{\mathbf{x}}^{+}, I_{\mathbf{x}}^{-})$, $(I_{\mathbf{y}}^{+}, I_{\mathbf{y}}^{-})$ as the sets of indices corresponding components of $T_{\alpha}(\mathbf{x}), T_{\alpha}(\mathbf{y})$ are $\frac{+1}{\sqrt{\alpha}},\frac{-1}{\sqrt{\alpha}}$, respectively, we have that
\begin{align}
    \alpha \cdot \langle T_{\alpha}(\mathbf{x}), T_{\alpha}(\mathbf{y}) \rangle =
    (|I_{\mathbf{x}}^{+} \cap I_{\mathbf{y}}^{+}| + |I_{\mathbf{x}}^{-} \cap I_{\mathbf{y}}^{-}|)
    - (|I_{\mathbf{x}}^{+} \cap I_{\mathbf{y}}^{-}| + |I_{\mathbf{x}}^{-} \cap I_{\mathbf{y}}^{+}|).\label{eq:33}
\end{align}
\noindent We focus on each term in Eq.~(\ref{eq:33}). For two random variables $\mathbf{X} = (X_{1}, \dots, X_{d})$ and $\mathbf{Y} = (Y_{1}, \dots, Y_{d})$, we introduce new random variables $I_{\mathbf{X}\mathbf{Y}}^{\mathrm{pp}}, I_{\mathbf{X}\mathbf{Y}}^{\mathrm{mm}}, I_{\mathbf{X}\mathbf{Y}}^{\mathrm{mp}}, I_{\mathbf{X}\mathbf{Y}}^{\mathrm{pm}}$ such that
\begin{align}
I_{\mathbf{X}\mathbf{Y}}^{\mathrm{pp}}= |I_{\mathbf{X}}^{+} \cap I_{\mathbf{Y}}^{+}|, \quad I_{\mathbf{X}\mathbf{Y}}^{\mathrm{mm}}=|I_{\mathbf{X}}^{-} \cap I_{\mathbf{Y}}^{-}|, \quad
I_{\mathbf{X}\mathbf{Y}}^{\mathrm{pm}} = |I_{\mathbf{X}}^{+} \cap I_{\mathbf{Y}}^{-}|, \quad I_{\mathbf{X}\mathbf{Y}}^{\mathrm{mp}} = |I_{\mathbf{X}}^{-} \cap I_{\mathbf{Y}}^{+}|. \label{eq:35}
\end{align}
To analysis the non-zero components after applying $T_{\alpha}$, we need to consider the $(d-\alpha)$th order statistics $\tilde{X}_{(d-\alpha)}$ of $|X_{1}|, \dots, |X_{d}|$. By using this notation, we can derive that for fixed $i \in [d]$,
\begin{align}
    \mathrm{Pr}[i \in I_{\mathbf{X}}^{+} \cap I_{\mathbf{Y}}^{+}]= 
    \mathrm{Pr}[(|X_{i}| > \tilde{X}_{(d-\alpha)}) \land (|Y_{i}| > \tilde{Y}_{(d-\alpha)}) \land (X_{i}>0) \land (Y_{i}>0)]. \label{eq:35new}
\end{align}
To compute this probability, we first need to study the property of order statistics. In fact, there is a lemma from \cite{monsteller1946on}, which tells us the asymptotic behavior of order statistics.

\begin{lemma}[\cite{monsteller1946on}]
    Let $X_{1}, \dots, X_{d}$ be i.i.d. random variables following the distribution $D$ and let us denote $X_{(d-\alpha)}$ as an $\alpha$th order statistics of them. Then for $\alpha = Cd \in [d]$ for constant $C\in(0,1)$, $X_{(d-\alpha)}$ asymptotically follows a normal distribution, whose mean and variance are 
    $
        \mathbb{E}[X_{(d-\alpha)}] = F_{X}^{-1}(1-\alpha/d)$ and $
        \mathrm{Var}[X_{(d-\alpha)}] = \frac{\alpha(d-\alpha)}{d^{2}(f_{X}(F_{X}^{-1}(1-\alpha/d)))^{2}},$
    where $f_{X}, F_{X}$ is the pdf, cdf of the distribution $D$, respectively.
\end{lemma}
\noindent In fact, we used $\alpha = \lfloor\frac{2d}{3}\rfloor$ for our transformation, so this lemma is applicable to our setting. From this lemma, we can obtain a reasonable threshold to estimate the output of the transform. More precisely, we can estimate the probability in Eq.~(\ref{eq:36}) for a fixed $i$ as the following lemma.
\begin{align}
    \mathrm{Pr}[i \in I_{\mathbf{X}}^{+}] = \mathrm{Pr}[(|X_{i}|>\tilde{X}_{(d-\alpha)}) \land (X_{i}>0)].\label{eq:36}
\end{align}
\begin{lemma}
    Let $X_{1}, \dots, X_{d}$ be i.i.d. random variables following the distribution $D$, and let us denote $X_{(\alpha)}$ as $\alpha$th order statistics of $X_{1}, \dots, X_{d}$. Then we have that for $i \in [d]$, $\alpha = Cd \in [d]$ for constant $C \in (0,1)$ and $\xi_{2} > 0$,
    \begin{align}
        \left| \mathrm{Pr}\left[X_{i} > X_{(d-\alpha)} \right] - \mathrm{Pr}\left[ X_{i} > F^{-1}\left(1-\frac{\alpha}{d}\right) + \xi_{2} \right] \right|
        < \max \left\{\xi_{2}, \frac{1}{8\xi_{2}^{2}d(f(F^{-1}(1-\alpha /d)))^{2}} \right\}, \nonumber 
    \end{align}
    where $f, F$ is the pdf, cdf corresponding to $D$, respectively.
\end{lemma}

% \vspace{5mm}

\begin{proof}
Let us denote $\mathcal{E}_{1}$, $\mathcal{E}_{2}$, and $\mathcal{E}_{3}$ as events corresponding to $X_{i} > X_{(d-\alpha)}$, $X_{i} > F^{-1}(1-\alpha/d) + \xi_{2}$, and $X_{(d-\alpha)} > F^{-1}(1-\alpha/d) + \xi_{2}$, respectively. In this setting, we can deduce that if $\mathcal{E}_{2}$ holds but $\mathcal{E}_{1}$ does not, then $\mathcal{E}_{3}$ holds automatically. This gives $\mathrm{Pr}[\lnot \mathcal{E}_{1} \land \mathcal{E}_{2}] \le \mathrm{Pr}[\mathcal{E}_{3}]$. We will focus on calculating $\mathrm{Pr}[\mathcal{E}_{3}]$. According to \textit{Lemma 4} and Chebyshev's inequality, we have that
\begin{align}
    \mathrm{Pr}\left[ \left| X_{(d-\alpha)} - F^{-1}\left(1-\frac{\alpha}{d} \right) \right| > \xi_{2} \right]
    < \frac{\xi_{2}^{-2}d^{-1}}{4(f(F^{-1}(1-\alpha /d)))^{2}}. \nonumber
\end{align}
Here, we used the fact that $\frac{\alpha}{d}(1 - \frac{\alpha}{d}) \le \frac{1}{4}$ because $0<\frac{\alpha}{d}<1$. Since $X_{(d-\alpha)}$ asymptotically follows normal distribution, we can expect that $X_{(d-\alpha)}$ has a symmetry on the line $x = F^{-1}(1-\alpha / d)$. Therefore, we have 
\begin{align}
    \mathrm{Pr}[\mathcal{E}_{1}] - \mathrm{Pr}[\mathcal{E}_{2}] &=  \nonumber \mathrm{Pr}[\mathcal{E}_{1} \land \lnot \mathcal{E}_{2}] - \mathrm{Pr}[\lnot \mathcal{E}_{1} \land \mathcal{E}_{2}] \\
    &\ge -\mathrm{Pr}[\lnot \mathcal{E}_{1} \land \mathcal{E}_{2}] 
    \ge -\frac{1}{8\xi_{2}^{2}d } \cdot \frac{1}{(f(F^{-1}(1-\alpha /d)))^{2}}.\label{eq:41}
\end{align}
\noindent Also, if we calculate $\mathrm{Pr}[\mathcal{E}_{1}]$ and $\mathrm{Pr}[\mathcal{E}_{2}]$ explicitly, we have that $\mathrm{Pr}[\mathcal{E}_{1}] = \alpha / d$, and
\begin{align}
    \mathrm{Pr}[\mathcal{E}_{2}] = \int_{\tau}^{\infty}f(x) dx = \lim_{t \rightarrow \infty}F(t) - F(\tau) = 1 - F(\tau), \nonumber
\end{align}
\noindent where $\tau = F^{-1}(1 - \alpha / d) + \xi_{2}$. If we consider first order Taylor approximation of $F$ on $\tau$, we obtain
\begin{align}
    \mathrm{Pr}[\mathcal{E}_{2}] = \alpha / d  - \xi_{2}f(F^{-1}(1 - \alpha / d))- O(\xi_{2}^{2}) \nonumber
\end{align}
\noindent Therefore, we have that $\mathrm{Pr}[\mathcal{E}_{1}] - \mathrm{Pr}[\mathcal{E}_{2}] < \xi_{2}$. By combining this with Eq.~(\ref{eq:41}), we finally obtain the following inequality:
\begin{align}
    \left|\mathrm{Pr}[\mathcal{E}_{1}] - \mathrm{Pr}[\mathcal{E}_{2}]\right|
    <\max \left\{\xi_{2}, \frac{1}{8\xi_{2}^{2}d(f(F^{-1}(1-\alpha /d)))^{2}} \right\}.\label{eq:47}
\end{align}
This completes the proof.
\end{proof}

% \vspace{3mm}
If we set $\xi_{2}$ such that $\xi_{2} =d^{-1/2 + \epsilon}$ for some $\epsilon \in (0, 1/2)$, then as $d \rightarrow \infty$, the R.H.S of the Eq.~(\ref{eq:47}) tends to 0. That is, the probability calculated from the threshold ($\mathrm{Pr}[\mathcal{E}_{2}]$) is a good approximation for the probability of whether the given entry surpasses the $(d-\alpha)$th order statistics or not ($\mathrm{Pr}[\mathcal{E}_{1}]$). We can apply the same argument on estimating Eq.~(\ref{eq:35new}), as the following lemma. We will denote the upper bound on the probability obtained in \textit{Lemma 5} as $\mathcal{U}_{\alpha, \xi_{2}}$.
\begin{lemma}
    Let $\mathbf{X}, \mathbf{Y}$ be two independent random variables such that $\mathbf{X} \sim N(0, I_{d})$ and $\mathbf{Y} \sim N(0, (\tan^{2}\theta) I_{d})$. Then for the random variable $\mathbf{W} = \frac{X+Y}{\sqrt{1 + \tan^{2}\theta}}$, $\alpha = Cd \in [d]$ for constant $C \in (0,1)$, fixed $i \in [d]$ and $\xi_{2} \in (0, C)$, define
    \begin{align}
        P^{+}(\theta, \xi_{2}) = \int_{c+\xi_{2}}^{\infty} \int_{\frac{c+\xi_{2}}{\cos\theta} - u}^{\infty} f_{UV}(u,v)dvdu, \quad 
        \hat{P}^{+}(\theta) = \mathrm{Pr}\left[i \in I_{\mathbf{X}}^{+} \cap I_{\mathbf{W}}^{+} \right], \nonumber
    \end{align}
    \noindent where $f_{UV}$ is the joint pdf of two independent random variables $U \sim N(0, 1)$ and $V \sim N(0, \tan^{2}\theta)$, and $c = \sqrt{2}\mathrm{erf}^{-1}(1-\frac{\alpha}{d})$. Then we have that
    \begin{align}
        |P^{+}(\theta, \xi) - \hat{P}^{+}(\theta)| < \left(1+\frac{2C - \xi_{2} + 1}{C^{2} - C\xi_{2}}  \right) \mathcal{U}_{\alpha, \xi_{2}}, \nonumber
    \end{align}
    \noindent for pdf of the random variable $|X|$ for $X \sim N(0, 1)$ calculated as
    \begin{align}
        f_{|X|}(x) = \begin{cases}
        \frac{2}{\sqrt{2\pi}}e^{-\frac{x^{2}}{2}}, \quad \text{if $x\ge0$} \\
        0, \qquad \text{Otherwise}
        \end{cases} \nonumber
    \end{align}
\end{lemma}
\begin{proof} 
Let us denote $\mathbf{X} = (X_{1}, \dots, X_{d})$ and $\mathbf{W} = (W_{1}, \dots, W_{d})$, and let $\tilde{X}_{(\alpha)}$ and $\tilde{W}_{(\alpha)}$ be the order statistics of $|X_{1}|, \dots, |X_{d}|$ and $|W_{1}|, \dots, |W_{d}|$, respectively. Also, we define events $\hat{\mathcal{E}_{1}}, \hat{\mathcal{E}_{2}}, \mathcal{E}_{1}$ and $\mathcal{E}_{2}$ as 
\begin{gather}
(X_{i} > \tilde{X}_{(d-\alpha)}), \;(W_{i} > \tilde{W}_{(d-\alpha)}), %
(X_{i} > c + \xi_{2}), \; (W_{i} > c + \xi_{2}), \nonumber
\end{gather}
respectively, where $c = \sqrt{2}\mathrm{erf}^{-1}(1 - \alpha/d)$. Then $\hat{P}^{+}(\theta) = \mathrm{Pr}[\hat{\mathcal{E}_{1}} \land \hat{\mathcal{E}_{2}}]$, and for joint pdf $f_{UV}(u,v)$ of two independent random variables $U \sim N(0, 1)$ and $V \sim N(0, \tan^{2}\theta)$,
\begin{align}
    \mathrm{Pr}[\mathcal{E}_{1} \land \mathcal{E}_{2}] \nonumber
    = \mathrm{Pr}[(X_{i} > c + \xi_{2}) \land (W_{i} > c + \xi_{2})] %\\ \nonumber
    = \int_{c+\xi_{2}}^{\infty} \int_{\frac{c+\xi_{2}}{\cos\theta}-u}^{\infty}f_{UV}(u,v)dvdu = P^{+}(\theta, \xi_{2}). \nonumber
\end{align}
\noindent In fact, $c$ corresponds to the $F^{-1}(1-\alpha/d)$ term on \textit{Lemma 4}. Note that the cdf $F_{|X|}$ of $|X|$ is
\begin{align}
    F_{|X|}(x) = \begin{cases}
        \frac{2}{\sqrt{2\pi}}\int_{0}^{x}e^{-\frac{t^{2}}{2}} dt, \quad \text{if $x \ge 0$} \\ \nonumber
        0, \qquad \text{Otherwise}
    \end{cases},
\end{align}
\noindent and can be expressed by the error function defined as $\mathrm{erf}(x) = \frac{2}{\sqrt{\pi}}\int_{0}^{x}e^{-t^{2}}dt$.

Now, to estimate the probability $|P^{+}(\theta, \xi_{2}) - \hat{P}^{+}(\theta)|$, we first derive that
\begin{align}
    \left|P^{+}(\theta, \xi_{2}) - \hat{P}^{+}(\theta)\right| = \nonumber\left|\mathrm{Pr}[\hat{\mathcal{E}_{1}} \land \hat{\mathcal{E}_{2}}] - \mathrm{Pr}[\mathcal{E}_{1} \land \mathcal{E}_{2}] \right|
    \le \left| \mathrm{Pr}[\hat{\mathcal{E}_{2}}|\hat{\mathcal{E}_{1}}] - \mathrm{Pr}[\mathcal{E}_{2} | \mathcal{E}_{1}] \right| + \left|\mathrm{Pr}[\mathcal{E}_{1}]-\mathrm{Pr}[\hat{\mathcal{E}_{1}}]\right|. \nonumber
\end{align}
\noindent For the former term, we can obtain
\begin{align}
    &\left|\mathrm{Pr}[\hat{\mathcal{E}_{2}} | \hat{\mathcal{E}_{1}}] - \mathrm{Pr}[\mathcal{E}_{2} | \mathcal{E}_{1}]\right| 
    \le \left|\mathrm{Pr}[\hat{\mathcal{E}_{2}} |  \nonumber\hat{\mathcal{E}_{1}}] - \mathrm{Pr}[\mathcal{E}_{2} | \hat{\mathcal{E}_{1}}]\right| + \left|\mathrm{Pr}[\mathcal{E}_{2} | \hat{\mathcal{E}_{1}}] - \mathrm{Pr}[\mathcal{E}_{2} | \mathcal{E}_{1}]\right| \\ \nonumber
    &\le \frac{\left|\mathrm{Pr}[\hat{\mathcal{E}_{2}} \land \hat{\mathcal{E}_{1}}] - \mathrm{Pr}[\mathcal{E}_{2} \land \hat{\mathcal{E}_{1}}]\right|}{\mathrm{Pr}[\hat{\mathcal{E}}_{1}]}
    + \left|\frac{\mathrm{Pr}[\mathcal{E}_{2} \land \hat{\mathcal{E}_{1}}]}{\mathrm{Pr}[\hat{\mathcal{E}_{1}]}} - \frac{\mathrm{Pr}[\mathcal{E}_{2} \land \hat{\mathcal{E}_{1}]}}{\mathrm{Pr}[\mathcal{E}_{1}]} \right|
    + \left| \frac{\mathrm{Pr}[\mathcal{E}_{2} \land \hat{\mathcal{E}_{1}]}}{\mathrm{Pr}[\mathcal{E}_{1}]} - \frac{\mathrm{Pr}[\mathcal{E}_{2} \land \mathcal{E}_{1}]}{\mathrm{Pr}[\mathcal{E}_{1}]} \right| \nonumber \\
    &\le \frac{\left|\mathrm{Pr}[\hat{\mathcal{E}_{2}}] - \mathrm{Pr}[\mathcal{E}_{2}]  \right|}{\mathrm{Pr}[\hat{\mathcal{E}}_{1}]} + \left|\mathrm{Pr}[\mathcal{E}_{2} \land \hat{\mathcal{E}_{1}}]\right| \left|\frac{\mathrm{Pr}[\mathcal{E}_{1}] - \mathrm{Pr}[\hat{\mathcal{E}_{1}}]}{\mathrm{Pr}[\mathcal{E}_{1}] \cdot \mathrm{Pr}[\hat{\mathcal{E}_{1}}]}\right|
    + \frac{\left|\mathrm{Pr}[\mathcal{E}_{1}] - \mathrm{Pr} [\hat{\mathcal{E}_{1}}]\right|}{\mathrm{Pr}[\mathcal{E}_{1}]} \nonumber \\ 
    &= \left(\frac{1}{\mathrm{Pr}[\mathcal{E}_{1}]} + \frac{1}{\mathrm{Pr}[\hat{\mathcal{E}}_{1}]} + \frac{\mathrm{Pr}[\mathcal{E}_{2} \land \hat{\mathcal{E}_{1}}]}{\mathrm{Pr}[\mathcal{E}_{1}] \cdot \mathrm{Pr}[\hat{\mathcal{E}_{1}}]} \right) \left|\mathrm{Pr}[\mathcal{E}_{1}] - \mathrm{Pr}[\hat{\mathcal{E}_{1}}] \right|. \nonumber
\end{align}
\noindent Since $\mathrm{Pr}[\hat{\mathcal{E}_{1}}] = \alpha / d$ and $\mathrm{Pr}[\mathcal{E}_{1}] > \alpha /d - \xi_{2}$, we have that for the constant $C = \alpha / d$,
\begin{align}
    \bigg|P^{+}(\theta, \xi_{2}) - \hat{P}^{+}(\theta) \bigg|
    &\le \left(1 + \frac{1}{\mathrm{Pr}[\mathcal{E}_{1}]} + \frac{1}{\mathrm{Pr}[\hat{\mathcal{E}}_{1}]} + \frac{\mathrm{Pr}[\mathcal{E}_{2} \land \hat{\mathcal{E}_{1}}]}{\mathrm{Pr}[\mathcal{E}_{1}] \cdot \mathrm{Pr}[\hat{\mathcal{E}_{1}}]} \right) \cdot \left|\mathrm{Pr}[\mathcal{E}_{1}] - \mathrm{Pr}[\hat{\mathcal{E}_{1}}] \right| \nonumber \\
    &\le \left(1 + \frac{1}{C} + \frac{1}{C - \xi_{2}} + \frac{1}{C(C-\xi_{2})} \right) \mathcal{U}_{\alpha, \xi_{2}} \nonumber \\ 
    &= \left(1+\frac{2C - \xi_{2} + 1}{C^{2} - C\xi_{2}}  \right)\mathcal{U}_{\alpha, \xi_{2}}. \label{eq:71}
\end{align}
\noindent Therefore, we obtained the desired bound.
\end{proof}

Remark that in Eq.~(\ref{eq:71}), if we assume that $\xi_{2} \in (d^{-1}, d^{-1/2})$, then for sufficiently large $d$, $\xi_{2} < C/2$. Since the function $f(x) = \frac{2C + 1 - x}{C^{2}-Cx}$ monotonically increases for $x < C$, under the condition of \textit{Lemma 1}, we obtain
\begin{align}
    \left|P^{+}(\theta, \xi_{2}) - \hat{P}^{+}(\theta) \right| \nonumber
    %\le \left(1+\frac{2C - \xi_{2} + 1}{C^{2} - C\xi_{2}}  \right)\mathcal{U}_{\alpha, \xi_{2}}. \\ \nonumber
    \le \frac{(C+1)(C+2)}{C^{2}} \mathcal{U}_{\alpha, \xi_{2}} \nonumber
\end{align}
Now we are ready to complete the proof of \textit{Lemma 1.}

% Proof of lemma 1.

\begin{proof}[Proof of Lemma 1] It suffices to derive Eq.~(\ref{eq:7}). For two independent random variables $\mathbf{X}, \mathbf{Y} \sim N(0, I_{d})$ and a fixed $\theta \in [0, \pi/2)$, let $\mathbf{W} = \frac{\mathbf{X} + \tan\theta\mathbf{Y}}{\sqrt{1 + \tan^{2}\theta}}$. From this, we define the random variables $I_{\mathbf{XW}}^\mathrm{pp}$, $I_{\mathbf{XW}}^\mathrm{mm}$, $I_{\mathbf{XW}}^\mathrm{pm}$, $I_{\mathbf{XW}}^\mathrm{mp}$ in same way as Eq.~(\ref{eq:35}). We will estimate
\begin{align}
I_{\mathbf{XW}}^\mathrm{pp}+I_{\mathbf{XW}}^\mathrm{mm} -I_{\mathbf{XW}}^\mathrm{pm}-I_{\mathbf{XW}}^\mathrm{mp} \nonumber
\end{align}
\noindent by using previous lemmas. Recall that 
\begin{align}
    \alpha \cdot \langle T_{\alpha}(\mathbf{X}), T_{\alpha}(\mathbf{W}) \rangle = 
    (|I_{\mathbf{X}}^{+} \cap I_{\mathbf{W}}^{+}| + |I_{\mathbf{X}}^{-} \cap I_{\mathbf{W}}^{-}|)
    - (|I_{\mathbf{X}}^{+} \cap I_{\mathbf{W}}^{-}| + |I_{\mathbf{X}}^{-} \cap I_{\mathbf{W}}^{+}|), \nonumber
\end{align}
\noindent and for each indices $i \in [d]$, 
\begin{align}
    \mathrm{Pr}[i \in I_{\mathbf{X}}^{+} \cap I_{\mathbf{W}}^{+}] = 
    %\hspace*{2mm}
    \mathrm{Pr}[(|X_{i}| > \tilde{X}_{(d-\alpha)}) \land (|W_{i}| > \tilde{W}_{(d-\alpha)})
    \land (X_{i}>0) \land (W_{i}>0)]. \nonumber
\end{align}
\noindent By using the result of \textit{Lemma 6}, we can approximate this probability from
\begin{align}
    \mathrm{Pr}[(X_{i} > c + \xi_{2}) \land (W_{i} > c + \xi_{2}) \land (X_{i} > 0) \land (W_{i} > 0)] \nonumber
\end{align}
\noindent within $\mathcal{U}_{\alpha, \xi_{2}}$ error, where c = $\mathrm{erf}(1 - \alpha / d)$. One may observe that the events in the above probability term are mutually independent with respect to each index. That is, when we define a random variable $\mathbf{Z} \sim N(dP^{+}(\theta, \xi_{2}),dP^{+}(\theta, \xi_{2})(1-P^{+}(\theta, \xi_{2})))$ and $\hat{\mathbf{Z}} \sim N(d\hat{P}^{+}(\theta), d\hat{P}^{+}(\theta)(1-\hat{P}^{+}(\theta)))$, $I_{\mathbf{XW}}^{pp}$ first can be approximated by $\hat{\mathbf{Z}}$, which again well approximated by $\mathcal{Z}$. To be precise, from the result of \textit{Lemma 6}, with Chebyshev's inequality for $\xi_{3}>C^{\ast}\mathcal{U}_{\alpha, \xi_{2}}d$ and $C^{\ast} = \frac{(C+1)(C+2))}{C^{2}}$,
\begin{align}
    \mathrm{Pr}\left[\left|I_{\mathbf{XW}}^{\mathrm{pp}} - P^{+}(\theta, \xi_{2})d\right| < \xi_{3} \right] \nonumber
    \ge 1 - \frac{4}{d \left(\xi_{3} -  C^{\ast}\mathcal{U}_{\alpha, \xi_{2}}d\right)^{2}}, \nonumber
\end{align}
\noindent because $\hat{P}^{+}(\theta, \xi_{2})(1-\hat{P}^{+}(\theta, \xi_{2})) < \frac{1}{4}$. By the symmetry of the normal distribution, we can do the same approximation to $I_{\mathbf{XW}}^{\mathrm{mm}}$ which results in the same bound. Also, for
\begin{align}
    P^{-}(\theta, \xi) = \int_{c+\xi}^{\infty}\int_{-\infty}^{-\frac{c+\xi}{\cos{\theta}}-u} f_{UV}(u,v) dv du, \nonumber
\end{align}
\noindent we can apply a similar argument for $I_{\mathbf{XW}}^\mathrm{pm}$ and $I_{\mathbf{XW}}^\mathrm{mp}$, namely,
\begin{align}
    \mathrm{Pr}\left[\left|I_{\mathbf{XW}}^{\mathrm{pm}} - P^{-}(\theta, \xi_{2})d \right| < \xi_{3} \right] 
    > 1 - \frac{4}{d \left(\xi_{3} - C^{\ast} \mathcal{U}_{\alpha, \xi_{2}}d \right)^{2}}. \nonumber
\end{align}
\noindent By applying triangular inequality, for
\begin{align}
    P(\theta, \xi) &= P^{+}(\theta, \xi) - P^{-}(\theta, \xi)
    = \int_{c+\xi}^{\infty}\int_{\frac{c+\xi}{\cos{\theta}}-u}^{\infty} f_{UV}(u,v) dv du 
    - \int_{c+\xi}^{\infty}\int_{-\infty}^{-\frac{c+\xi}{\cos{\theta}}-u} f_{UV}(u,v) dv du, \nonumber
\end{align}
\noindent we obtain
\begin{align}
    &\left|I_{\mathbf{XW}}^{\mathrm{pp}} + I_{\mathbf{XW}}^{\mathrm{mm}} - I_{\mathbf{XW}}^{\mathrm{pm}} - I_{\mathbf{XW}}^{\mathrm{mp}} - 2P(\theta, \xi_{2})d\right| \nonumber \\
    &\le |I_{\mathbf{XW}}^{\mathrm{pp}} - P^{+}(\theta, \xi_{2})| + |I_{\mathbf{XW}}^{\mathrm{mm}} - P^{+}(\theta, \xi_{2})|
    + |I_{\mathbf{XW}}^{\mathrm{pm}} - P^{-}(\theta, \xi_{2})| + |I_{\mathbf{XW}}^{\mathrm{mp}} - P^{-}(\theta, \xi_{2}) |.\label{eq:85} 
\end{align}
\noindent In this inequality, we can calculate the probability that each term of the R.H.S. in Eq.~(\ref{eq:85}) is bounded by $\xi_{3}/4$. Therefore, we have that
\begin{align}
\hspace*{-.5cm}
    \mathrm{Pr}\left[ \left| \langle T_{\alpha}(\mathbf{X}), T_{\alpha}(\mathbf{W}) \rangle - \frac{2P(\theta, \xi_{2})d}{\alpha}\right| < \xi_{3} \right] 
    > \left(1 - \frac{16}{d(\alpha \xi_{3} - 4d \mathcal{U}_{\alpha, \xi_{2}})^{2}}\right)^{4}
    \ge 1 - \frac{64}{d(\alpha \xi_{3} - 4 d C^{\ast}\mathcal{U}_{\alpha, \xi_{2}})^{2}}, \nonumber
\end{align}
\noindent where
\begin{align}
    P(\theta, \xi) = P^{+}(\theta, \xi) - P^{-}(\theta, \xi)= \int_{c_\xi}^{\infty}\int_{\frac{c+\xi}{\cos{\theta}}-u}^{\infty} f_{UV}(u,v) dv du
    - \int_{c+\xi}^{\infty}\int_{-\infty}^{-\frac{c_\xi}{\cos{\theta}}-u} f_{UV}(u,v) dv du. \nonumber
\end{align}
\noindent Therefore, by combining the result of \textit{Lemma 2} and taking $\alpha\xi_{3} = 4dC^{\ast}C^{-1}\mathcal{U}_{\alpha, \xi_{2}} + \xi_{2} > C^{\ast}\mathcal{U}_{\alpha, \xi_{2}}d$, we finally obtain
\begin{align}
    &\mathrm{Pr}\bigg[\bigg| E - |\cos\theta - \frac{2P(\theta, \xi_{2})d}{\alpha}| \bigg| 
    < \xi_{1} + \xi_{2} + 4C^{\ast}C^{-1}\mathcal{U}_{\alpha, \xi_{2}} + o(1) \bigg] \nonumber \\
    &\ge \mathrm{Pr}\bigg[\bigg| E - |\cos\theta - \frac{2P(\theta, \xi_{2})d}{\alpha}| \bigg| 
    < \xi_{1} + \alpha^{-1}\xi_{2} + 4C^{\ast}C^{-1}\mathcal{U}_{\alpha, \xi_{2}} + o(1) \bigg] 
    > 1-\delta \nonumber,
\end{align}
\noindent where $\delta = d^{-1}(64\xi_{2}^{-2} + 3\xi_{1}^{-2})$. Under our condition that $\xi_{2} = d^{-1/2 + \epsilon}$ for $\epsilon \in (0, 1/2)$,
\begin{align}
\mathcal{U}_{\alpha, \xi_{2}} = \max \left\{\xi_{2},\frac{\xi_{2}^{-2}d^{-1}}{8(f(F^{-1}(1-\alpha /d)))^{2}} \right\}, \nonumber
\end{align}
\noindent tends to 0, as $d$ becomes $\infty$. This completes the proof.
\end{proof}

\begin{wrapfigure}{R}{.45\linewidth}
    \vspace{-2mm}
    \includegraphics[width=\linewidth]{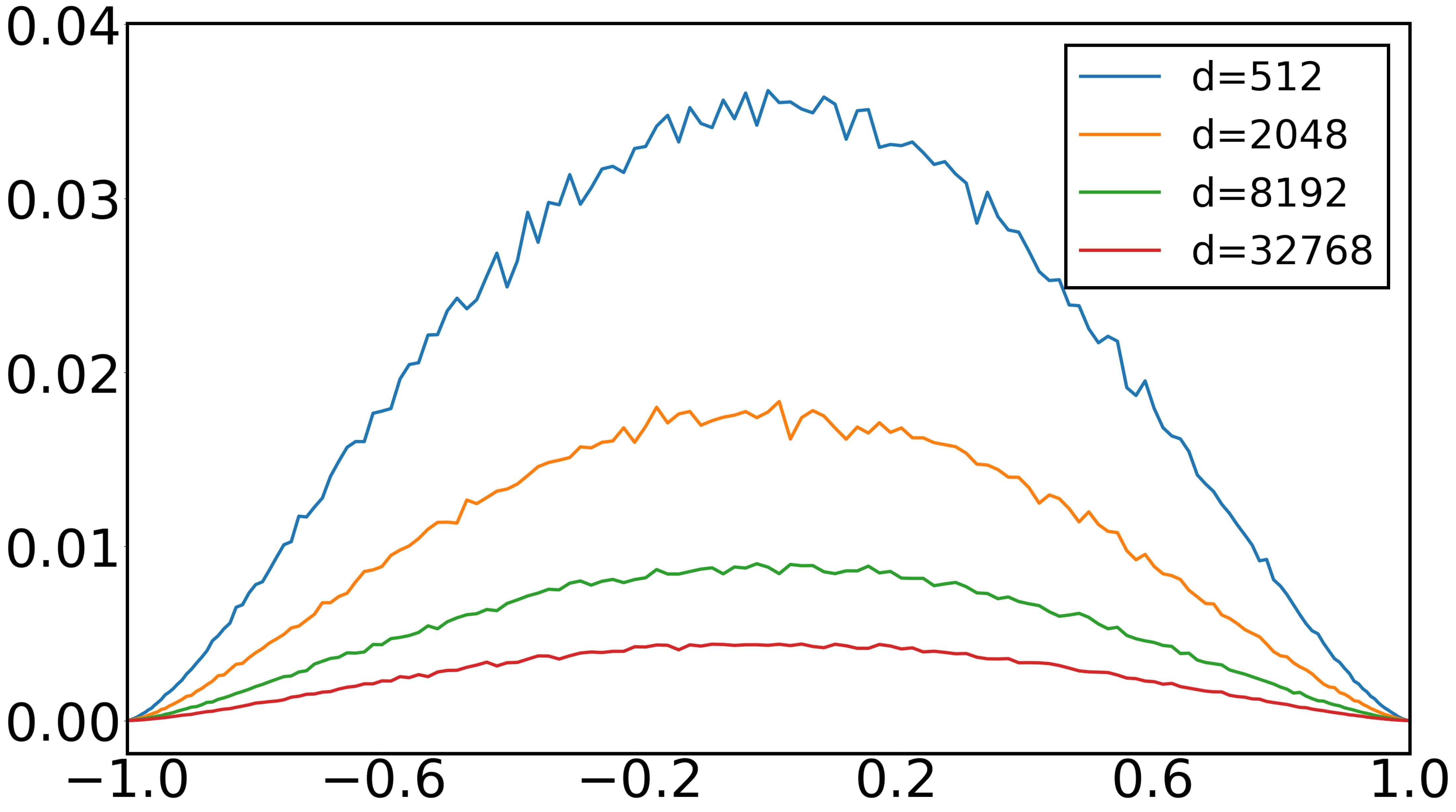}
    \vspace{-7mm}
    \caption{Empirical justification of \textit{Assumption 1}. 
    }\label{fig:assumption}
\end{wrapfigure}

\noindent \textbf{Verification of the Assumption 1.}
To complete the proof, it suffices to check whether \textit{Assumption 1} holds or not. To this end, we conducted the following experiment: Let us denote $\mathbf{X}$ and $\mathbf{Y}$ as two independent random variables following the standard normal distribution $N(0, I_{d})$. For fixed $\theta$, define $\mathbf{W} = \frac{\mathbf{X} + \tan\theta{Y}}{\sqrt{1 + \tan^{2}\theta}}$. We will check that $C = \mathbb{E}\left[\frac{\langle \mathbf{X}, \mathbf{W} \rangle}{||\mathbf{X}||_{2} \cdot ||\mathbf{W}||_{2}}\right]$ is estimated by $\cos\theta$ with error $o(1)$. 
For this, we sample 1,000 pairs of vectors $(\mathbf{x}, \mathbf{y})$ corresponding to random variables $(\mathbf{X}, \mathbf{Y})$ and calculate the sample mean $\overline{C}$ of $\frac{\langle \mathbf{X}, \mathbf{W} \rangle}{||\mathbf{X}||_{2} \cdot ||\mathbf{W}||_{2}}$. In this setting, we compute $|\overline{C} - \cos\theta|$ with various $d$ and $\theta$, and visualize it as Figure~\ref{fig:assumption}, showing that as $d$ larger, the error term $|C - \cos\theta|$ tends to 0. This is the evidence showing the \textit{Assumption 1} holds.

\section{Analysis on the Parameter Selection}\label{sec:supp_B}

In fact, Proposition~\ref{prop_1_full} does not tell us the explicit value of $\epsilon_{\alpha, \theta}$ and $\delta_{\alpha}$, or at least an upper bound of them. Since our main interest is the extent to preserve the distance relationship, we provide our experimental results with respect to $\epsilon_{\alpha, \theta}$, along with the desirable choice of $\alpha$ that minimizes $\epsilon_{\alpha, \theta}$ for all $\theta \in [0, \pi]$.

\newpage

\subsection{Desirable Choice of $\alpha$}

\begin{wrapfigure}{R}{.5\linewidth}
    \vspace{-5mm}
    \includegraphics[width=\linewidth]{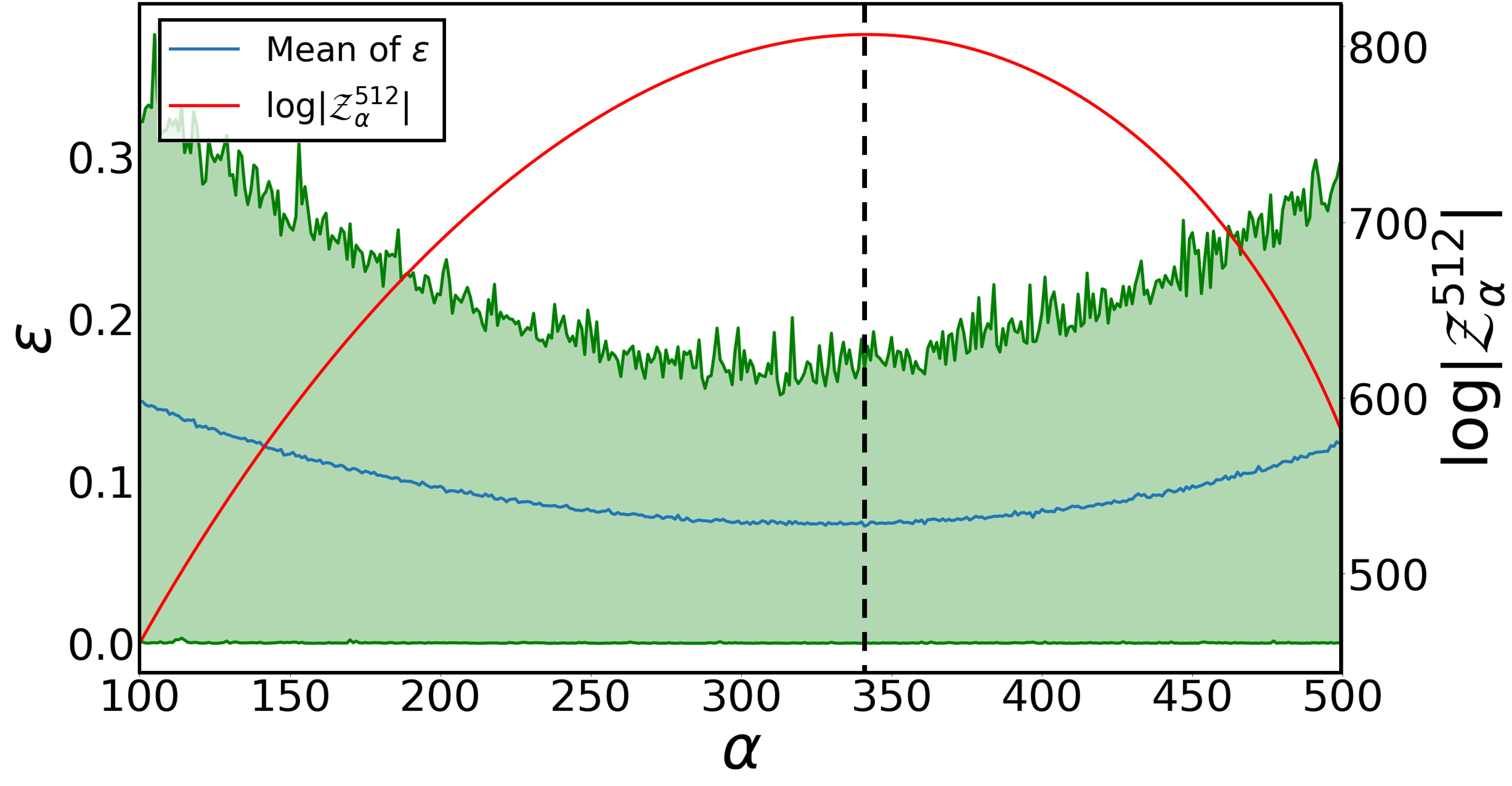}
    \vspace{-9mm}
    \caption{$\epsilon$ and $|\mathcal{Z}^{d}_{\alpha}|$ with various $\alpha$ and fixed $d = 512$. This graph tells us that $\epsilon$ is minimized when $|\mathcal{Z}^{d}_{\alpha}|$ is maximized. The dashed line indicates when $\alpha = 341$.}
    \label{fig:transf}
    \vspace{-2mm}
\end{wrapfigure}

Our first goal is to find the desirable parameter $\alpha$ that minimizes the difference between $\langle \mathbf{x}, \mathbf{y} \rangle$ and $\langle T_{\alpha}(\mathbf{x}), T_{\alpha}(\mathbf{y}) \rangle$.
To this end, one strawman's proposal is to seek a direct minimization of the value $\epsilon_{\alpha, \theta}$ in Proposition~\ref{prop_1_full}. Instead of solving such an intricate mimization problem, we take an alternative approach: selecting $\alpha$ that minimizes the deformation of the space of face templates from the transformation $T_{\alpha}$. Intuitively, this can be done by maximizing $|\mathcal{Z}_{\alpha}^{d}|$, expecting that denser transformation candidates in $\mathbb{S}^{d-1}$ leads a smaller movement of the input face template in terms of cosine similarity. 
Since $|\mathcal{Z}_{\alpha}^{d}| = {\binom{d}{\alpha}}2^{\alpha}$, which is a concave function with respect to $\alpha$, one can easily derive that $|\mathcal{Z}_{\alpha}^{d}|$ gets the maximum value when $\alpha = \lfloor \frac{2}{3}d \rfloor$.

Surprisingly, we observed that such a heuristic approach really pays off. We measured the range and the mean value of $\epsilon = |\langle T_{\alpha}(\mathbf{x}), T_{\alpha}(\mathbf{y}) \rangle - \langle \mathbf{x}, \mathbf{y} \rangle|$ from randomly sampled 100 pairs of random vectors in $\mathbb{S}^{d-1}$ with various angles and $\alpha$. Since the dimension $d$ of face templates from recent feature extractors is set to 512, according to our intuition, $\alpha = 341$ should give the smallest $\epsilon$. The visualization of the result is given in Fig.~\ref{fig:transf}. We can figure out despite of the fluctuation in the range of $\epsilon$ (green line), the mean value of $\epsilon$ (blue line) decreases as the number of codewords (red line) increases. Furthermore, the mean value attains the minimum when $\alpha = 341$, \textit{i.e.}, when $|\mathcal{Z}_{\alpha}^{d}|$ is maximized. This justifies our choice of $\alpha = 341$ throughout the parameter selection of $\mathsf{IDFace}$.

\subsection{Effect of $d$ and $\theta$ on the Error Term.} 
One may notice that for a fixed $\theta$, the value of $\epsilon_{\alpha, \theta}$ solely depends on the ratio between $\alpha$ and $d$. On the other hand, as shown in Lemma~\ref{lemma1}, the error term $o(1)$ added to the $\epsilon_{\alpha, \theta}$ is about the dimension $d$. Since we already derived that $\alpha = \lfloor \frac{2}{3}d \rfloor$, or $\frac{\alpha}{d} = \frac{2}{3}$, gives the smallest $\epsilon_{\alpha, \theta}$ on average with respect to $\theta$, we now turn our attention to analyze the effect of $d$ and $\theta$ on the exact $\epsilon_{\alpha, \theta}$ and the error term $o(1)$ when $\alpha = \lfloor \frac{2}{3} d \rfloor$.

To this end, we conducted the following experiment: for each angle $\theta \in [0, \pi]$, we sample 100 pairs of $(\mathbf{x}, \mathbf{y})$ satisfying $\langle \mathbf{x}, \mathbf{y} \rangle = \cos{\theta}$ and calculated $\left| \langle T_{\alpha}(\mathbf{x}) , T_{\alpha}(\mathbf{y}) \rangle - \langle \mathbf{x}, \mathbf{y} \rangle \right|$ for a fixed $\alpha = \lfloor \frac{2}{3}d \rfloor$.
In Figure~\ref{fig:kirruk}, we provide our experiment result for $d = 512, 2048, 8192$ and $32768$ with theoretical $\epsilon_{\alpha, \theta}$ calculated from Proposition~\ref{prop_1_full}. This tells us that the estimated $\epsilon_{\alpha, \theta}$ fits well into the empirical result $\left| \langle T_{\alpha}(\mathbf{x}) , T_{\alpha}(\mathbf{y}) \rangle - \langle \mathbf{x}, \mathbf{y} \rangle \right|$, without depending on the dimension $d$. 
As $d$ becomes larger, $\left| \langle T_{\alpha}(\mathbf{x}) , T_{\alpha}(\mathbf{y}) \rangle - \langle \mathbf{x}, \mathbf{y} \rangle \right|$ converges to our theoretical $\epsilon_{\alpha, \theta}$, \textit{i.e.}, $o(1)$ term with respect to $d$ decays to 0, as we predicted in Proposition~\ref{prop_1_full}.
On the effect of $\theta$, we can figure out that $\epsilon_{\alpha, \theta}$ gets the maximum value when $\cos\theta$ is near $\pm0.7$. We note that in this case $\epsilon_{\alpha, \theta}$ is 0.111.

\begin{figure*}[t]
    \centering
    \includegraphics[width=.99\textwidth]{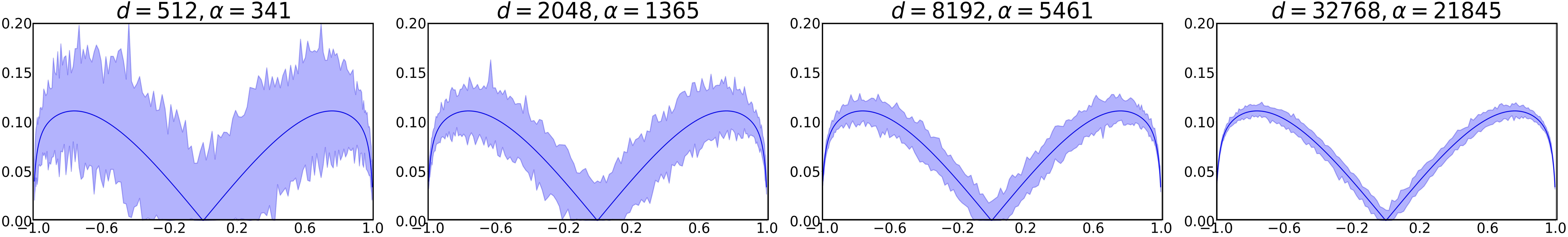}
    \vspace{-3mm}
    \caption{Experimental result on calculating $\left| \langle T_{\alpha}(\mathbf{x}) , T_{\alpha}(\mathbf{y}) \rangle - \langle \mathbf{x}, \mathbf{y} \rangle \right|$ and theoretical $\epsilon_{\alpha, \theta}$ with various $d$ and fixed $\alpha = \lfloor \frac{2}{3}d \rfloor$. The $x$-axis indicates the cosine value between $\mathbf{x}$ and $\mathbf{y}$, and $y$-axis indicates corresponding range of $\left| \langle T_{\alpha}(\mathbf{x}) , T_{\alpha}(\mathbf{y}) \rangle - \langle \mathbf{x}, \mathbf{y} \rangle \right|$. We highlighted the theoretical $\epsilon_{\alpha, \theta}$ value of each angle as a blue solid line.}
    \label{fig:kirruk}
    \vspace{-4mm}
\end{figure*}

\section{Additional Experiments and Discussions}\label{sec:supp_D}

We now provide additional experimental results omitted in the main text. This section includes (1) benchmark results of $\mathsf{IDFace}$ with other various face recognition models, (2) analysis on the effect of the almost-isometric transformation on intra-class compactness and inter-class discrepancy, (3) comparison with the PCA-based dimensionality reduction technique for further optimization, and (4) omitted discussions about the experimental results.

We note that for all experiments conducted in this section, we utilized the \texttt{insightface} library~\cite{insightface} as a baseline and simply modified the loss function for implementing other recognition models. For training each recognition model, we used four NVIDIA V100 GPUs, following the same training parameter settings, including settings for the optimizer, number of total epochs, and (total) batch size, provided by their implementation of ArcFace\footnote{For more information, we refer to the following code in the \texttt{insightface} library:  \url{https://github.com/deepinsight/insightface/blob/master/recognition/arcface_torch/configs/ms1mv3_r100.py}}.

\subsection{Evaluation Metrics}
We provide the precise definitions of evaluation metrics appeared in this paper. First of all, in the face verification scearios where LFW~\cite{huang2008labeled}, CFP-FP~\cite{sengupta2016frontal}, AgeDB~\cite{moschoglou2017agedb}, and IJB-C~\cite{maze2018iarpa} verification benchmark datasets are used for evaluation, we used true accept rate (TAR), false accept rate (FAR), and accuracy as evaluation metrics. For a precise description, we assume that the benchmark dataset $(\mathcal{D})$ consists of pairs of images from the same $(\mathcal{D}_{p})$ or different $(\mathcal{D}_{n})$ identities. 
In this setting, for a feature extract $F$ and a threshold parameter $\tau$, we define the TAR and FAR at $\tau$ as the ratio of accepted image pairs over $\mathcal{D}_{p}$ and $\mathcal{D}_{n}$, respectively, each of which is formally described as follows:
\begin{align}
    \mathsf{TAR}(\tau) = \frac{|\{(I_{x}, I_{y}) \in \mathcal{D}_{p}: s_{\cos}(F(I_{x}), F(I_{y})) > \tau \}|}{|\mathcal{D}_{p}|}, \quad 
    \mathsf{FAR}(\tau) = \frac{|\{(I_{x}, I_{y}) \in \mathcal{D}_{n}: s_{\cos}(F(I_{x}), F(I_{y})) > \tau \}|}{|\mathcal{D}_{n}|}
    \nonumber.
\end{align}
Here, $| \cdot |$ denotes the cardinality of the set and $s_{\cos}(\cdot, \cdot)$ denotes the cosine similarity between two vectors of same length. 

In addition, the accuracy at the threshold $\tau$ is defined as the ratio of correctly decided image pairs, \textit{i.e.}, over the whole dataset. That is,
\begin{align}
    \mathsf{ACC}(\tau) = \frac{|\{(I_{x}, I_{y}) \in \mathcal{D}_{p}: s_{\cos}(F(I_{x}), F(I_{y})) > \tau \}| + |\{(I_{x}, I_{y}) \in \mathcal{D}_{n}: s_{\cos}(F(I_{x}), F(I_{y})) < \tau \}|}{|\mathcal{D}_{p}|+|\mathcal{D}_{n}|} \nonumber.
\end{align}
We note that the accuracy is in fact one of the standard evaluation metrics for the LFW, CFP-FP, and AgeDB datasets. This is the reason why we provided the accuracy on the benchmark results.

On the other hand, for the identification scenario in IJB-C dataset~\cite{maze2018iarpa}, we used a different evaluation metric called true positive identification rate (TPIR) and false positive identification rate (FPIR). For evaluating the accuracy of face identification, we consider two sets of facial images called \textit{galleries}, say $\mathcal{G}_{1}$ and $\mathcal{G}_{2}$, whose identities are non-overlap each other. In addition, we also consider a larger \textit{mixed} set $\mathcal{G}_{\mathrm{mixed}}$ that contains all identities in $\mathcal{G}_{1}$ and $\mathcal{G}_{2}$. 
We divide $\mathcal{G}_{\mathrm{mixed}}$ to two subsets, $\mathcal{G}_{\mathrm{mixed}}^{1}$ and $\mathcal{G}_{\mathrm{mixed}}^{2}$, with respect to the identity.
If identities in $\mathcal{G}_{1}$ are enrolled to the database, the accuracy evaluation were done by measuring (1) the ratio of the images in $\mathcal{G}$ that are well-identified, \textit{i.e.}, the identity corresponindg to the highest similarity score is identical to the queried one, 
and (2) the ratio of the images in $\mathcal{G}_{\mathrm{mixed}}^{2}$ that are recognized as beloning to the enrolled identities.
Each ratio corresponds to the TPIR and FPIR, respectively.

With the identification threshold $\tau$ for determining whether the queried templates belongs to the enrolled identities or not, we formally describe the TPIR and FPIR as follows:
\begin{gather}
    \mathsf{TPIR}(\tau) = \frac{| \{(I, id) \in \mathcal{G}_{\mathrm{mixed}}^{1}: (I^{\ast}, id^{\ast}) := \arg\max_{(I', id') \in \mathcal{G}_{1}} \{s_{\cos}(F(I), F(I')) \}; id = id' \land s_{\cos}(F(I), F(I')) > \tau  \} |}{|\mathcal{G}_{1}|}, \nonumber \\
    \mathsf{FPIR}(\tau) = \frac{| \{I \in \mathcal{G}_{\mathrm{mixed}}^{2}: \max_{(I', id') \in \mathcal{G}_{1}} \{s_{\cos}(F(I), F(I')) \} > \tau \} |}{|\mathcal{G}_{\mathrm{mixed}}^{2}|} \nonumber,
\end{gather}
assuming that the identities in $\mathcal{G}_{1}$ are enrolled. We can evaluate the TPIR and FPIR for $\mathcal{G}_{2}$ in the same manner.
Throughout this paper, we reported the average value of TPIR for both cases when $\mathcal{G}_{1}$ and $\mathcal{G}_{2}$ are enrolled, respectively.

\subsection{Statistics of the Benchmark Datasets}

We remark that all the benchmark datasets we used, including LFW, CFP-FP, AgeDB, and IJB-C datasets, are publicly available and already frequently utilizd in the literature of face recognition. For reproducibility, we provide the detailed statistics of each dataset in Table~\ref{tab:stat_benchmark}.

\begin{table}[h]
\centering
\begin{minipage}{.715\linewidth}
\resizebox{\linewidth}{!}{
\begin{tabular}{c|c|c|c|c}
    \hline
    Dataset & LFW & CFP-FP & AgeDB & IJB-C (V) \\ \hline
    Imgs/IDs & 13,233 / 5,749 & 7,000 / 500 & 16,488 / 568 & 138,836 + 11,779 + 10,040 / 3,531 \\ \hline
    T/F Pairs & 3,000 / 3,000 &	3,500 / 3,500 &	3,000 / 3,000 & 19,557 / 15,638,932
 \\ \hline
\end{tabular}
}
\end{minipage}
\quad
\begin{minipage}{.26\linewidth}
\resizebox{\linewidth}{!}{
\begin{tabular}{c|c}
    \hline
     Dataset & IJB-C (ID)  \\ \hline
     Imgs in G1/G2 & 5,588 / 6,011  \\ \hline
     IDs in G1/G2 & 1,772 / 1,759  \\ \hline \hline
     Mixed IDs / Imgs &  3,531 / 127,152 \\ \hline
\end{tabular}
}
\end{minipage}
\vspace{-3mm}
\caption{The statistics of various benchmark datasets. We note that IJB-C dataset consists of facial images (138,836), videos (11,779), and non-facial images (10,040). For the IJB-C identification dataset, we provide the statistics of each gallery set.}
% \vspace{-3mm}
\label{tab:stat_benchmark}
\end{table}

\subsection{Result on Various Face Recognition Models}
Our theoretical results on Proposition~\ref{prop_1_full} seem independent of the choice of face recognition model because we did not specify it at that moment. However, we note that the definition of $(\epsilon, \delta, \theta)$-isometry assumes the distribution of the unit feature vectors. Since the distribution of face features is quite far from uniform, to be honest, our theoretical analysis may or may not fit with reality.

\begin{figure}
    \centering
    \subfloat[ArcFace]{
    \resizebox{\linewidth}{!}{
    \begin{tabular}{c|c|c|cccccc}
\hline
\multirow{2}{*}{Dataset} & Metric & \multirow{2}{*}{Plain} & \multicolumn{6}{c}{$\mathsf{IDFace}$, Parameters: $(\alpha, \beta)$} \\ \cline{4-9} 
 & (FAR/FPIR) &  & \multicolumn{1}{c|}{(512, 512)} & \multicolumn{1}{c|}{(341, 341)} & \multicolumn{1}{c|}{(127, 127)} & \multicolumn{1}{c|}{(63, 63)} & \multicolumn{1}{c|}{(341, 127)} & (341, 63) \\ \hline
\multirow{2}{*}{LFW} & Accuracy & 99.83\% & \multicolumn{1}{c|}{99.87\%} & \multicolumn{1}{c|}{99.80\%} & \multicolumn{1}{c|}{99.78\%} & \multicolumn{1}{c|}{99.67\%} & \multicolumn{1}{c|}{99.79\%} & 99.72\% \\
 & TAR@FAR & 99.73\%@0.03\% & \multicolumn{1}{c|}{99.77\%@0.03\%} & \multicolumn{1}{c|}{99.77\%@0.07\%} & \multicolumn{1}{c|}{99.63\%@0.03\%} & \multicolumn{1}{c|}{99.57\%@0.13\%} & \multicolumn{1}{c|}{99.67\%@0.02\%} & 99.57\%@0.05\% \\ \hline
\multirow{2}{*}{CFP-FP} & Accuracy & 99.04\% & \multicolumn{1}{c|}{98.57\%} & \multicolumn{1}{c|}{98.83\%} & \multicolumn{1}{c|}{98.76\%} & \multicolumn{1}{c|}{97.66\%} & \multicolumn{1}{c|}{98.79\%} & 98.46\% \\
 & TAR@FAR & 98.40\%@0.14\% & \multicolumn{1}{c|}{97.83\%@0.63\%} & \multicolumn{1}{c|}{98.37\%@0.54\%} & \multicolumn{1}{c|}{97.89\%@0.26\%} & \multicolumn{1}{c|}{96.83\%@1.51\%} & \multicolumn{1}{c|}{98.16\%@0.43\%} & 97.77\%@0.73\%  \\ \hline
\multirow{2}{*}{AGE-DB} & Accuracy & 98.38\% & \multicolumn{1}{c|}{98.02\%} & \multicolumn{1}{c|}{98.35\%} & \multicolumn{1}{c|}{98.02\%} & \multicolumn{1}{c|}{97.07\%} & \multicolumn{1}{c|}{97.98\%} & 97.91\% \\
 & TAR@FAR & 96.97\%@0.20\% & \multicolumn{1}{c|}{96.47\%@0.37\%} & \multicolumn{1}{c|}{97.13\%@0.43\%} & \multicolumn{1}{c|}{96.60\%@0.57\%} & \multicolumn{1}{c|}{95.83\%@1.70\%} & \multicolumn{1}{c|}{96.53\%@0.37\%} & 96.52\%@0.67\% \\ \hline
IJB-C(V) & TAR (1e-3) & 98.02\% & \multicolumn{1}{c|}{97.27\%} & \multicolumn{1}{c|}{97.71\%} & \multicolumn{1}{c|}{97.34\%} & \multicolumn{1}{c|}{95.95\%} & \multicolumn{1}{c|}{97.50\%} & 97.12\% \\ \hline
\multirow{3}{*}{IJB-C(ID)} & TPIR(1e-2) & 95.35\% & \multicolumn{1}{c|}{94.23\%} & \multicolumn{1}{c|}{95.04\%} & \multicolumn{1}{c|}{94.00\%} & \multicolumn{1}{c|}{91.16\%} & \multicolumn{1}{c|}{94.59\%} & 93.71\% \\
 & TPIR(1e-3) & 88.96\% & \multicolumn{1}{c|}{87.80\%} & \multicolumn{1}{c|}{89.18\%} & \multicolumn{1}{c|}{87.58\%} & \multicolumn{1}{c|}{83.85\%} & \multicolumn{1}{c|}{87.12\%} & 87.14\% \\
 & TPIR(1e-4) & 64.71\% & \multicolumn{1}{c|}{70.69\%} & \multicolumn{1}{c|}{66.27\%} & \multicolumn{1}{c|}{54.71\%} & \multicolumn{1}{c|}{52.91\%} & \multicolumn{1}{c|}{61.34\%} & 58.96\% \\ \hline
\end{tabular}}
    }\\
    \subfloat[MagFace]{
    \resizebox{\linewidth}{!}{
    \begin{tabular}{c|c|c|cccccc}
\hline
\multirow{2}{*}{Dataset} & Metric & \multirow{2}{*}{Plain} & \multicolumn{6}{c}{$\mathsf{IDFace}$, Parameters: $(\alpha, \beta)$} \\ \cline{4-9} 
 & (FAR/FPIR) &  & \multicolumn{1}{c|}{(512, 512)} & \multicolumn{1}{c|}{(341, 341)} & \multicolumn{1}{c|}{(127, 127)} & \multicolumn{1}{c|}{(63, 63)} & \multicolumn{1}{c|}{(341, 127)} & (341, 63) \\ \hline
\multirow{2}{*}{LFW} & Accuracy & 99.80\% & \multicolumn{1}{c|}{99.63\%} & \multicolumn{1}{c|}{99.75\%} & \multicolumn{1}{c|}{99.65\%} & \multicolumn{1}{c|}{99.55\%} & \multicolumn{1}{c|}{99.63\%} & 99.68\% \\
 & TAR@FAR & 99.67\%@0.03\% & \multicolumn{1}{c|}{99.43\%@0.07\%} & \multicolumn{1}{c|}{99.63\%@0.10\%} & \multicolumn{1}{c|}{99.40\%@0.10\%} & \multicolumn{1}{c|}{99.30\%@0.20\%} & \multicolumn{1}{c|}{99.43\%@0.02\%} & 99.45\%@0.05\% \\ \hline
\multirow{2}{*}{CFP-FP} & Accuracy & 98.91\% & \multicolumn{1}{c|}{98.27\%} & \multicolumn{1}{c|}{98.43\%} & \multicolumn{1}{c|}{98.37\%} & \multicolumn{1}{c|}{97.21\%} & \multicolumn{1}{c|}{98.53\%} & 98.21\% \\
 & TAR@FAR & 98.26\%@0.26\% & \multicolumn{1}{c|}{97.34\%@0.80\%} & \multicolumn{1}{c|}{97.69\%@0.60\%} & \multicolumn{1}{c|}{97.26\%@0.51\%} & \multicolumn{1}{c|}{95.31\%@0.89\%} & \multicolumn{1}{c|}{97.53\%@0.41\%} & 97.29\%@0.79\% \\ \hline
\multirow{2}{*}{AGE-DB} & Accuracy & 98.15\% & \multicolumn{1}{c|}{97.77\%} & \multicolumn{1}{c|}{98.15\%} & \multicolumn{1}{c|}{97.78\%} & \multicolumn{1}{c|}{96.37\%} & \multicolumn{1}{c|}{97.92\%} & 97.53\% \\
 & TAR@FAR & 97.03\%@0.50\% & \multicolumn{1}{c|}{95.90\%@0.27\%} & \multicolumn{1}{c|}{97.17\%@0.70\%} & \multicolumn{1}{c|}{96.43\%@0.80\%} & \multicolumn{1}{c|}{95.20\%@2.07\%} & \multicolumn{1}{c|}{96.25\%@0.37\%} & 96.12\% @0.90\% \\ \hline
IJB-C(V) & TAR (1e-3) & 97.04\% & \multicolumn{1}{c|}{96.14\%} & \multicolumn{1}{c|}{96.66\%} & \multicolumn{1}{c|}{96.27\%} & \multicolumn{1}{c|}{94.53\%} & \multicolumn{1}{c|}{96.50\%} & 95.91\% \\ \hline
\multirow{3}{*}{IJB-C(ID)} & TPIR(1e-2) & 94.00\% & \multicolumn{1}{c|}{91.94\%} & \multicolumn{1}{c|}{93.34\%} & \multicolumn{1}{c|}{92.19\%} & \multicolumn{1}{c|}{89.01\%} & \multicolumn{1}{c|}{92.95\%} & 91.95\% \\
 & TPIR(1e-3) & 85.63\% & \multicolumn{1}{c|}{85.00\%} & \multicolumn{1}{c|}{86.73\%} & \multicolumn{1}{c|}{85.32\%} & \multicolumn{1}{c|}{81.32\%} & \multicolumn{1}{c|}{84.24\%} & 84.49\% \\
 & TPIR(1e-4) & 58.34\% & \multicolumn{1}{c|}{57.25\%} & \multicolumn{1}{c|}{55.13\%} & \multicolumn{1}{c|}{51.85\%} & \multicolumn{1}{c|}{48.18\%} & \multicolumn{1}{c|}{55.63\%} & 64.60\% \\ \hline
\end{tabular}}
    }\\
    \subfloat[SphereFace2]{
    \resizebox{\linewidth}{!}{
    \begin{tabular}{c|c|c|cccccc}
\hline
\multirow{2}{*}{Dataset} & Metric & \multirow{2}{*}{Plain} & \multicolumn{6}{c}{$\mathsf{IDFace}$, Parameters: $(\alpha, \beta)$} \\ \cline{4-9} 
 & FAR(FPIR) &  & \multicolumn{1}{c|}{(512, 512)} & \multicolumn{1}{c|}{(341, 341)} & \multicolumn{1}{c|}{(127, 127)} & \multicolumn{1}{c|}{(63, 63)} & \multicolumn{1}{c|}{(341, 127)} & (341, 63) \\ \hline
\multirow{2}{*}{LFW} & Accuracy & 99.78\% & \multicolumn{1}{c|}{99.75\%} & \multicolumn{1}{c|}{99.78\%} & \multicolumn{1}{c|}{99.78\%} & \multicolumn{1}{c|}{99.72\%} & \multicolumn{1}{c|}{99.81\%} & 99.79\% \\
 & TAR@FAR & 99.67\%@0.00\% & \multicolumn{1}{c|}{99.67\%@0.10\%} & \multicolumn{1}{c|}{99.73\%@0.07\%} & \multicolumn{1}{c|}{99.67\%@0.03\%} & \multicolumn{1}{c|}{99.47\%@0.03\%} & \multicolumn{1}{c|}{99.70\%@0.00\%} & 99.70\%@0.03\% \\ \hline
\multirow{2}{*}{CFP-FP} & Accuracy & 99.06\% & \multicolumn{1}{c|}{98.67\%} & \multicolumn{1}{c|}{99.07\%} & \multicolumn{1}{c|}{98.70\%} & \multicolumn{1}{c|}{97.83\%} & \multicolumn{1}{c|}{98.96\%} & 98.63\% \\
 & TAR@FAR & 98.49\%@0.29\% & \multicolumn{1}{c|}{97.89\%@0.46\%} & \multicolumn{1}{c|}{98.34\%@0.14\%} & \multicolumn{1}{c|}{98.26\%@0.71\%} & \multicolumn{1}{c|}{96.29\%@0.63\%} & \multicolumn{1}{c|}{98.19\%@0.23\%} & 97.87\%@0.54\% \\ \hline
\multirow{2}{*}{AGE-DB} & Accuracy & 98.55\% & \multicolumn{1}{c|}{98.12\%} & \multicolumn{1}{c|}{98.45\%} & \multicolumn{1}{c|}{98.18\%} & \multicolumn{1}{c|}{97.32\%} & \multicolumn{1}{c|}{98.34\%} & 97.90\% \\
 & TAR@FAR & 97.80\%@0.53\% & \multicolumn{1}{c|}{96.80\%@0.40\%} & \multicolumn{1}{c|}{97.10\%@0.20\%} & \multicolumn{1}{c|}{96.97\%@0.43\%} & \multicolumn{1}{c|}{95.87\%@0.43\%} & \multicolumn{1}{c|}{97.35\%@0.55\%} & 96.65\%@0.68\% \\ \hline
IJB-C(V) & TAR (1e-3) & 97.98\% & \multicolumn{1}{c|}{97.54\%} & \multicolumn{1}{c|}{97.83\%} & \multicolumn{1}{c|}{97.55\%} & \multicolumn{1}{c|}{96.76\%} & \multicolumn{1}{c|}{97.68\%} & 97.41\% \\ \hline
\multirow{3}{*}{IJB-C(ID)} & TPIR(1e-2) & 95.06\% & \multicolumn{1}{c|}{94.08\%} & \multicolumn{1}{c|}{94.73\%} & \multicolumn{1}{c|}{93.92\%} & \multicolumn{1}{c|}{91.16\%} & \multicolumn{1}{c|}{94.35\%} & 93.75\% \\
 & TPIR(1e-3) & 88.40\% & \multicolumn{1}{c|}{85.83\%} & \multicolumn{1}{c|}{86.55\%} & \multicolumn{1}{c|}{87.42\%} & \multicolumn{1}{c|}{86.53\%} & \multicolumn{1}{c|}{87.51\%} & 86.06\% \\
 & TPIR(1e-4) & 52.04\% & \multicolumn{1}{c|}{50.13\%} & \multicolumn{1}{c|}{48.43\%} & \multicolumn{1}{c|}{53.79\%} & \multicolumn{1}{c|}{47.55\%} & \multicolumn{1}{c|}{50.75\%} & 55.13\% \\ \hline
\end{tabular}}
    }\\
    \subfloat[ElasticFace]{
    \resizebox{\linewidth}{!}{
    \begin{tabular}{c|c|c|cccccc}
\hline
\multirow{2}{*}{Dataset} & Metric & \multirow{2}{*}{Plain} & \multicolumn{6}{c}{$\mathsf{IDFace}$, Parameters: $(\alpha, \beta)$} \\ \cline{4-9} 
 & FAR(FPIR) &  & \multicolumn{1}{c|}{(512, 512)} & \multicolumn{1}{c|}{(341, 341)} & \multicolumn{1}{c|}{(127, 127)} & \multicolumn{1}{c|}{(63, 63)} & \multicolumn{1}{c|}{(341, 127)} & (341, 63) \\ \hline
\multirow{2}{*}{LFW} & Accuracy & 99.82\% & \multicolumn{1}{c|}{99.80\%} & \multicolumn{1}{c|}{99.82\%} & \multicolumn{1}{c|}{99.83\%} & \multicolumn{1}{c|}{99.65\%} & \multicolumn{1}{c|}{99.82\%} & 99.80\% \\
 & TAR@FAR & 99.67\%@0.00\% & \multicolumn{1}{c|}{99.63\%@0.00\%} & \multicolumn{1}{c|}{99.67\%@0.00\%} & \multicolumn{1}{c|}{99.67\%@0.00\%} & \multicolumn{1}{c|}{99.67\%@0.17\%} & \multicolumn{1}{c|}{99.70\%@0.03\%} & 99.67\%@0.00\% \\ \hline
\multirow{2}{*}{CFP-FP} & Accuracy & 99.09\% & \multicolumn{1}{c|}{98.67\%} & \multicolumn{1}{c|}{99.01\%} & \multicolumn{1}{c|}{98.81\%} & \multicolumn{1}{c|}{98.20\%} & \multicolumn{1}{c|}{99.03\%} & 98.79\% \\
 & TAR@FAR & 98.49\%@0.17\% & \multicolumn{1}{c|}{97.86\%@0.34\%} & \multicolumn{1}{c|}{98.26\%@0.17\%} & \multicolumn{1}{c|}{98.31\%@0.66\%} & \multicolumn{1}{c|}{97.46\%@1.06\%} & \multicolumn{1}{c|}{98.41\%@0.29\%} & 97.99\%@0.36\% \\ \hline
\multirow{2}{*}{AGE-DB} & Accuracy & 98.45\% & \multicolumn{1}{c|}{98.00\%} & \multicolumn{1}{c|}{98.23\%} & \multicolumn{1}{c|}{98.25\%} & \multicolumn{1}{c|}{97.53\%} & \multicolumn{1}{c|}{98.26\%} & 98.10\% \\
 & TAR@FAR & 97.33\%@0.23\% & \multicolumn{1}{c|}{96.90\%@0.70\%} & \multicolumn{1}{c|}{96.93\%@0.17\%} & \multicolumn{1}{c|}{97.30\%@0.80\%} & \multicolumn{1}{c|}{95.83\%@0.77\%} & \multicolumn{1}{c|}{96.98\%@0.32\%} & 97.12\%@0.85\% \\\hline
IJB-C(V) & TAR (1e-3) & 97.93\% & \multicolumn{1}{c|}{97.31\%} & \multicolumn{1}{c|}{97.68\%} & \multicolumn{1}{c|}{97.48\%} & \multicolumn{1}{c|}{96.36\%} & \multicolumn{1}{c|}{97.56\%} & 97.20\% \\ \hline
\multirow{3}{*}{IJB-C(ID)} & TPIR(1e-2) & 94.99\% & \multicolumn{1}{c|}{93.83\%} & \multicolumn{1}{c|}{94.51\%} & \multicolumn{1}{c|}{93.76\%} & \multicolumn{1}{c|}{91.17\%} & \multicolumn{1}{c|}{94.42\%} & 93.60\% \\
 & TPIR(1e-3) & 89.56\% & \multicolumn{1}{c|}{87.83\%} & \multicolumn{1}{c|}{88.45\%} & \multicolumn{1}{c|}{87.45\%} & \multicolumn{1}{c|}{82.90\%} & \multicolumn{1}{c|}{88.54\%} & 87.54\% \\
 & TPIR(1e-4) & 60.18\% & \multicolumn{1}{c|}{60.45\%} & \multicolumn{1}{c|}{59.88\%} & \multicolumn{1}{c|}{62.10\%} & \multicolumn{1}{c|}{45.43\%} & \multicolumn{1}{c|}{62.55\%} & 56.36\% \\ \hline
\end{tabular}}
    }
    \vspace{-3mm}
    \caption{Various Face Recognition Benchmark results on non-protected template extractor (Plain) and $\mathsf{IDFace}$ with various $(\alpha, \beta)$.}
    \label{fig:total_table1}
    \vspace{-3mm}
\end{figure}

For this reason, we conducted extensive experiments to evaluate the effectiveness of $\mathsf{IDFace}$ on other recently proposed face recognition models. We select ElasticFace~\cite{boutros2022elasticface}, MagFace~\cite{meng2021magface}, SphereFace2~\cite{wen2022sphereface2}, AdaFace~\cite{kim2022adaface}, and AdaFace-KPRPE~\cite{kim2024keypoint}. We implemented the first three FR models by following their official source codes (\cite{magface}, \cite{elasticface}, and \cite{opensphere} for ElasticFace, MagFace, and SphereFace2, respectively). Here, we used ResNet100 as a backbone and trained each model with the MS1MV3 dataset.
In addition, for AdaFace and AdaFace-KPRPE, we exploited the pre-trained models provided in the CVLFace library~\cite{CVLface} that is maintained by the authors of AdaFace. For each model, we used \texttt{IR101} and \texttt{ViT-KPRPE} architectures trained by the WebFace12M dataset~\cite{zhu2021webface260m}, respectively.
The benchmark results are given in Figure~\ref{fig:total_table1},~\ref{fig:total_table2}, which indicates that our $\mathsf{IDFace}$ can harmonize well with various recognition systems without significant accuracy degradation.

\begin{figure}
    \centering
    \subfloat[AdaFace-IR101]{
    \resizebox{\linewidth}{!}{
    \begin{tabular}{c|c|c|cccccc}
\hline
\multirow{2}{*}{Dataset} & Metric & \multirow{2}{*}{Plain} & \multicolumn{6}{c}{$\mathsf{IDFace}$, Parameters: $(\alpha, \beta)$} \\ \cline{4-9} 
 & FAR(FPIR) &  & \multicolumn{1}{c|}{(512, 512)} & \multicolumn{1}{c|}{(341, 341)} & \multicolumn{1}{c|}{(127, 127)} & \multicolumn{1}{c|}{(63, 63)} & \multicolumn{1}{c|}{(341, 127)} & (341, 63) \\ \hline
\multirow{2}{*}{LFW} & Accuracy & 99.82\% & \multicolumn{1}{c|}{99.77\%} & \multicolumn{1}{c|}{99.78\%} & \multicolumn{1}{c|}{99.83\%} & \multicolumn{1}{c|}{99.73\%} & \multicolumn{1}{c|}{99.80\%} & 99.82\% \\
 & TAR@FAR & 99.67\%@0.03\% & \multicolumn{1}{c|}{99.57\%@0.03\%} & \multicolumn{1}{c|}{99.63\%@0.07\%} & \multicolumn{1}{c|}{99.67\%@0.00\%} & \multicolumn{1}{c|}{99.57\%@0.10\%} & \multicolumn{1}{c|}{99.65\%@0.05\%} & 99.68\%@0.03\% \\ \hline
\multirow{2}{*}{CFP-FP} & Accuracy & 99.24\% & \multicolumn{1}{c|}{98.96\%} & \multicolumn{1}{c|}{99.24\%} & \multicolumn{1}{c|}{99.09\%} & \multicolumn{1}{c|}{98.20\%} & \multicolumn{1}{c|}{99.19\%} & 98.99\% \\
 & TAR@FAR & 98.66\%@0.17\% & \multicolumn{1}{c|}{98.14\%@0.23\%} & \multicolumn{1}{c|}{98.71\%@0.23\%} & \multicolumn{1}{c|}{98.51\%@0.34\%} & \multicolumn{1}{c|}{96.94\%@0.54\%} & \multicolumn{1}{c|}{98.53\%@0.16\%} & 98.24\%@0.26\% \\ \hline
\multirow{2}{*}{AGE-DB} & Accuracy & 98.00\% & \multicolumn{1}{c|}{97.67\%} & \multicolumn{1}{c|}{97.97\%} & \multicolumn{1}{c|}{97.55\%} & \multicolumn{1}{c|}{96.48\%} & \multicolumn{1}{c|}{97.80\%} & 97.23\% \\
 & TAR@FAR & 96.73\%@0.73\% & \multicolumn{1}{c|}{96.37\%@1.03\%} & \multicolumn{1}{c|}{96.50\%@0.57\%} & \multicolumn{1}{c|}{96.23\%@1.13\%} & \multicolumn{1}{c|}{94.17\%@1.20\%} & \multicolumn{1}{c|}{96.65\%@1.05\%} & 95.62\%@1.15\% \\\hline
IJB-C(V) & TAR (1e-3) & 98.39\% & \multicolumn{1}{c|}{97.93\%} & \multicolumn{1}{c|}{98.27\%} & \multicolumn{1}{c|}{97.90\%} & \multicolumn{1}{c|}{96.94\%} & \multicolumn{1}{c|}{98.14\%} & 97.78\% \\ \hline
\multirow{3}{*}{IJB-C(ID)} & TPIR(1e-2) & 96.42\% & \multicolumn{1}{c|}{95.15\%} & \multicolumn{1}{c|}{95.95\%} & \multicolumn{1}{c|}{95.32\%} & \multicolumn{1}{c|}{92.87\%} & \multicolumn{1}{c|}{95.72\%} & 95.19\% \\
 & TPIR(1e-3) & 85.81\% & \multicolumn{1}{c|}{85.83\%} & \multicolumn{1}{c|}{85.11\%} & \multicolumn{1}{c|}{85.40\%} & \multicolumn{1}{c|}{86.21\%} & \multicolumn{1}{c|}{86.59\%} & 85.96\% \\
 & TPIR(1e-4) & 65.60\% & \multicolumn{1}{c|}{63.88\%} & \multicolumn{1}{c|}{68.38\%} & \multicolumn{1}{c|}{62.38\%} & \multicolumn{1}{c|}{61.39\%} & \multicolumn{1}{c|}{63.45\%} & 64.66\% \\ \hline
\end{tabular}}
    }\\
    \subfloat[AdaFace-KPRPE]{
    \resizebox{\linewidth}{!}{
    \begin{tabular}{c|c|c|cccccc}
\hline
\multirow{2}{*}{Dataset} & Metric & \multirow{2}{*}{Plain} & \multicolumn{6}{c}{$\mathsf{IDFace}$, Parameters: $(\alpha, \beta)$} \\ \cline{4-9} 
 & FAR(FPIR) &  & \multicolumn{1}{c|}{(512, 512)} & \multicolumn{1}{c|}{(341, 341)} & \multicolumn{1}{c|}{(127, 127)} & \multicolumn{1}{c|}{(63, 63)} & \multicolumn{1}{c|}{(341, 127)} & (341, 63) \\ \hline
\multirow{2}{*}{LFW} & Accuracy & 99.82\% & \multicolumn{1}{c|}{99.70\%} & \multicolumn{1}{c|}{99.82\%} & \multicolumn{1}{c|}{99.78\%} & \multicolumn{1}{c|}{99.70\%} & \multicolumn{1}{c|}{99.81\%} & 99.78\% \\
 & TAR@FAR & 99.63\%@0.00\% & \multicolumn{1}{c|}{99.53\%@0.13\%} & \multicolumn{1}{c|}{99.63\%@0.00\%} & \multicolumn{1}{c|}{99.60\%@0.03\%} & \multicolumn{1}{c|}{99.53\%@0.13\%} & \multicolumn{1}{c|}{99.63\%@0.02\%} & 99.60\%@0.03\% \\ \hline
\multirow{2}{*}{CFP-FP} & Accuracy & 99.30\% & \multicolumn{1}{c|}{98.77\%} & \multicolumn{1}{c|}{99.11\%} & \multicolumn{1}{c|}{99.17\%} & \multicolumn{1}{c|}{98.19\%} & \multicolumn{1}{c|}{99.11\%} & 98.78\% \\
 & TAR@FAR & 98.66\%@0.06\% & \multicolumn{1}{c|}{97.91\%@0.37\%} & \multicolumn{1}{c|}{98.51\%@0.29\%} & \multicolumn{1}{c|}{98.51\%@0.17\%} & \multicolumn{1}{c|}{97.00\%@0.63\%} & \multicolumn{1}{c|}{98.37\%@0.16\%} & 98.04\%@0.49\% \\ \hline
\multirow{2}{*}{AGE-DB} & Accuracy & 98.10\% & \multicolumn{1}{c|}{97.30\%} & \multicolumn{1}{c|}{97.67\%} & \multicolumn{1}{c|}{97.65\%} & \multicolumn{1}{c|}{96.23\%} & \multicolumn{1}{c|}{97.55\%} & 97.32\% \\
 & TAR@FAR & 97.10\%@0.90\% & \multicolumn{1}{c|}{95.67\%@1.07\%} & \multicolumn{1}{c|}{96.77\%@1.43\%} & \multicolumn{1}{c|}{96.17\%@0.87\%} & \multicolumn{1}{c|}{94.97\%@2.50\%} & \multicolumn{1}{c|}{96.22\%@1.12\%} & 95.72\%@1.08\% \\\hline
IJB-C(V) & TAR (1e-3) & 98.40\% & \multicolumn{1}{c|}{97.92\%} & \multicolumn{1}{c|}{98.20\%} & \multicolumn{1}{c|}{97.98\%} & \multicolumn{1}{c|}{96.67\%} & \multicolumn{1}{c|}{98.07\%} & 97.77\% \\ \hline
\multirow{3}{*}{IJB-C(ID)} & TPIR(1e-2) & 96.46\% & \multicolumn{1}{c|}{95.27\%} & \multicolumn{1}{c|}{96.00\%} & \multicolumn{1}{c|}{95.35\%} & \multicolumn{1}{c|}{92.10\%} & \multicolumn{1}{c|}{95.74\%} & 94.80\% \\
 & TPIR(1e-3) & 90.52\% & \multicolumn{1}{c|}{90.55\%} & \multicolumn{1}{c|}{89.79\%} & \multicolumn{1}{c|}{87.45\%} & \multicolumn{1}{c|}{85.06\%} & \multicolumn{1}{c|}{89.54\%} & 86.84\% \\
 & TPIR(1e-4) & 60.87\% & \multicolumn{1}{c|}{57.59\%} & \multicolumn{1}{c|}{60.94\%} & \multicolumn{1}{c|}{62.14\%} & \multicolumn{1}{c|}{64.17\%} & \multicolumn{1}{c|}{65.31\%} & 61.61\% \\ \hline
\end{tabular}}
    }
    \vspace{-3mm}
    \caption{Various Face Recognition Benchmark results on non-protected template extractor (Plain) and $\mathsf{IDFace}$ with various $(\alpha, \beta)$.}
    \label{fig:total_table2}
    \vspace{-3mm}
\end{figure}

In particular, to visualize the trade-off relationship between TAR and FAR, we also evaluated the receiver operating characteristic (ROC) and detection error trade-off (DET) curves for IJB-C verification and identification benchmarks, respectively. We tested the $\mathsf{IDFace}$ with the same feature extractors in Figure~\ref{fig:total_table1},~\ref{fig:total_table2} on the parameter settings in $\alpha =341$ and $\beta = 341, 127, 63$.
The results are given in Figure~\ref{fig:roc_det1},~\ref{fig:roc_det2}. For ROC curves in the verification task, one may observe that a slight accuracy degradation when $\mathsf{IDFace}$ is applied. Nevertheless, in terms of area under the curve (AUC) score, the amount of accuracy drop seems insignificant. On the other hand, in the identification task which is our main interest, we can observe that the differences between DET curves from $\mathsf{IDFace}$ and the plain model seems negligible. These results indicate that our $\mathsf{IDFace}$ can be applied to various face recognition models and various levels of FAR, without significant accuracy degradation.

\begin{figure}
    \begin{subfigure}[b]{1\textwidth}
        \centering
        \includegraphics[width = \textwidth]{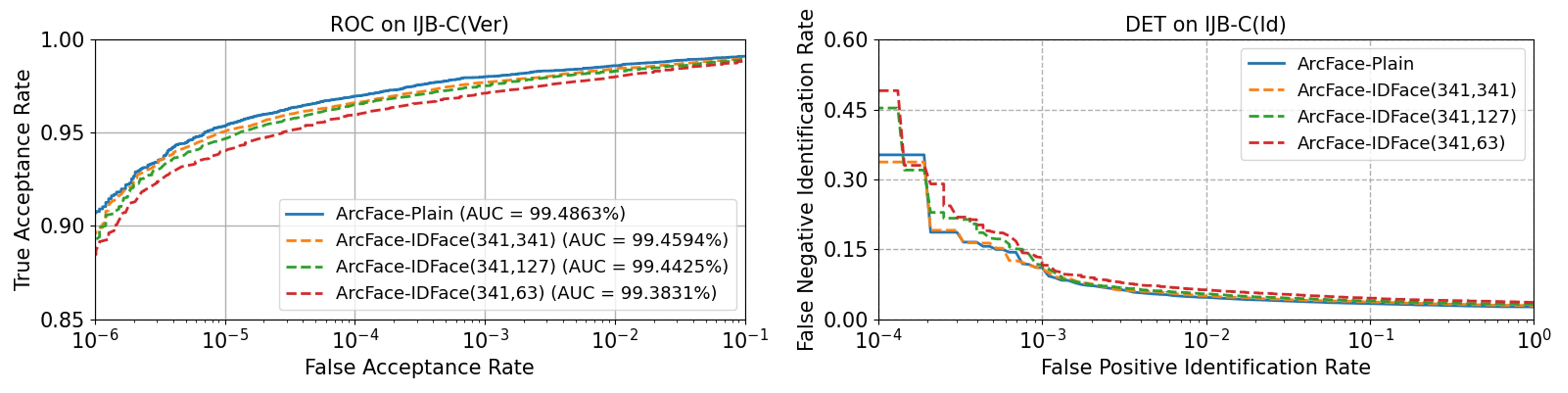}
        \caption{ArcFace}
    \end{subfigure}
        \begin{subfigure}[b]{1\textwidth}
        \centering
        \includegraphics[width = \textwidth]{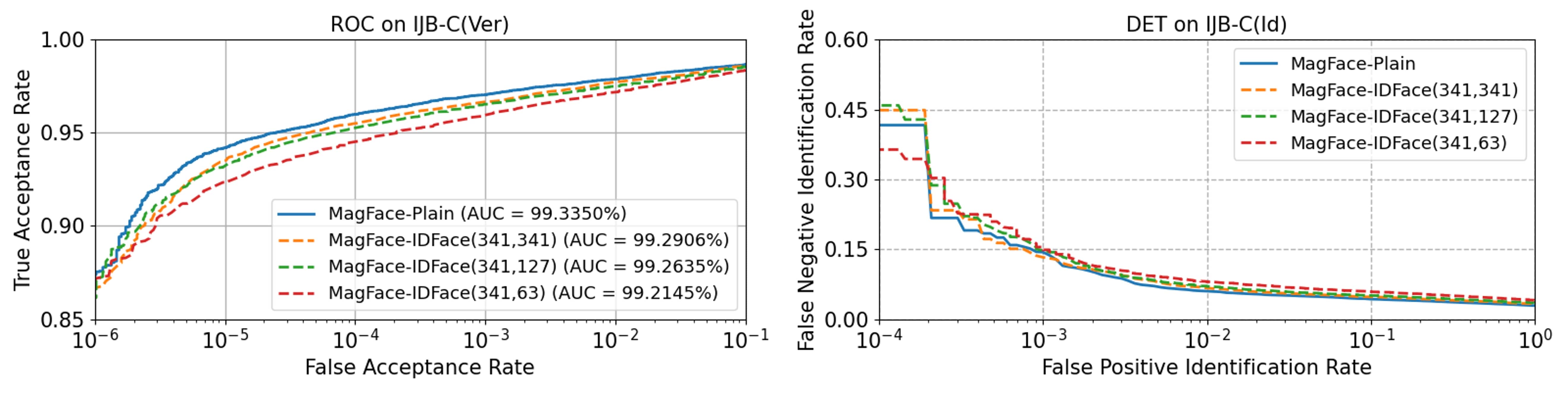}
        \caption{MagFace}
    \end{subfigure}
        \begin{subfigure}[b]{1\textwidth}
        \centering
        \includegraphics[width = \textwidth]{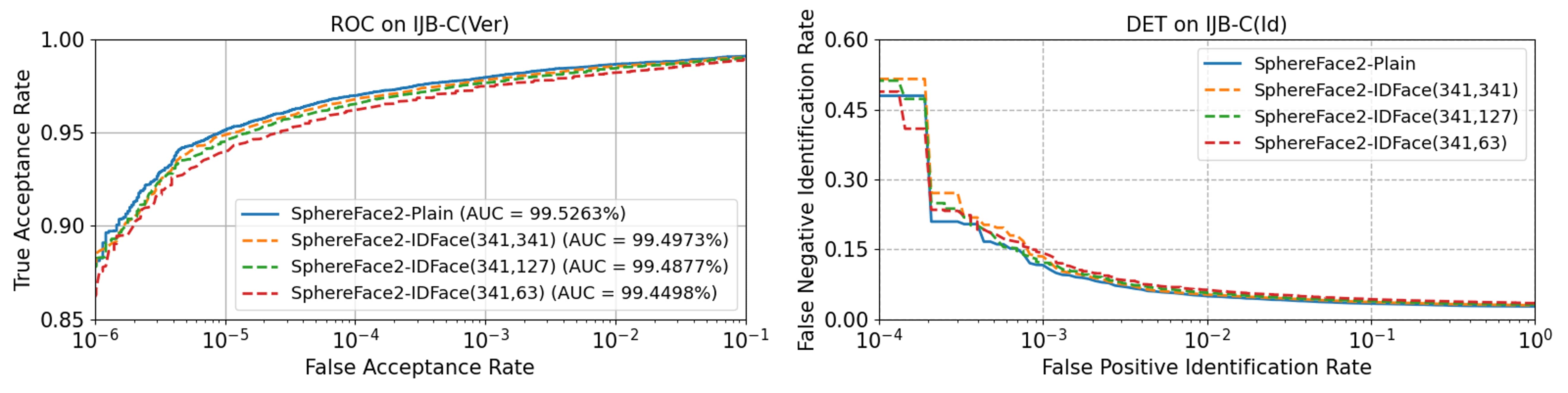}
        \caption{SphereFace2}
    \end{subfigure}
        \begin{subfigure}[b]{1\textwidth}
        \centering
        \includegraphics[width = \textwidth]{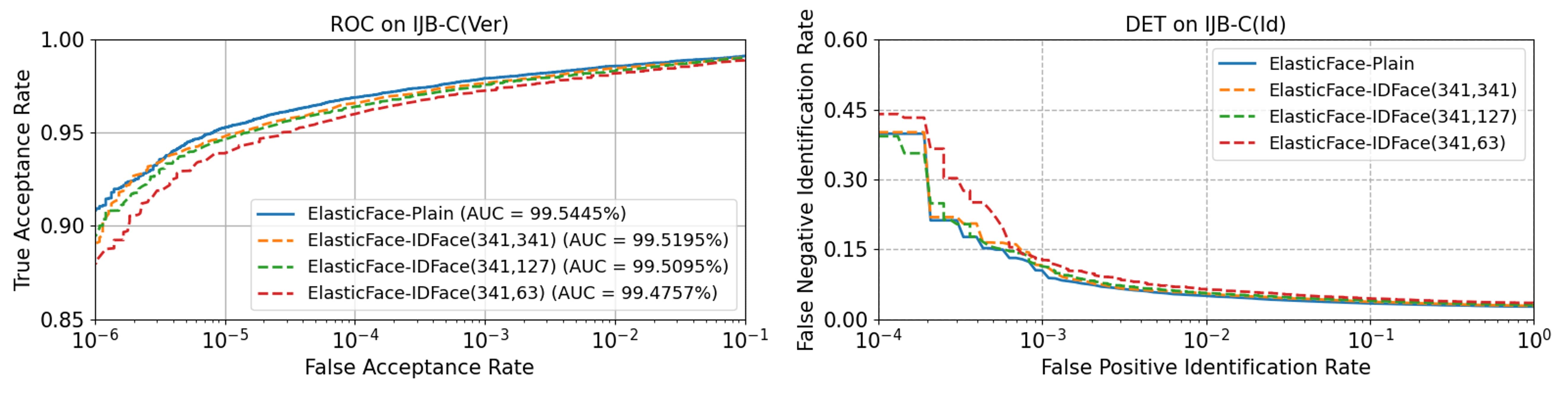}
        \caption{ElasticFace}
    \end{subfigure}
    \vspace{-3mm}    
    \caption{IJB-C verification/identification benchmark results for non-protected feature extractor (Plain) and $\mathsf{IDFace}$ with various $(\alpha, \beta)$. We also reported AUC-ROC curves and DET curves for each task, respectively.}
    \label{fig:roc_det1}
\end{figure}

\begin{figure}
    \begin{subfigure}[b]{1\textwidth}
        \centering
        \includegraphics[width = \textwidth]{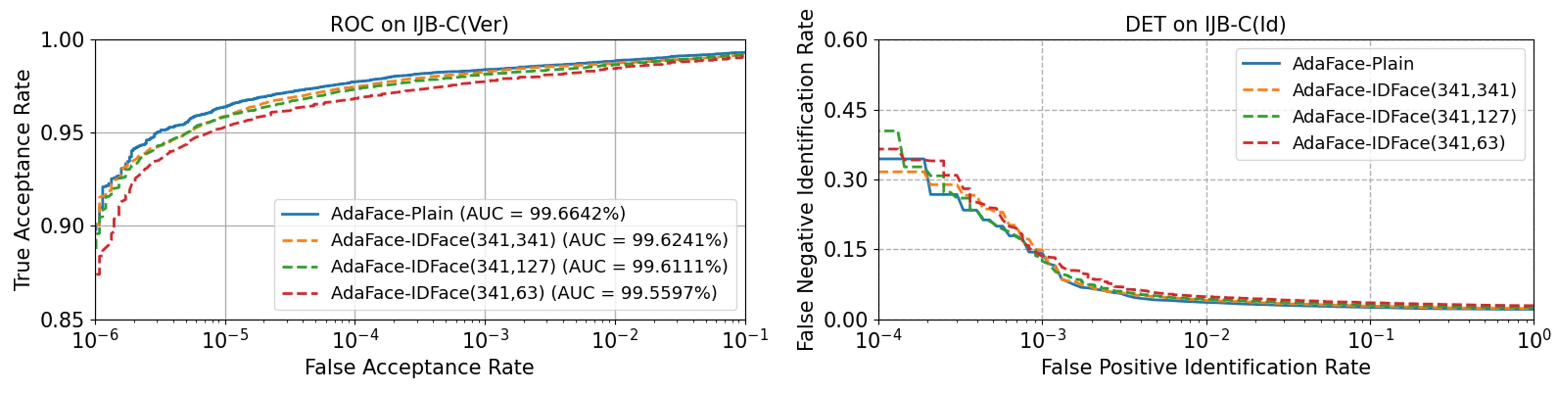}
        \caption{AdaFace-IR101}
    \end{subfigure}
        \begin{subfigure}[b]{1\textwidth}
        \centering
        \includegraphics[width = \textwidth]{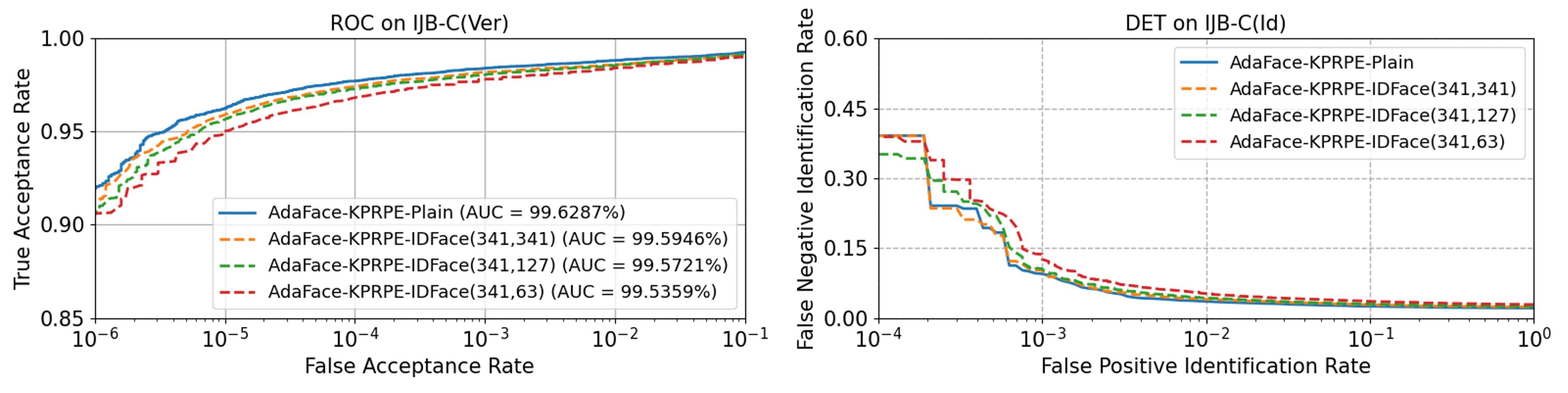}
        \caption{AdaFace-KPRPE}
    \end{subfigure}
    % \vspace{-3mm}
    \caption{IJB-C verification/identification benchmark results for non-protected feature extractor (Plain) and $\mathsf{IDFace}$ with various $(\alpha, \beta)$. We also reported AUC-ROC curves and DET curves for each task, respectively.}
    \label{fig:roc_det2}
\end{figure}

\subsection{Effect of our Transformation on Intra and Inter Class Variations}

\begin{figure*}[t]
    \centering
    \begin{subfigure}[b]{1\textwidth}
        \centering
        \includegraphics[width = \textwidth]{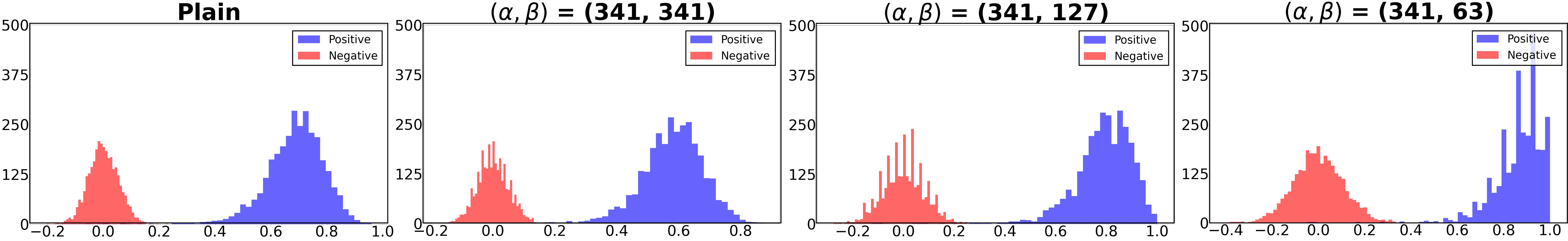}
        \caption{LFW}
    \end{subfigure}
    \vspace{1mm}
    \hfill
    \begin{subfigure}[b]{1\textwidth}
        \centering
        \includegraphics[width=\textwidth]{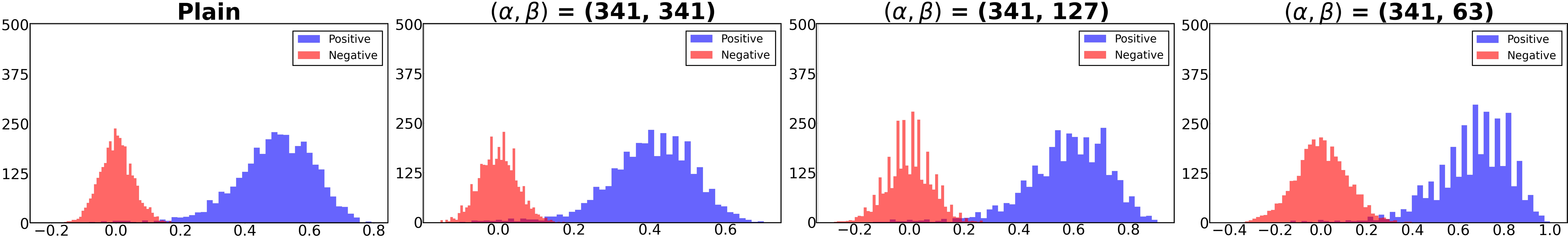}
        \caption{CFP-FP}
    \end{subfigure}
    \vspace{1mm}
    \hfill
    \begin{subfigure}[b]{1\textwidth}
        \centering
        \includegraphics[width=\textwidth]{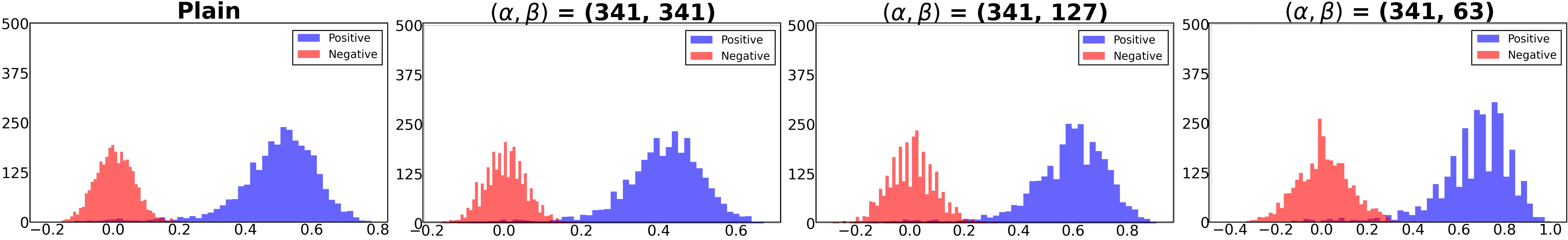}
        \caption{AgeDB}
    \end{subfigure}
    \vspace{-8mm}  
    \caption{The distribution of the rescaled cosine value from each positive, negative pairs on LFW, CFP-FP, AgeDB benchmark datasets. We applied our proposed transform with parameter $\alpha, \beta$, by varying $\beta = 63, 127, 341$ while $\alpha$ is fixed to $341$. We used ArcFace as a feature extractor.}
    \label{fig:alphabeta}
    \vspace{-4mm}
\end{figure*}

We also analyzed the effect of our transformation in terms of intra-class compactness and inter-class discrepancy on various transformation parameters. To this end, we measured the distribution of the cosine distance between positive and negative pairs on the LFW dataset, respectively. We note that when $\beta < \alpha$, the range of the inner product value between transformed vectors becomes narrower. Concretely, when we choose $(\alpha, \beta) = (341, 63)$, then the range of all possible cosine values is $\left[-\sqrt{\frac{63}{341}}, \sqrt{\frac{63}{341}} \approx 0.4298\right]$. 

To compensate for such a range mismatch, we multiply an appropriate constant to fit the range of the cosine value into $[-1,1]$. We note that such re-scaling itself does not affect the accuracy of the recognition system. The result of each distribution after re-scaling is provided in Figure~\ref{fig:alphabeta}.
We can figure out that the distribution of negative pairs is more widespread than the result from Plain due to the effect of re-scaling. However, the distributions of cosine values from positive pairs and negative pairs, respectively, are still well distinguished. This result explains the reason for the small accuracy degradation of $\mathsf{IDFace}$ even when $\beta < \alpha$.

\subsection{Comparison with PCA-based Dimensionality Reduction}\label{sec:supp_D_3}
Recall that lots of previous HE-based face template protections~\cite{boddeti2018secure, engelsma2022hers, Bausp2022improved} employed dimensionality reduction techniques for a good trade-off between efficiency and accuracy. However, as reported in their previous works, the amount of accuracy degradation from PCA~\cite{boddeti2018secure, Bausp2022improved} or neural network-based approach~\cite{engelsma2022hers} is at most $13.5\%$, which is quite large compared to that from $\mathsf{IDFace}$. Of course, the amount of degradation depends on the experimental setting, such as the accuracy of an unprotected face recognition model or the dimension of face templates before dimensionality reduction. Hence, for a fair comparison between $\mathsf{IDFace}$ and previous works in terms of accuracy degradation, we provide the result of accuracy degradation when PCA is applied to the templates extracted by the recognition model in our experimental setting. 

\begin{wraptable}{O}{.5\linewidth}
\centering
\vspace{-2mm}
\resizebox{\linewidth}{!}{
\begin{tabular}{c|c|c||ccc}
\hline
Dataset	& Metric (FAR) &	Plain&	128D&	64D&	32D \\ \hline
LFW& \multirow{3}{*}{Accuracy} & 	99.83\%&   	99.66\%&   	99.16\%&   	99.31\%
\\ \cline{1-1}
CFP-FP	&  &	99.04\%&	98.02\%&	96.15\%&	96.11\%
\\ \cline{1-1}
AgeDB&  &	98.38\%&	97.46\%&	95.90\%&	95.71\%
\\ \hline
IJB-C(V)&	TAR (1e-3)&	98.01\%&	96.43\%&	93.29\%&	93.27\%
\\ \hline
\multirow{5}{*}{IJB-C(ID)} & TPIR (1e-0)&	97.38\%&	95.77\%&	91.75\%&	91.72\%
\\ 
 &	TPIR (1e-1)&	96.68\%&	94.83\%&	89.32\%&	89.27\%
\\ 
&TPIR (1e-2)	&95.35\%&	93.40\%&	85.61\%&	85.48\%
\\ 
&	TPIR (1e-3)&	88.96\%&	86.45\%&	78.66\%&	79.23\%
\\ 
&	TPIR (1e-4)&	64.71\%&	60.04\%&	53.41\%&	53.17\%
\\ \hline
\end{tabular}
}
\vspace{-3mm}
\caption{Benchmark results of ArcFace on evaluation datasets with PCA on various dimensions.}
\vspace{-3mm}
\label{tab:pca}
\end{wraptable}

For training the PCA matrix, we exploited the same dataset (MS1MV3) that was used for training ArcFace. We selected the output dimensions as 32, 64, and 128 for comparison with the previous methods in Table \textcolor{iccvblue}{1} in main text. After training the PCA matrix, we evaluated the accuracy by running benchmarks as if the PCA matrix were added to the final layer of the feature extractor. The result is given in Table~\ref{tab:pca}, which tells us that the accuracy degradation from PCA is around $0.17\%$--$4.67\%$, $0.67\%$--$11.30\%$, and $0.52\% $--$11.54\%$ for 128D, 64D and 32D, respectively. We note that even in the fastest parameter setting in $\mathsf{IDFace}$ ($\beta=63$), the accuracy loss from $\mathsf{IDFace}$ is smaller than that from PCA with almost all benchmarks we tested. This result indicates that (1) the proposed $\mathsf{IDFace}$ is even better than previous methods with PCA of 128D in terms of not only efficiency but also accuracy, and (2) applying PCA to our $\mathsf{IDFace}$ would show a larger accuracy degradation than this result.

\subsection{Discussion on the Accuracy Drop}

Through experiments, we showed that the tendency of accuracy degradation with respect to the transformation parameters, giving intuition that the reason for accuracy degradation is strongly tied to the extent how the transformation preserves the distance relationship. However, we found that such an interpretation is still insufficient for fully explaining the accuracy degradation of $\mathsf{IDFace}$. For example, we can observe an odd phenomenon that was commonly observed regardless of the choice of the feature extractor: the accuracy increases after employing the $\mathsf{IDFace}$. To clarify these aspects, we further investigate how $\mathsf{IDFace}$, especially our almost isometric transformation, yields accuracy degradation.

\;

\noindent \textbf{On the Distance Preservation Property With Respect to Inputs.}
According to our analyses, the change in inner product value from ours is at most 0.111 under the suggested parameter. However, the upper bound 0.111 is too high to explain the small accuracy drop of $\mathsf{IDFace}$, so more analysis for our transformation is necessary.

In fact, the bound 0.111 is about the worst case, and the change of the inner product value is also affected by that value before transformation. As we can observe in Figure~\ref{fig:kirruk}, the difference in the inner product value becomes the largest when the inner product before the transformation is near 0.7 and becomes narrower when the inner product value of inputs tends to 0. We note that the best threshold for each recognition model we evaluated is near $0.15$--$0.25$ for all datasets. This implies that the difference in inner product value near the threshold is relatively small (less than 0.05) compared to the upper bound we derived. We can expect that the inner product values nearby the threshold are relatively less changed than the upper bound, where pairs of faces whose inner product value is nearby here are subject to changing the decision of the model.

It is also worthwhile to mention that the distribution of feature vectors is quite different from uniform, which is in fact a gap between theory and reality. More precisely, in the definition of $(\epsilon, \delta, \theta)$-isometry, we calculated $\epsilon$ by considering a uniformly sampled unit vector $x$ and another unit vector $y$ uniformly sampled from vectors whose angles from $x$ are $\theta$. Along with the ease of proof, we intended to construct a data-agnostic distance-preserving transformation via this definition; we note that if some part of the BTP relies on the specific distribution derived from the input dataset, then this BTP cannot be deployed in an open-set face recognition scenario. In fact, our experimental results show that (1) the accuracy degradation carried by $\mathsf{IDFace}$ is admissible in various face recognition models, and (2) the almost-isometric transformation would not harm both intra-class compactness and inter-class discrepancy. Nevertheless, from a theoretical perspective, there is still a gap that needs to be further investigated. We leave this as an interesting future work.

\;

\begin{figure}
    \begin{subfigure}[b]{1\textwidth}
        \centering
        \includegraphics[width = \textwidth]{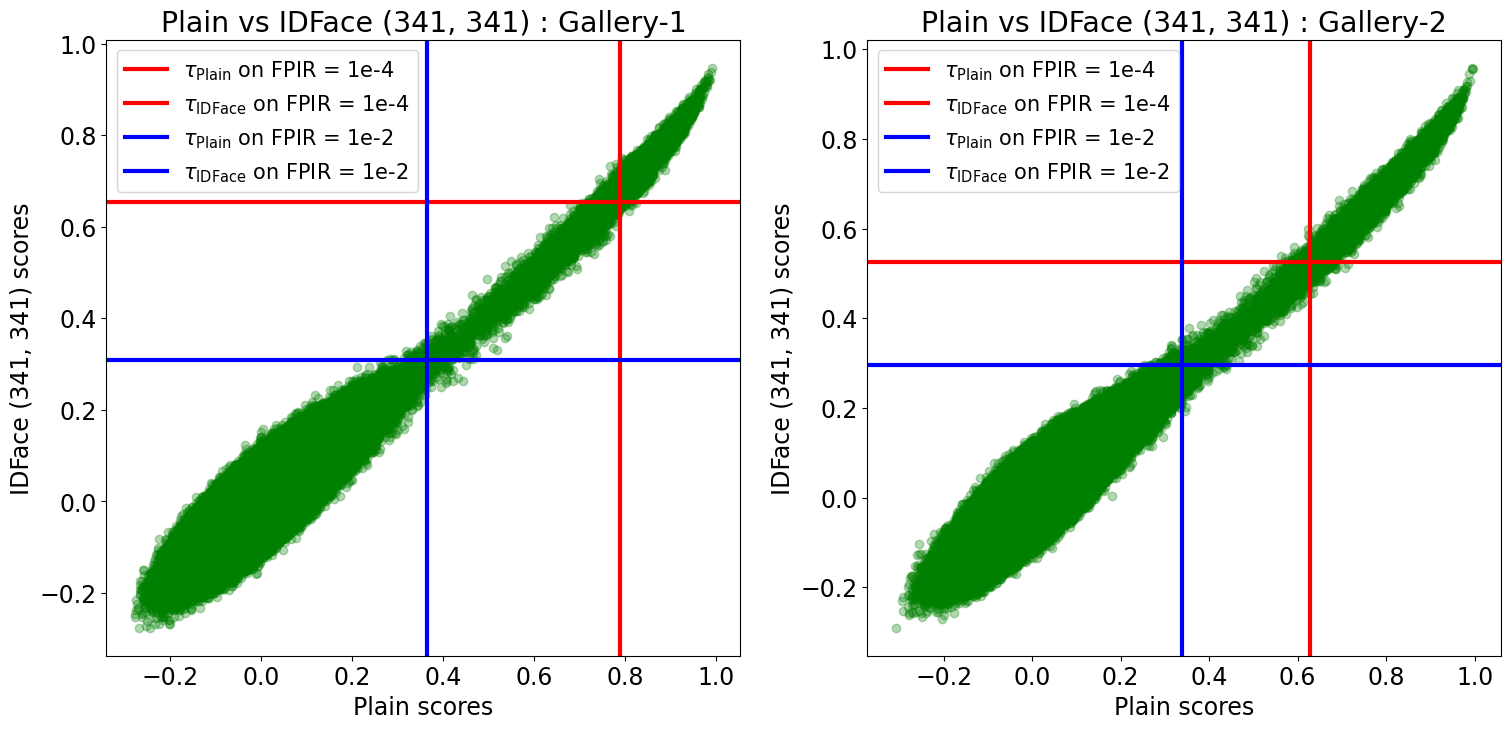}
        \caption{Plain vs. $\mathsf{IDFace}$ with $(\alpha,\beta) = (341,341)$.}\label{fig:odd_phenomenon_a}
    \end{subfigure}
        \begin{subfigure}[b]{1\textwidth}
        \centering
        \includegraphics[width = \textwidth]{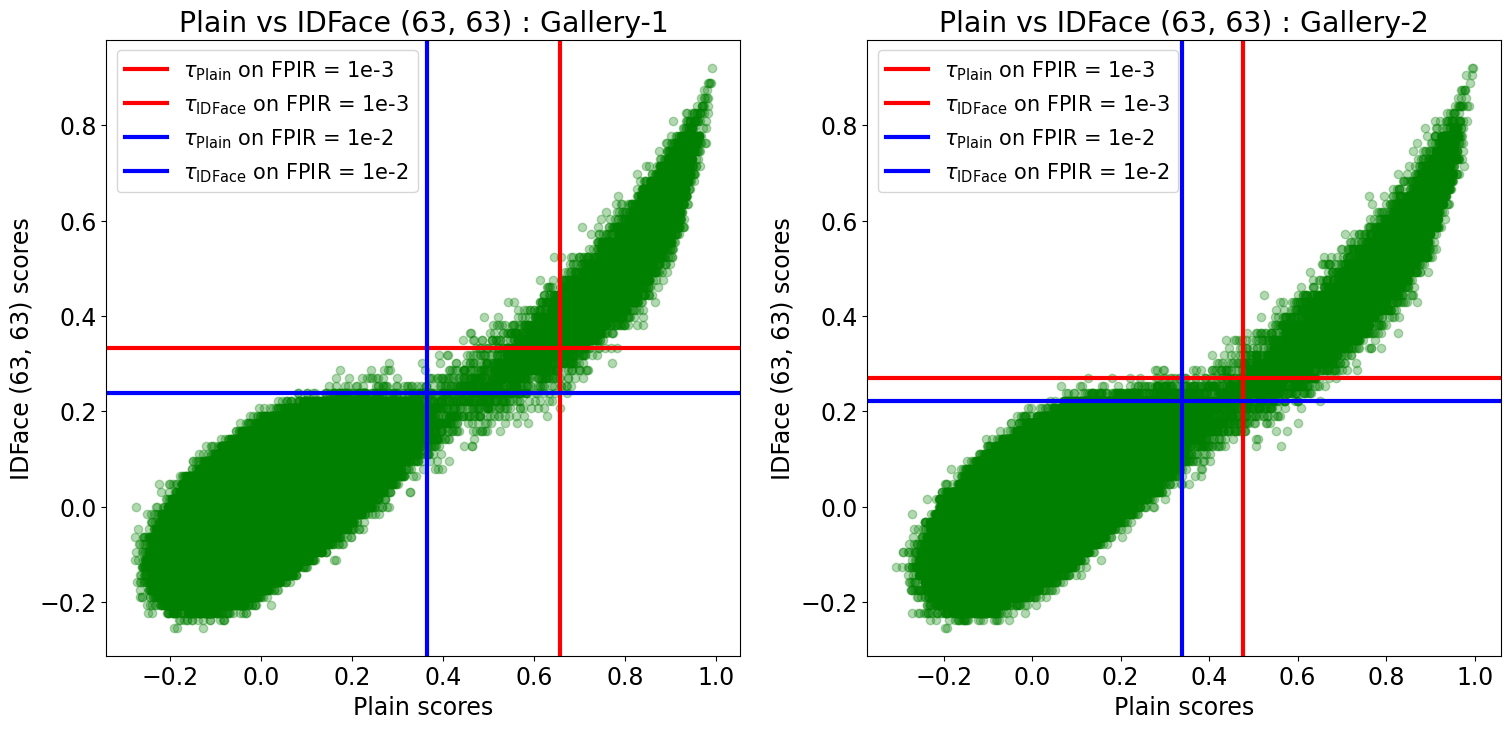}
        \caption{Plain vs. $\mathsf{IDFace}$ with $(\alpha,\beta) = (63,63)$.}
    \end{subfigure}
    \vspace{-3mm}    
    \caption{Score distribution in the IJB-C identification task on AdaFace-IResNet101~\cite{kim2022adaface} with non-protected feature extractor (Plain) and $\mathsf{IDFace}$. Each point represents an image pair, with scores from Plain (x-axis) and $\mathsf{IDFace}$. Vertical and horizontal lines indicate the thresholds for each setting.}
    \label{fig:odd_phenomenon}
\end{figure}

\noindent \textbf{Investigating an Odd Phenomenon: Observing Changes on the Decision of the FR Model.}
Although the above-mentioned analysis gives an insight about the effect of the transformation nearby the threshold, still this does not give a direct explanation on the accuracy drop and is insufficient to explain an unexpected and rather odd phenomenon across Tab.~\ref{fig:total_table1},\ref{fig:total_table2}, and Tab.~\textcolor{iccvblue}{3} in the main text: $\mathsf{IDFace}$ achieved slightly better accuracy than the plain model for some settings, and the amount of accuracy change becomes larger for low FPIR settings, e.g., FPIR-1e-3 or 1e-4. To complement this, we investigated the occurrence of the ``bad" events that would change the decision of the FR model after transformation. More precisely, we can consider two cases, (1) the pairs of faces from the same identity are no longer recognized as the same person, or (2) those pairs from different identities become recognized as the same person. 

To this end, we tracked the cosine similarity scores from each pair of faces before and after transformation. We considered all the pairs of facial images that appeared during IJB-C identification benchmark procedure for each gallery (G1 and G2). The total numbers of pairs made from each gallery are 34,718,796 and 34,464,087, respectively. We used the AdaFace-IResNet101~\cite{kim2022adaface} as a feature extractor and tested $\mathsf{IDFace}$ with parameters $(\alpha, \beta) = (341,341)$ and $(63,63)$, where we can observe that both the increase and decrease of accuracy at the same parameter but different levels of FPIRs.

We visualized the data through the scatter plot about the scores of pairs after transformation ($y$-axis) over those before transformation ($x$-axis), and the result is provided in Fig.~\ref{fig:odd_phenomenon}. We also depict the threshold from the benchmark before/after transformation by drawing horizontal and vertical solid lines. From this figure, we can observe that the points in the 2nd and 4th quadrants correspond to the case that the FR model gives a different decision because of the transformation. More precisely, points in the 2nd quadrant were recognized as the different person but not after transformation, and vice versa. Since FPIR is fixed for each threshold, we can ensure that points in the 2nd quadrant contribute to increasing the accuracy after transformation, and vice versa.

\begin{table}[t]
    \centering
    \begin{subtable}{.38\textwidth}
    \resizebox{\textwidth}{!}{
    \begin{tabular}{c|c|c|c|c}
        \hline
            & FPIR & Settings & TPIR(\%) & Thres. \\ \hline
         \multirow{4}{*}{G1} & \multirow{2}{*}{1e-2} & Plain & 97.05 & 0.3649  \\ \cline{3-5}
         & & $\mathsf{IDFace}$ & 96.69 & 0.3079 \\ \cline{2-5}
         & \multirow{2}{*}{1e-4} & Plain & 51.68 & 0.7888 \\ \cline{3-5}
         & & $\mathsf{IDFace}$ & 58.13 & 0.6539 \\ \hline
         \multirow{4}{*}{G2} & \multirow{2}{*}{1e-2} & Plain & 95.78 & 0.3389 \\ \cline{3-5}
         & & $\mathsf{IDFace}$ & 95.22 & 0.2962 \\ \cline{2-5}
         & \multirow{2}{*}{1e-4} & Plain & 79.52 & 0.6282 \\ \cline{3-5}
         & & $\mathsf{IDFace}$ & 78.63 & 0.5249 \\ \hline
         \multirow{4}{*}{Avg.} & \multirow{2}{*}{1e-2} & Plain & 96.41 & \multirow{4}{*}{N/A} \\ \cline{3-4}
         & & $\mathsf{IDFace}$ & 95.95 & \\ \cline{2-4}
         & \multirow{2}{*}{1e-4} & Plain & 65.60 & \\ \cline{3-4}
         & & $\mathsf{IDFace}$ & 68.38 &  \\ \hline
    \end{tabular}    
    }
    \caption{Plain vs. IDFace, $(\alpha, \beta) = (341,341)$.}
    \end{subtable}
    \qquad\quad
    \begin{subtable}{.38\textwidth}
    \resizebox{\textwidth}{!}{
    \begin{tabular}{c|c|c|c|c}
        \hline
            & FPIR & Settings & TPIR(\%) & Thres. \\ \hline
         \multirow{4}{*}{G1} & \multirow{2}{*}{1e-2} & Plain & 97.05 & 0.3649  \\ \cline{3-5}
         & & $\mathsf{IDFace}$ & 93.80 & 0.2381 \\ \cline{2-5}
         & \multirow{2}{*}{1e-3} & Plain & 80.00 & 0.6578 \\ \cline{3-5}
         & & $\mathsf{IDFace}$ & 84.56 & 0.3333 \\ \hline
         \multirow{4}{*}{G2} & \multirow{2}{*}{1e-2} & Plain & 95.78 & 0.3389 \\ \cline{3-5}
         & & $\mathsf{IDFace}$ & 91.95 & 0.2222 \\ \cline{2-5}
         & \multirow{2}{*}{1e-3} & Plain & 91.63 & 0.4756 \\ \cline{3-5}
         & & $\mathsf{IDFace}$ & 87.86 & 0.2698 \\ \hline
         \multirow{4}{*}{Avg.} & \multirow{2}{*}{1e-2} & Plain & 96.41 & \multirow{4}{*}{N/A} \\ \cline{3-4}
         & & $\mathsf{IDFace}$ & 92.87 & \\ \cline{2-4}
         & \multirow{2}{*}{1e-3} & Plain & 85.81 & \\ \cline{3-4}
         & & $\mathsf{IDFace}$ & 86.21 &  \\ \hline
    \end{tabular}
    }
    \caption{Plain vs. IDFace, $(\alpha, \beta) = (63,63)$.}
    \end{subtable}    
    \caption{Evaluation results of IJB-C identification benchmark with respect to galleries. We reported results from the same setting as in Fig.~\ref{fig:odd_phenomenon}. G1, G2, and Avg. indicate gallery 1, gallery 2, and average value of TPIR from each gallery, respectively.}
    \vspace{-3mm}
    \label{tab:odd_phenomenon_data}
\end{table}

From these figures, first we can observe that at the left hand side (Gallery 1) of Fig.~\ref{fig:odd_phenomenon_a}, we can observe that more samples lie in the 2nd quadrant than the 4th quadrant when FPIR=1e-4, whereas these quantities are almost the same in FPIR=1e-2. On the other hand, on the right hand side (Gallery 2), the number of samples in the 2nd and 4th quadrants is similar when FPIR=1e-4, whereas there are more samples in 4th quadrant when FPIR=1e-2. We figured out that this corresponds to the benchmark result provided in Tab.~\ref{tab:odd_phenomenon_data}. We can observe that in gallery 1, there is a huge increase on TPIR (6.45\%) after transformation when FPIR=1e-4, whereas there is a slight degradation on TPIR (0.36\%) when FPIR=1e-2. Both results were expected from Fig.~\ref{fig:odd_phenomenon_a}. We can interpret the result in $(\alpha,\beta)=(63,63)$ in a similar way.

Along with the above-mentioned analysis based on ``bad" events, we can also observe that larger variations of the accuracy appear when FPIR is low. We guess that this coincides with our analysis on the distance preservation property of the proposed transformation. More precisely, the acceptance threshold in low FPIR settings becomes tighter, \textit{i.e.}, becomes larger in our setting compared to relatively high FPIR settings, and as we observed in Figure~\ref{fig:kirruk}, $\mathsf{IDFace}$ would affect the distance relationship more at this region. Concretely, in the IJB-C benchmark for identification, the threshold for FRIR=1e-4 is set to 0.7888, whereas that for FPIR=1e-2 is set to 0.3649 in gallery 1 without transformation. Thus, we can expect that the accuracy more fluctuates in a stochastic way as the threshold becomes tighter, and this coincides with our analysis results.

We remark that although our analysis provides details about how the proposed transformation affects the accuracy of the FR model, this does not tell us why such a phenomenon \textit{should} occur. At this moment, we guess that our results are merely (un)fortunate coincidences on ``bad" events. We leave detailed investigations about the phenomenon, e.g., clarifying what conditions make ``bad" events happen, as interesting future work.

\section{Viewing Almost-Isometric Transformation as a Ternary Quantization Method}
As mentioned in Section~\textcolor{iccvblue}{4}, one may regard our almost-isometric transformation as a ternary quantization tailored for real-valued unit vectors. In fact, the quantization methods themselves are frequently appearing in several realms of deep learning literature, especially for reducing the computational cost of doing expensive operations~\cite{zhu2017trained, li2021trq, ma2024era, hubara2018quantized}, \textit{e.g.}, matrix multiplication. For this reason, employing quantization methods for improving efficiency would seem to be rather natural. Nevertheless, we found that many previous HE-based BTP constructions~\cite{engelsma2019learning, osorio2022stable, boddeti2018secure, Bausp2022improved, engelsma2022hers, huang2023efficient, yukun2017secure, meng2022towards} regarded quantization methods as a way of fitting the real-valued templates to discrete spaces where the cryptographic tools are defined, \textit{e.g.}, finite fields. In addition, analysis of quantization methods in terms of distance preservation has not been considered in their works. To our knowledge, ours is the first attempt to exploit the ternary quantization method for designing an efficient HE-based BTP, along with a theoretical guarantee on the distance preservation.

Of course, the proposed almost-isometric transformation may be replaced to other ternary quantization methods. However, when considering our application scenario, we found that the quantization method would satisfy some desirable properties tailored for $\mathsf{IDFace}$. 
First, it would be beneficial if the quantization method is independent of the distribution of the input data. If not the quantization method does depend on the distribution of the face templates, which is determined by the choice of feature extractor, then the resulting BTP scheme becomes no longer applicable in a plug-and-play manner; additional processes such as training would be required for the quantization method.
In fact, several quantization methods, \textit{e.g.}, product quantization~\cite{jegou2010product}, entropy-based quantization~\cite{kohavi1996error}, or quantization methods for the weight of neural networks~\cite{zhu2017trained, li2021trq, ma2024era, hubara2018quantized}, exploit the distribution of input data or require additional training; so they do not satisfy this property.

In addition, to ensure the correct computation of the cosine similarity, we recommend that the nonzero elements of each quantized vector would not depend on the input vector. Recall that when we compute the cosine similarity between two quantized vectors $T_{\alpha}(\mathbf{x})$ and $T_{\beta}(\mathbf{y})$, we omitted the normalization of each vector. In our almost-isometric transformation, this does not affect the correctness of the cosine similarity value because the $\ell_{2}$ norm of $T_{\alpha}(\mathbf{x})$ and $T_{\beta}(\mathbf{y})$ are always $\sqrt{\alpha}$ and $\sqrt{\beta}$, respectively, regardless of the input vectors $\mathbf{x}$ and $\mathbf{y}$. That is, the cosine similarity between $T_{\alpha}(\mathbf{x})$ and $T_{\beta}(\mathbf{y})$ can be viewed as the inner product $\langle T_{\alpha}(\mathbf{x}), T_{\beta}(\mathbf{y}) \rangle$ with a constant factor $\frac{1}{\sqrt{\alpha \beta}}$. 
Of course, one may mitigate this issue by post-processing the output of the quantization method or slightly tweak the proposed $\mathsf{IDFace}$. Nevertheless, ternary quantization methods without this property, \textit{e.g.}, setting intervals for assigning the corresponding value for each component such as equal-width or equal-probability quantization methods~\cite{lim2011analysis, drozdowski2018bench}, would not be employed for $\mathsf{IDFace}$ \textit{as is}.

By reflecting these properties, few quantization methods remain. Especially, these methods can be classified via the number of nonzero elements that appeared in the quantized vector. At this moment, we now introduce our rationale to select our almost-isometric transformation. Recall that to minimize the change of cosine similarity carried by the quantization, as seen in Section~\ref{sec:supp_B}, it is beneficial to make the number of possible quantization results as large as possible. Since the proposed transformation fully utilizes the space of $\mathcal{Z}_{\alpha}$, along with a mathematical proof that it is indeed almost-isometric regardless of the input data, we selected it for our $\mathsf{IDFace}$. We leave investigations and analyses on other quantization methods satisfying those two requirements as future work.

\section{Omitted Algorithms and Full Description of \textsf{IDFace}}\label{sec:supp_C}

In this section, we provide full descriptions of omitted algorithms and the proposed \textsf{IDFace} for the sake of completeness.

\subsection{Omitted Algorithms}\label{sec:supp_C_1}

We first provide the full algorithms of the almost-isometric transformation~($T_{\alpha}$ in Section~\textcolor{iccvblue}{4.1}) and the space-efficient encoding~(\textsf{Encode} and \textsf{Decode} in Section~\textcolor{iccvblue}{4.2}), each of which is described in Algorithm~\ref{alg:transf}, \ref{alg:encode} and \ref{alg:decode}, respectively. Recall that $\mathcal{Z}_{\alpha}^{d}$ is a set of ternary vectors of length $d$ with $\alpha$ nonzero components whose components are one of $\{0, \pm 1\}$.

\begin{figure}[h]
\vspace{-5mm}
\begin{algorithm}[H]
\renewcommand{\algorithmicensure}{\textbf{Functionality:}}
\caption{$\mathsf{Transformation}$ $T_{\alpha}$}\label{alg:transf}
\begin{algorithmic}[1]
\ENSURE $T_{\alpha}:\mathbb{S}^{d-1} \rightarrow \mathcal{Z}^{d}_{\alpha}$
\REQUIRE $\mathbf{x} = (x_{1}, \dots, x_{d}) \in \mathbb{S}^{d-1}$ and $\alpha \in [d]$, 
\STATE Find a set $J\subset[d]$ of $\alpha$ indices such that for any $j\in J$ and $k\not\in J$, $|x_j|\geq|x_k|$.
\STATE For all $j\in[d]$, set $z_j=\left\{
\begin{array}{cl}
\frac{x_{j}}{ |x_{j}|  }&\text{if } j\in J\\
0&\text{if } j\not\in J
\end{array}\right.$
\RETURN $\mathbf{z}=(z_{1},\dots,z_{d})\in \mathcal{Z}^{d}_{\alpha}$
\end{algorithmic}
\end{algorithm}
\vspace{-5mm}
\end{figure}

\begin{figure}[h]
\vspace{-5mm}
\begin{algorithm}[H]
\caption{$\mathsf{Encode}$}\label{alg:encode}
\begin{algorithmic}[1]
\REQUIRE $\mathbf{x}_{1},\dots,\mathbf{x}_{m} \in \mathcal{Z}_{\alpha}^{d}$ and $p \in \mathbb{N}$ such that $ p \ge \alpha$
\STATE For $i\in[m]$, compute $\mathbf{x}_i^{+} = (|\mathbf{x}_i|+\mathbf{x}_i)/2$ and $\mathbf{x}_i^{-} = (|\mathbf{x}_i|-\mathbf{x}_i)/2$
\STATE For $*\in\{+,-\}$, compute $\mathbf{x}^{*}=\sum_{i=1}^{m=1}p^{i-1} \cdot \mathbf{x}_{i}^*$
\RETURN ($\mathbf{x}^+,\mathbf{x}^-$).
\end{algorithmic}
\end{algorithm}
\vspace{-5mm}
\end{figure}

\begin{figure}[h]
\vspace{-5mm}
\begin{algorithm}[H]
\caption{$\mathsf{Decode}$}\label{alg:decode}
\begin{algorithmic}[1]
\REQUIRE $(\mathbf{x}^{+}, \mathbf{x}^{-}) \in \mathbb{N}^{d} \times \mathbb{N}^{d}$ and $p \in \mathbb{N}$
\STATE For $i$ from $1$ to $m$ do:
\STATE \qquad For $* \in \{+, -\}$, compute $\mathbf{s}_{i}^{\ast} \gets \mathbf{x}^{*} \pmod p$
\STATE \qquad Compute $\mathbf{s}_{i} \gets \mathbf{s}_{i}^{+} - \mathbf{s}_{i}^{-}$
\STATE \qquad For $* \in \{+, -\}$, update $\mathbf{x}^{*} \gets \frac{\mathbf{x}^{*} - \mathbf{s}_{i}^{*}}{p}$
\RETURN $(\mathbf{s}_{1}, \dots, \mathbf{s}_{m})$
\end{algorithmic}
\end{algorithm}
\vspace{-5mm}
\end{figure}

In addition, we provide the full algorithms of our improved database encryption scheme $(\mathsf{IDFace.ENC}_{\mathtt{DB}}, \mathsf{IDFace.IP}_{\mathtt{DB}})$ in Algorithm~\ref{alg7:enrolltotal} and~\ref{alg7:idsub}, respectively, which were introduced in Section~\textcolor{iccvblue}{4.3}.

\begin{figure*}[h]
\begin{minipage}[t]{.49\linewidth}
\begin{algorithm}[H]
\caption{$\mathsf{IDFace.ENC}_{\mathtt{DB}}$ }\label{alg7:enrolltotal}
\begin{algorithmic}[1]
\REQUIRE $\mathbf{X} \in \mathbb{R}^{mN \times d}$, $\alpha \in [d]$, and $\mathsf{pk}$
\STATE Parse $\mathbf{X}$ as $[ \mathbf{x}_{1},\dots,\mathbf{x}_{mN} ]^{T}$
\STATE For $i\in [mN]$, compute $\mathbf{z}_{i} =T_{\alpha}(\mathbf{x}_{i})$
\STATE For $i\in [N]$, compute \\
$(\mathbf{x}^{+}_{i},\mathbf{x}^{-}_{i})\hspace{-0.1cm} \leftarrow \mathsf{Encode}(\mathbf{z}_{(i-1)m+1}\ldots,\mathbf{z}_{im})$
\STATE For $*\in\{+,-\}$, set $\mathbf{X}^{*} = [\mathbf{x}^{*}_{1},\dots,\mathbf{x}^{*}_{N}]^{T}$.
\STATE For $*\in\{+,-\}$, set $\mathcal{C}^{*} \gets \mathsf{ENC}_{\mathtt{DB}}(\mathbf{X}^{*},\mathsf{pk})$
\RETURN $\mathcal{C}$ = ($\mathcal{C}^{+}, \mathcal{C}^{-}$)
\end{algorithmic}
\end{algorithm}
\end{minipage}
\hfill
\begin{minipage}[t]{.49\linewidth}
\begin{algorithm}[H]
\caption{$\mathsf{IDFace.IP}_{\mathtt{DB}}$}\label{alg7:idsub}
\begin{algorithmic}[1]
\REQUIRE $\mathbf{y} \in \mathbb{R}^{d}$, a set $\mathcal{C}$, and $\beta \in [d]$
\STATE $\mathbf{z}\in \{-1,0,1 \}^{d} \gets T_{\beta}(\mathbf{y})$
\STATE Compute $\mathbf{z}^{+} = (|\mathbf{z}|+\mathbf{z})/2$ 
\\ $\hspace{0.5cm}$and\ \ $\mathbf{z}^{-} = (|\mathbf{z}|-\mathbf{z})/2$. 
\STATE Parse $(\mathcal{C}^{+}, \mathcal{C}^{-})$ as $\mathcal{C}$ and compute \\
$\hspace{0.5cm}$ $\mathbf{ct}^{+} =\mathsf{IP}_{\mathtt{DB}}(\mathbf{z}^{+},\mathcal{C}^{+}) \oplus \mathsf{IP}_{\mathtt{DB}}(\mathbf{z}^{-},\mathcal{C}^{-})$, \\
$\hspace{0.5cm}$ $\mathbf{ct}^{-} =\mathsf{IP}_{\mathtt{DB}}(\mathbf{z}^{+},\mathcal{C}^{-}) \oplus \mathsf{IP}_{\mathtt{DB}}(\mathbf{z}^{-},\mathcal{C}^{+})$ 
\RETURN $(\mathbf{ct}^{+}$ , $\mathbf{ct}^{-})$
\end{algorithmic}
\end{algorithm}
\end{minipage}
\vspace{-4mm}
\end{figure*}

\subsection{Full Description of \textsf{IDFace}}\label{sec:supp_C_2}
We now provide the full description of \textsf{IDFace}. Recall that \textsf{IDFace} is an instantiation of the biometric identification scheme from the proposed database encryption scheme $(\mathsf{IDFace.ENC}_{\db}, \mathsf{IDFace.IP_{\db}})$. Although its generic construction was already given in Section~\textcolor{iccvblue}{4.3}, we give a detailed description for the sake of completeness. Here, $\textsf{DEC}$ denotes the decryption algorithm of the underlying AHE, and $(\mathsf{pk}, \mathsf{sk})$ is a pair of the public key and secret key of it. In addition, we consider the database $\db$ as a list, which is initialized as an emptyset. The description of \textsf{IDFace} is given in Fig.~\ref{fig:idface_full}.

\begin{figure}[h]
\begin{center}
\fbox{
\parbox[h]{0.9\linewidth}
{
\centering
\begin{description}
\item \begin{center}$\mathsf{Enroll}$\end{center}
\vspace*{1mm}
\item \hspace*{0.0cm}Parameter: $\alpha \in [d]$ and a bound $p\ge \alpha$.
\item \hspace*{0.0cm}$\mathcal{S}_{local}$'s input: $\mathbf{X} \in \mathbb{R}^{mN \times d}$, $\mathtt{ID}$, $\db$, and $\mathsf{pk}$\vspace{-1mm}

\hspace*{-1cm}
\rule{.9\textwidth}{0.1mm}
% \vspace*{1mm}

\item \hspace*{-0.15cm}\fbox{$\mathcal{S}_{local}$} : $\mathcal{C} \gets \mathsf{IDFace.ENC}_{\mathtt{DB}}(X, \alpha, \mathsf{pk})$

\vspace*{1mm}

\item \hspace*{-0.15cm}\fbox{$\mathcal{S}_{local}$} :
Add ($\mathtt{ID}, \mathcal{C}$) on the bottom row of $\db$\vspace{-1mm}

\vspace*{1mm}

\hspace*{-1.15cm}
\rule{.923\textwidth}{0.2mm}

%
%   identity
%
\item \begin{center}$\mathsf{Identify}$\end{center}
\vspace*{1mm}
\item \hspace*{0.0cm}Parameter: $\beta \in [d]$, a threshold $\tau$, and a bound $p$.
\item \hspace*{0.0cm}$\mathcal{S}_{local}$'s input: $\mathbf{y} \in \mathbb{R}^{d}$ and $\db$ of size $D$
\item \hspace*{0.0cm}$\mathcal{S}_{key}$'s input: secret key $\mathsf{sk}$\vspace{-1mm}

\hspace*{-1cm}
\rule{.9\textwidth}{0.1mm}
% \vspace*{1mm}
%
\item \hspace*{-0.15cm}\fbox{$\mathcal{S}_{local}$} : For $(\mathtt{ID}_{i}, \mathcal{C}_{i}) \in \db$, compute $(\mathbf{ct}_{i}^{+}, \mathbf{ct}_{i}^{-}) \leftarrow \mathsf{IDFace.IP}_{\mathtt{DB}}(\mathbf{y}, \mathcal{C}_{i}, \beta)$
\vspace*{1mm}
\item \hspace*{-0.15cm}\fbox{$\mathcal{S}_{local}\rightarrow\mathcal{S}_{key}$} : Send $(\mathbf{ct}_{i}^{+}, \mathbf{ct}_{i}^{-})_{i=1}^D$.
\vspace*{1mm}
\item \hspace*{-0.15cm}\fbox{$\mathcal{S}_{key}$} : For $i\in[D]$ and $*\in\{+,-\}$ do:
\\
\hspace*{.75cm}$\mathbf{pt}_{i}^{*}\leftarrow \mathsf{DEC_{sk}}(\mathbf{ct}_{i}^{*})_{i=1}^D$
\vspace*{1mm}
\item \hspace*{-0.15cm}\fbox{$\mathcal{S}_{local}$} : For $i\in[D]$ do:\\
\hspace*{0.75cm}For $*\in\{+,-\}$, parse $\mathbf{pt}_{i}^{*}$ as $(s^{*}_{ij})_{j=1}^{N}$.\\
\hspace*{0.75cm}For $j\in [N]$, compute $({s}_{ijk})_{k=1}^{m} \gets \mathsf{Decode}(s^{+}_{ij}, s^{-}_{ij}, p)$. 
\\
\hspace*{0.15cm}If $\max_{i,j,k}\{s_{ijk}\}>\tau$: \\
\vspace{1mm}
\hspace*{0.75cm}Set $idx \gets (i, mj + k - 1)$. \\
\vspace{1mm}
\hspace*{0.15cm}Otherwise, Set $idx \gets \bot$.
\vspace*{1mm}
\item \hspace*{-0.15cm}\fbox{$\mathcal{S}_{key}\rightarrow\mathcal{S}_{local}$} : Send $idx$.
\\
\item \hspace*{-0.15cm}\fbox{$\mathcal{S}_{local}$} : If $idx \neq \bot$: \\
\vspace{1mm}
\hspace*{0.75cm}Parse $idx \rightarrow (r, s)$ and \textbf{return} $\mathtt{ID}_{r}[s]$. \\
\vspace{1mm}
\hspace*{0.15cm}Otherwise, \textbf{return} "reject".
\end{description}
}
}
\vspace{-5mm}
\end{center}\caption{Full description of $\mathsf{IDFace}$}\label{fig:idface_full}
\vspace{-3mm}
\end{figure}

\section{Secure Two-Party Computation-based Variant of \textsf{IDFace}}\label{sec:supp_E}

Recall that the proposed transformation is independent of the HE, so that it can be utilized to speed up secure face identification protocols from other cryptographic tools. As an example of this, we provide another face identification protocol called $\mathsf{2PCFace}$ from a two-party computation-based approach.

The basic idea of designing $\mathsf{2PCFace}$ is as follows: For two binary vectors $\mathbf{x}, \mathbf{y} \in \{0, 1\}^{d}$, their inner product can be computed as $\langle \mathbf{x}, \mathbf{y} \rangle = \mathcal{HW}(\mathbf{x} \land \mathbf{y})$, where $\mathcal{HW}$ denotes the Hamming weight of the given binary vector, and $\land$ is a component-wise AND operation. If we consider an additive share $(\mathbf{x}_{1}, \mathbf{x}_{2})$ of $x$, which means $\mathbf{x}_{1} \oplus \mathbf{x}_{2} = \mathbf{x}$ for component-wise XOR operator $\oplus$, we have that $\langle \mathbf{x}, \mathbf{y} \rangle = \mathcal{HW}((\mathbf{x}_{1} \land \mathbf{y}) \oplus (\mathbf{x}_{2} \land \mathbf{y}))$. From this, we can construct an identification protocol with two servers, $\mathcal{S}_{1}$ and $\mathcal{S}_{2}$, where each server stores the enrolled templates after transformation in the form of an additive share.
After receiving the transformed template $\mathbf{y}$, the matching score between the enrolled templates can be obtained by first calculating $\mathbf{x}_{1} \land \mathbf{y}$ and $\mathbf{x}_{2} \land \mathbf{y}$ locally for each server, sharing the result between servers, and finally calculating $\mathcal{HW}((\mathbf{x}_{1} \land \mathbf{y}) \oplus (\mathbf{x}_{2} \land \mathbf{y}))$ using the shared value. 
Although the output of almost-isometric transformation is a ternary vector comprised of $\{-1, 0, 1\}$, we can decompose it into the subtraction of two binary vectors by the sign of each component, as we did in the space-efficient encoding.

For a precise description, we introduce a helper algorithm called $\mathsf{GenShare}$, which will be a subroutine of $\mathsf{2PCFace}$. $\mathsf{GenShare}$ takes a template $\mathbf{x} \in \mathbb{R}^{mN \times d}$ and a transformation parameter $\alpha$ as inputs, returning additive shares of the positive components and negative components of $T_{\alpha}(\mathbf{x})$, respectively. The formal description of $\mathsf{Share}$ is provided in Algorithm~\ref{alg:genshare}.
Using $\mathsf{GenShare}$ with the above idea to compute matching scores, we describe the full protocol of $\mathsf{2PCFace}$ in Figure~\ref{fig:2pc_id}.

\begin{algorithm}[h]
\caption{$\mathsf{GenShare}$}
\label{alg:genshare}
\begin{algorithmic}[1]
\REQUIRE A template $\mathbf{x} \in \mathbb{R}^{d}$, and a parameter $\alpha$
\ENSURE $\mathbf{z}_{1}^{+}, \mathbf{z}_{1}^{-1}, \mathbf{z}_{2}^{+}, \mathbf{z}_{2}^{-} \in \{0, 1\}^{d}$ such that $(\mathbf{z}_{1}^{+} \oplus \mathbf{z}_{2}^{+}) - (\mathbf{z}_{1}^{-} \oplus \mathbf{z}_{2}^{-}) = T_{\alpha}(\mathbf{x})$
\STATE Compute $\mathbf{z} \gets T_{\alpha}(\mathbf{x})$ and set $\mathbf{z}^{+} \gets (|\mathbf{z}| + \mathbf{z})/2$ and $\mathbf{z}^{-} \gets (|\mathbf{z}| - \mathbf{z})/2$.
\STATE Sample $\mathbf{z}_{1}^{+}, \mathbf{z}_{1}^{-} \xleftarrow{\$}\{0, 1\}^{d}$ and set $\mathbf{z}_{2}^{+} \gets \mathbf{z}_{1}^{+} \oplus \mathbf{z}^{+}$ and $\mathbf{z}_{2}^{-} \gets \mathbf{z}_{2}^{+} \oplus \mathbf{z}^{-}$.
\RETURN $(\mathbf{z}_{1}^{+}, \mathbf{z}_{2}^{+}, \mathbf{z}_{1}^{-}, \mathbf{z}_{2}^{-})$.
\end{algorithmic}
\end{algorithm}

\begin{figure}[t]
\begin{center}
\fbox{
\parbox[h]{0.9\linewidth}
{
\centering
\begin{description}
\item \begin{center}$\mathsf{Enroll}$\end{center}

\item \hspace*{0.0cm}Parameter: $\alpha \in [d]$
\item \hspace*{0.0cm}$\mathcal{S}_{1}$'s input: $\mathbf{x} \in \mathbb{R}^{d}$, $id$, $\db_{1}$
\item \hspace*{0.0cm}$\mathcal{S}_{2}$'s input: $\db_{2}$\vspace{-1mm}

\hspace*{-1cm}
\rule{.9\textwidth}{0.1mm}
% \vspace*{1mm}

\item \hspace*{-0.15cm}\fbox{$\mathcal{S}_{1}$} : $(\mathbf{z}_{1}^{+}, \mathbf{z}_{2}^{+}, \mathbf{z}_{1}^{-}, \mathbf{z}_{2}^{-}) \gets \mathsf{GenShare}(\mathbf{x}, \alpha)$.

\vspace*{1mm}

\item \hspace*{-0.15cm}\fbox{$\mathcal{S}_{1} \rightarrow \mathcal{S}_{2}$} : Send $(\mathbf{z}_{2}^{+}, \mathbf{z}_{2}^{-})$

\vspace*{1mm}

\item \hspace*{-0.15cm}\fbox{$\mathcal{S}_{1}, \mathcal{S}_{2}$} :
Add ($id, \mathbf{z}_{i}^{+}, \mathbf{z}_{i}^{-}$) on the bottom row of $\db_{i}$ for $i \in \{1, 2 \}$\vspace{-1mm}

\vspace*{1mm}

\hspace*{-1.15cm}
\rule{.923\textwidth}{0.2mm}

%
%   identity
%
\item \begin{center}$\mathsf{Identify}$\end{center}
\item \hspace*{0.0cm}Parameter: $\beta \in [d]$, threshold $\tau$.
\item \hspace*{0.0cm}$\mathcal{S}_{1}$'s input: $\mathbf{y} \in \mathbb{R}^{d}$ and $\db_{1}$ of size $D$
\item \hspace*{0.0cm}$\mathcal{S}_{2}$'s input: $\db_{2}$ of size $D$\vspace{-1mm}

\hspace*{-1cm}
\rule{.9\textwidth}{0.1mm}
% \vspace*{1mm}

\item \hspace*{-0.15cm}\fbox{$\mathcal{S}_{1}$} : Set $\tilde{\mathbf{y}} \gets T_{\beta}(\mathbf{y})$, $\mathbf{y}^{+} \gets (|\tilde{\mathbf{y}}| + \tilde{\mathbf{y}})/2$ and $\mathbf{y}^{-} \gets (|\tilde{\mathbf{y}}| - \tilde{\mathbf{y}})/2$.
\vspace*{1mm}
\item \hspace*{-0.15cm}\fbox{$\mathcal{S}_{1} \rightarrow \mathcal{S}_{2}$} : Send $(\mathbf{y}^{+}, \mathbf{y}^{-})$
\vspace*{1mm}
\item \hspace*{-0.15cm}\fbox{$\mathcal{S}_{1}, \mathcal{S}_{2}$} : For $(id_{i}, \mathbf{z}_{i,j}^{+}, \mathbf{z}_{i,j}^{-}) \in \db_{j}$ do ($j \in \{1, 2\}$):
\\
\hspace*{.5cm} $\mathbf{p}^{(1)}_{i,j} \gets \mathbf{z}_{i,j}^{+} \land \mathbf{y}^{+}$, $\mathbf{p}^{(2)}_{i, j} \gets \mathbf{z}_{i, j}^{-} \land \mathbf{y}^{-}$, \\
\hspace*{.5cm} $\mathbf{m}^{(1)}_{i,j} \gets \mathbf{z}_{i,j}^{+} \land \mathbf{y}^{-}$, and $\mathbf{m}^{(2)}_{i, j} \gets \mathbf{z}_{i, j}^{-} \land \mathbf{y}^{+}$
\vspace*{1mm}
\item \hspace*{-0.15cm}\fbox{$\mathcal{S}_{2}\rightarrow\mathsf{S}_{1}$} : Send $(\mathbf{p}^{(1)}_{i,2}, \mathbf{p}^{(2)}_{i,2}, \mathbf{m}^{(1)}_{i,2}, \mathbf{m}^{(2)}_{i,2})_{i=1}^D$
\vspace*{1mm}
\item \hspace*{-0.15cm}\fbox{$\mathcal{S}_{1}$} : For $i\in[D]$ do:\\
\hspace*{.15cm} Set $s_{i} \gets (\mathcal{HW}(\mathbf{p}^{(1)}_{i, 1} \oplus \mathbf{p}^{(1)}_{i, 2}) + \mathcal{HW}(\mathbf{p}^{(2)}_{i, 1} \oplus \mathbf{p}^{(2)}_{i, 2}))$ \\
\hspace*{1.25cm} $- (\mathcal{HW}(\mathbf{m}^{(1)}_{i, 1} \oplus \mathbf{m}^{(1)}_{i, 2}) + \mathcal{HW}(\mathbf{m}^{(2)}_{i, 1} \oplus \mathbf{m}^{(2)}_{i, 2}))$ \\
\hspace*{.25cm}If $\max_{i}\{s_{i}\}>\tau$:\\
\vspace{1mm}
\hspace*{1cm} \textbf{return} the corresponding $id_{i}$.
\\
\vspace{1mm}
\hspace*{.25cm}Otherwise, \textbf{return} $\bot$
\end{description}
}
}
\end{center}
\vspace{-3mm}
\caption{Full description of $\mathsf{2PCFace}$.}
\label{fig:2pc_id}
\vspace{-5mm}
\end{figure}

Since component-wise AND operations can be seen as look-up operations, we can utilize the same efficiency-accuracy trade-off trick in $\mathsf{2PCFace}$ by employing different transformation parameters $(\alpha, \beta)$ in $\mathsf{Enroll}$ and $\mathsf{Identify}$, respectively. More precisely, during identification, we can treat $\mathbf{z}^{\dagger}_{i,j} \land \mathbf{y}^{\ast}$ for $\dagger, \ast \in \{+, - \}$ as look-up components of $\mathbf{z}^{\dagger}_{i,j}$ on the positions where corresponding components in $\mathbf{y}^{\ast}$ are nonzero. Thus, if we denote $\mathbf{z}^{\dagger}_{i,j}[\mathbf{y}^{\ast}]$ as a subvector of $\mathbf{z}^{\dagger}_{i,j}$ by collecting the positions where corresponding $\mathbf{y}^{\ast}$'s components are nonzero, then we can optimize the $\mathsf{Identify}$ in terms of both communication cost and computational cost by replacing each $\mathbf{p}_{i,j}^{(1)}$, $\mathbf{p}_{i,j}^{(2)}$, $\mathbf{m}_{i,j}^{(1)}$, $\mathbf{m}_{i,j}^{(2)}$ to $\mathbf{z}_{i,j}^{+}[\mathbf{y}^{+}]$, $\mathbf{z}_{i,j}^{-}[\mathbf{y}^{-}]$, $\mathbf{z}_{i,j}^{+}[\mathbf{y}^{-}]$, $\mathbf{z}_{i,j}^{-}[\mathbf{y}^{+}]$, respectively. Note that the remaining operations are still well-defined because each corresponding binary vector computed by each server would be the same length.

Compared to HE-based approaches, including $\mathsf{IDFace}$, $\mathsf{2PCFace}$ has some advantages in computational efficiency and the underlying security assumption.
For the former, we note that each server in $\mathsf{2PCFace}$ stores bit strings of the same length as the transformed template, thus preventing the storage overhead from the encryption. In addition, the unit operation conducted by each server becomes bit-wise operations, which are considerably cheaper than homomorphic addition and (scalar) multiplication. 
For the latter, $\mathsf{2PCFace}$ does not require an assumption for securing the secret key, which is an essential requirement in HE-based approaches. In fact, the attacker cannot distinguish between the shares of enrolled templates and a random string of equal length unless the attacker steals the whole databases of both servers. Moreover, under the same attacker, one can easily check that $\mathsf{2PCFace}$ satisfies all standard security requirements for BTP: \textit{irreversibility}, \textit{revocability}, and \textit{unlinkability}.
We also remark that the amount of accuracy degradation in $\mathsf{2PCFace}$ is identical to that of $\mathsf{IDFace}$, under the same transformation parameters $(\alpha, \beta)$.

However, compared to $\mathsf{IDFace}$, $\mathsf{2PCFace}$ has some downsides on (1) its communication cost and (2) extension to the scenario with mulitple devices and servers, which is presented in Appendix~\ref{sec:supp_F}.

\;

\noindent \textbf{Comparison on the Communication Cost.}
In $\mathsf{IDFace}$, the ciphertext of whole inner product values can be computed locally, but this is impossible in $\mathsf{2PCFace}$; In this case, each server calculates only a partial part of inner product values and communicates with each other to obtain the full inner product value. That is, for one execution of $\mathsf{Identify}$, the server $\mathcal{S}_{2}$ in $\mathsf{2PCFace}$ should send the bit string of the same length as the entire database in one server, whereas $\mathsf{IDFace}$ suffices to send the encryption/decryption of the inner product value corresponding to each stored identity between servers $\mathcal{S}_{local}$ and $\mathcal{S}_{key}$. Although the cost of the former can be reduced by sending subvectors rather than the entire one, we figure out that the communication cost of $\mathsf{IDFace}$ is cheaper than $\mathsf{2PCFace}$ under the same transformation parameters.

For a concrete comparison, we provide the formula for the measuring communication cost of each protocol as follows: For simplicity, we assume that the transformation parameters are $(\alpha, \beta)$, the dimension of the template is $d$, and the number of enrolled identities is $D$.

\;

\noindent \textbf{Communication Cost of $\mathsf{IDFace}$.} We can figure out that the communication occurs when (1) $\mathcal{S}_{local}$ sends $(\mathbf{ct}^{+}_{i}, \mathbf{ct}^{-}_{i})_{i=1}^{\hat{D}}$ to $\mathcal{S}_{key}$ and (2) $\mathcal{S}_{key}$ sends $idx$ to $\mathcal{S}_{local}$, where $\hat{D} = \lceil \frac{D}{N_{s} \cdot N_{enc}} \rceil$ for the number of slots $N_{s}$ in underlying AHE and the number of templates $N_{enc}$ encoded at once. More precisely, $2\hat{D}$ ciphertexts and $1$ string representing the index are communicated. Since $\lceil \log D \rceil$ bits suffice for representing all $D$ identities, the total communication cost is $2\hat{D} \cdot S_{ct} + \log D$, where $S_{ct}$ stands for the size of ciphertext.

\;

\noindent \textbf{Communication Cost of $\mathsf{2PCFace}$.} We can figure out that the communication occurs when (1) $\mathcal{S}_{1}$ sends $(\mathbf{y}^{+}, \mathbf{y}^{-})$ to $\mathcal{S}_{2}$ and (2) $\mathcal{S}_{2}$ sends subvectors $(\mathbf{z}_{i,2}^{+}[\mathbf{y}^{+}]$, $\mathbf{z}_{i,2}^{-}[\mathbf{y}^{-}]$, $\mathbf{z}_{i,2}^{+}[\mathbf{y}^{-}]$, $\mathbf{z}_{i,2}^{-}[\mathbf{y}^{+}])_{i=1}^{D}$ to $\mathcal{S}_{1}$. The former takes $2d$ bits of communication. For the latter, we can observe that since the number of nonzero entries in $\mathbf{y}$ is $\beta$, the concatenation of these 4 vectors in the index $i$ has length $2 \beta$. Hence, the latter requires $2N\beta$ bits of communication, and the total cost is $2D\beta + 2d$ bits.

By following these formulae, we compute the exact communication cost of $\mathsf{IDFace}$ instantiated by CKKS and Paillier Cryptosystem (PC), and $\mathsf{2PCFace}$. For simplicity, we denote $\mathsf{IDFace_{CKKS}}$ and $\mathsf{IDFace_{PC}}$ for the instantiation of $\mathsf{IDFace}$ using CKKS and PC, respectively. The number of enrolled identities $D$ is set to 1 M. The corresponding $N_{s}, N_{enc}$ and $S_{ct}$ in the parameter settings of CKKS and PC used in $\mathsf{IDFace}$\footnote{The detailed parameter settings of each scheme were already provided in Section~\textcolor{iccvblue}{5}.}, along with the comparison result on the communication cost, are provided in Table~\ref{tab:comm_cost}. According to the table, the communication cost of $\mathsf{IDFace_{CKKS}}$ ($\mathsf{IDFace_{PC}}$, resp.) is about 1.84X $\sim$ 6.28X (5.13X $\sim$ 18.44X, resp.) cheaper than $\mathsf{2PCFace}$. Thus, for circumstances where the bandwidth of the communication channel is limited (\textit{e.g.,} 100MB/s or 1GB/s), deploying $\mathsf{IDFace}$ would be more favorable than $\mathsf{2PCFace}$.

\begin{table}[h]
\centering
\begin{minipage}{.4\linewidth}
    \resizebox{\linewidth}{!}{
    \begin{tabular}{c|ccc}
        \hline
        Protocols & $N_{s}$ & $N_{enc}$ & $S_{ct}$ \\ \hline
        $\mathsf{IDFace_{CKKS}}$ & 4096 & (8, 7, 5) & 132KB \\
        $\mathsf{IDFace_{PC}}$ & 1 & (341, 292, 227) & 0.5KB\\ \hline
    \end{tabular}
    }
\end{minipage}
\qquad\quad
\begin{minipage}{.4\linewidth}
    \resizebox{\linewidth}{!}{
    \begin{tabular}{c|cc|c}
        \hline
         $(\alpha, \beta)$ & $\mathsf{IDFace_{CKKS}}$  & $\mathsf{IDFace_{PC}}$ & $\mathsf{2PCFace}$ \\ \hline
         (341, 341)& 12.94MB & 4.41MB & 81.30MB  \\ 
         (341, 127)& 9.24MB & 3.43MB & 30.28MB \\ 
         (341, 63)& 8.18MB & 2.93MB & 15.02MB  \\ \hline
    \end{tabular}
    }
\end{minipage}
\vspace{-2mm}
\caption{$N_{s}$, $N_{enc}$, and $S_{ct}$ for $\mathsf{IDFace_{CKKS}}$ and $\mathsf{IDFace_{PC}}$ (Left), and comparison of communication cost required for one identification in each protocol (Right). In the column of $N_{enc}$, the tuple indicates $N_{enc}$ values for $\beta = 63, 127, 341$, respectively. The number of identities is set to 1M.}\label{tab:comm_cost}
\vspace{-5mm}
\end{table}

\section{Application of \textsf{IDFace} for Scenarios with Multiple Devices and Servers}\label{sec:supp_F}

Recall that both $\mathsf{IDFace}$ and $\mathsf{2PCFace}$ regard the scenario in which there is only a single device for recognition. We now consider an extension for this: The large-scale identification scenario where multiple recognition devices are involved. 
Although it is possible for a single server to deal with multiple devices, its computational burden becomes extremely severe if we consider lots of devices, \textit{e.g.} more than a hundred, or 1K.
As shown in some real-world applications~\cite{su2012extremely,mbss}, the computational burden for the server in this circumstance can be alleviated by employing multiple servers to distribute the \textsf{Identify} requests from the whole device. We assume that these multiple servers are well-synchronized, holding the same enrolled templates in the form of encrypted ones. Note that the stored data on each server is not necessarily identical. In fact, it is desirable to ensure that they are (computationally) indistinguishable from the random string for the sake of unlinkability.
In this paragraph, we will show that $\mathsf{IDFace}$ can be naturally extended with a negligible overhead on communication cost, whereas the extension of $\mathsf{2PCFace}$ requires relatively larger communication cost overhead proportional to the number of servers.

For a precise comparison, let us consider the case where there are $\mu$ servers, denoted by $\mathcal{S}_{1},\dots,\mathcal{S}_{\mu}$, that store enrolled templates and assume that each server handles the requests from $\nu$ devices, denoted by $\mathfrak{d}_{i, 1}, \dots, \mathfrak{d}_{i, \nu}$ for the server $\mathcal{S}_{i}$. We denote $\db_{i}$ as the database of each server, which is initialized to an empty list. In particular, for extending $\mathsf{IDFace}$, we assume that there is only one key server $\mathsf{S}_{key}$ in order to minimize the threat of the leakage of the secret key.
We will call the extension of $\mathsf{IDFace}$ ($\mathsf{2PCFace}$, resp.) as $(\mu,\nu)\text{-}\mathsf{IDFace}$ ($(\mu,\nu)\text{-}\mathsf{MPCFace}$, resp.). 
For ease of description, we assume that the device $\mathfrak{d}_{i, j}$ initially provides the input $(\mathbf{x}, id, \text{``enroll"})$ or $(\mathbf{y}, \text{"identify"})$ to the associated server $\mathcal{S}_{i}$ for $\mathsf{Enroll}$ and $\mathsf{Identify}$, respectively. In addition, we omit the threshold parameter $\tau$ for $\mathsf{Identify}$ and transformation parameters $\alpha, \beta$, which are used for $\mathsf{Enroll}$ and $\mathsf{Identify}$, respectively.

\;

\noindent \textbf{Description of $(\mu,\nu)\text{-}\mathsf{IDFace}$.} We denote $(\mathsf{pk}, \mathsf{sk})$ as the public key and secret key pair of the underlying AHE scheme. In addition, we will utilize $\mathsf{IDFace.Enroll}$ and $\mathsf{IDFace.Identify}$ as subroutines. 
Using them, we formally describe $(\mu,\nu)\text{-}\mathsf{IDFace}$ as Figure~\ref{fig:uvidface}:

\begin{figure}[t]
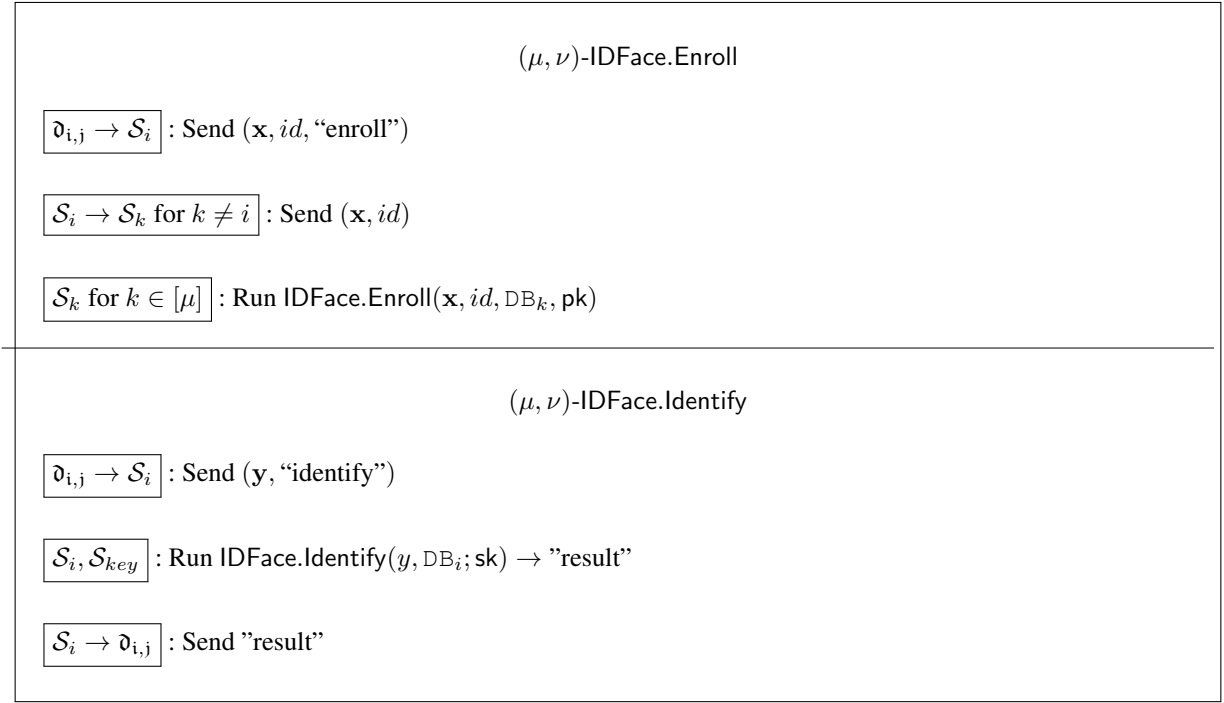

\centering
\fbox{
\parbox{.9\linewidth}{
\vspace{1mm}
\begin{description}
\item \begin{center} $(\mu,\nu)\text{-}\mathsf{IDFace.Enroll}$\end{center}
\vspace{2mm}
    \item \fbox{$\mathfrak{d_{i, j}} \rightarrow \mathcal{S}_{i}$} : Send $(\mathbf{x}, id, \text{``enroll"})$ \\
    \item \fbox{$\mathcal{S}_{i} \rightarrow \mathcal{S}_{k}$ for $ k\neq i$} : Send $(\mathbf{x}, id)$ \\
    \item \fbox{$\mathcal{S}_{k}$ for $k \in [\mu]$} : Run $\mathsf{IDFace.Enroll}(\mathbf{x}, id, \db_{k}, \mathsf{pk})$

\hspace*{-1.15cm}
\rule{.923\textwidth}{.1mm}

\vspace{1mm}

\item \begin{center}$(\mu,\nu)\text{-}\mathsf{IDFace.Identify}$\end{center}
\vspace{2mm}
    \item \fbox{$\mathfrak{d_{i, j}} \rightarrow \mathcal{S}_{i}$} : Send $(\mathbf{y}, \text{``identify"})$ \\
    \item \fbox{$\mathcal{S}_{i}, \mathcal{S}_{key}$} : Run $\mathsf{IDFace.Identify}(y, \db_{i}; \mathsf{sk}) \rightarrow \text{"result"}$ \\
    \item \fbox{$\mathcal{S}_{i} \rightarrow \mathfrak{d_{i,j}}$} : Send $\text{"result"}$
\end{description}
}}
\caption{Description of $(\mu,\nu)\text{-}\mathsf{IDFace}$}
\label{fig:uvidface}
\vspace{-3mm}
\end{figure}

\;

\noindent \textbf{Description of $(\mu, \nu)\text{-}\mathsf{MPCFace}$.}
For designing the extension, a naive approach is to divide each sever $\mathcal{S}_{i}$ into sub-servers $(\mathsf{S}_{i1}, \mathcal{S}_{i2})$ and follow the construction of $(\mu,\nu)\text{-}\mathsf{IDFace}$, using $\mathsf{2PCFace}$ as a subroutine. We will not consider this approach because, in this case, the adversary can recover the whole enrolled template by compromising only two sub-servers $\mathcal{S}_{i1}, \mathcal{S}_{i2}$. Rather, we extend $\mathsf{2PCFace}$ by the computation with $\mu$-parties by the following observation: For an additive share $(\mathbf{x}_{1}, \dots, \mathbf{x}_{\mu})$ of $\mathbf{x}$, \textit{i.e.}, $\oplus_{i=1}^{\mu}\mathbf{x}_{i} = \mathbf{x}$, we have that $\langle \mathbf{x}, \mathbf{y} \rangle = \mathcal{HW}(\oplus_{i=1}^{\mu} (\mathbf{x}_{i} \land \mathbf{y}))$. From this, we extend the phase of broadcasting $\mathbf{x}_{i} \land \mathbf{y}$ in $\mathsf{2PCFace.Identify}$ to $\mu$ parties. 

We will denote $\mathsf{GenShare}(\mathbf{x}; \mu)$ as an extension of the original $\mathsf{GenShare}$ algorithm to make an additive share $(\mathbf{z}_{1}^{\dagger}, \dots, \mathbf{z}_{\mu}^{\dagger})_{\dagger \in \{+,-\}}$ of transformed $\mathbf{x}$ for $\mu$ parties. In addition, we denote $\mathsf{DBAND}$ as an algorithm that takes a database $\db = (id_{i}, \mathbf{z}_{i}^{+},\mathbf{z}_{i}^{-})_{i=1}^{D}$ and two binary strings $\mathbf{y}^{+}, \mathbf{y}^{-} \in \{0, 1\}^{d}$ as inputs and returning tuples of sub-vectors of each $\mathbf{z}_{i}^{+}$ and $\mathbf{z}_{i}^{-}$: $(\mathbf{z}_{i}^{+}[\mathbf{y}^{+}]$, $\mathbf{z}_{i}^{-}[\mathbf{y}^{-}]$, $\mathbf{z}_{i}^{+}[\mathbf{y}^{-}]$, $\mathbf{z}_{i}^{-}[\mathbf{y}^{+}])_{i=1}^{D}$. Using them as subroutines, we formally describe $(\mu,\nu)\text{-}\mathsf{MPCFace}$ as Figure~\ref{fig:uvmpcface}:

\begin{figure}[t]
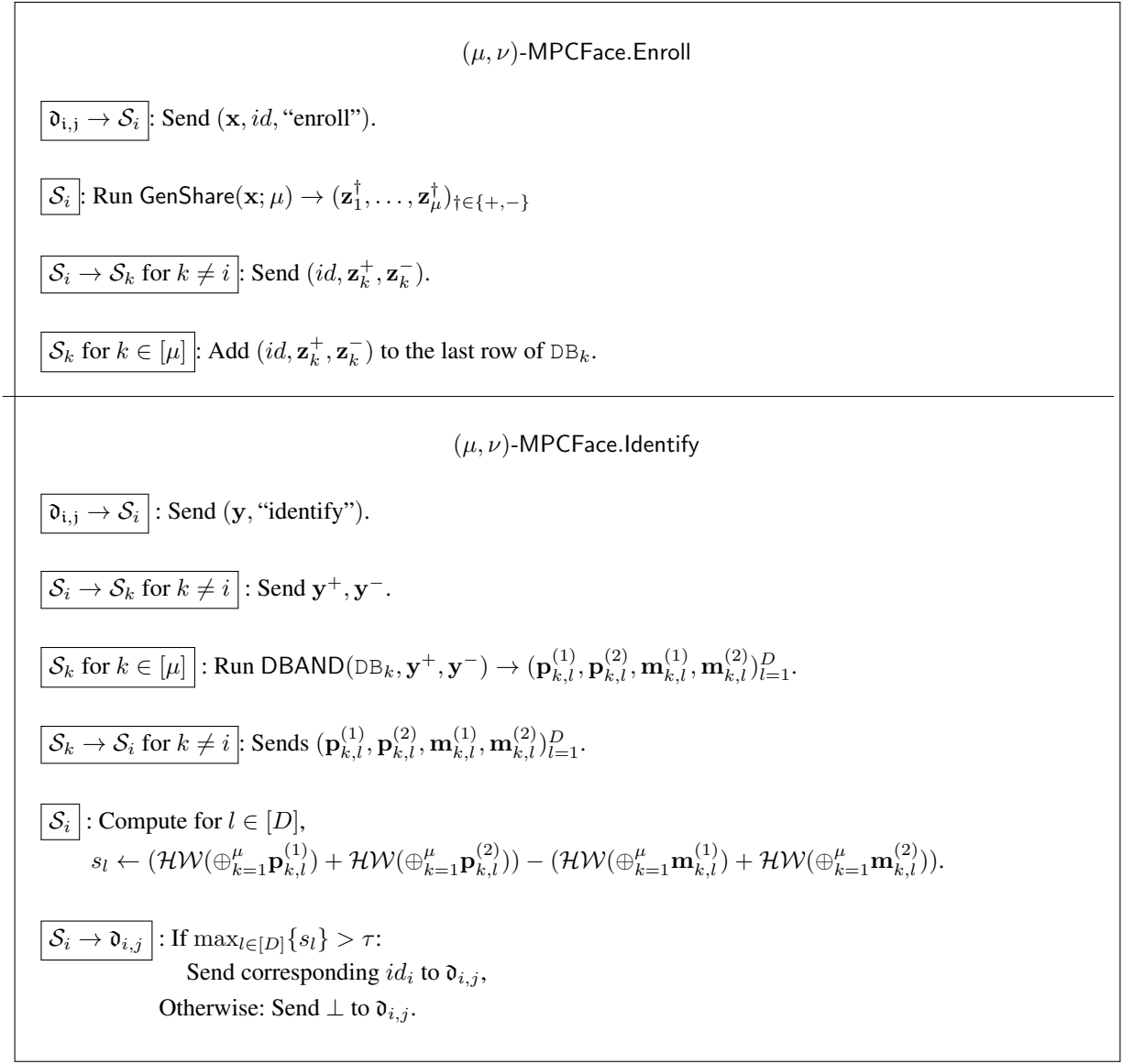

\centering
\fbox{
\parbox{.9\linewidth}{
\begin{description}
\vspace{1mm}
\item \begin{center}$(\mu,\nu)\text{-}\mathsf{MPCFace.Enroll}$\end{center}
\vspace{2mm}
    \item \fbox{$\mathfrak{d_{i, j}} \rightarrow \mathcal{S}_{i}$}: Send $(\mathbf{x}, id, \text{``enroll"})$. \\
    \item \fbox{$\mathcal{S}_{i}$}: Run $\mathsf{GenShare}(\mathbf{x}; \mu) \rightarrow (\mathbf{z}_{1}^{\dagger}, \dots, \mathbf{z}_{\mu}^{\dagger})_{\dagger \in \{+,-\}}$ \\
    \item \fbox{$\mathcal{S}_{i} \rightarrow \mathcal{S}_{k}$ for $ k\neq i$}: Send $(id, \mathbf{z}_{k}^{+}, \mathbf{z}_{k}^{-})$. \\
    \item \fbox{$\mathcal{S}_{k}$ for $k \in [\mu]$}: Add $(id, \mathbf{z}_{k}^{+}, \mathbf{z}_{k}^{-})$ to the last row of $\db_{k}$.

% \vspace{-1mm}

\hspace*{-1.15cm}
\rule{.923\textwidth}{.1mm}

\vspace{1mm}
\item \begin{center}{$(\mu,\nu)\text{-}\mathsf{MPCFace.Identify}$}\end{center}
    \vspace{2mm}
    \item \fbox{$\mathfrak{d_{i, j}} \rightarrow \mathcal{S}_{i}$} : Send $(\mathbf{y}, \text{``identify"})$. \\ 
    \item \fbox{$\mathcal{S}_{i} \rightarrow \mathcal{S}_{k}$ for $k\neq i$} : Send $\mathbf{y}^{+},\mathbf{y}^{-}$. \\
    \item \fbox{$\mathcal{S}_{k}$ for $k \in [\mu]$} : Run $\mathsf{DBAND}(\db_{k}, \mathbf{y}^{+}, \mathbf{y}^{-}) \rightarrow (\mathbf{p}_{k,l}^{(1)},\mathbf{p}_{k,l}^{(2)},\mathbf{m}_{k,l}^{(1)},\mathbf{m}_{k,l}^{(2)})_{l=1}^{D}$. \\
    \item \fbox{$\mathcal{S}_{k} \rightarrow \mathcal{S}_{i}$ for $k\neq i$}: Sends $(\mathbf{p}_{k,l}^{(1)},\mathbf{p}_{k,l}^{(2)},\mathbf{m}_{k,l}^{(1)},\mathbf{m}_{k,l}^{(2)})_{l=1}^{D}$. \\
    \item \fbox{$\mathcal{S}_{i}$} : Compute for $l \in [D]$, \\
    \hspace*{2mm}$s_{l} \gets (\mathcal{HW}(\oplus_{k=1}^{\mu}\mathbf{p}_{k,l}^{(1)}) + \mathcal{HW}(\oplus_{k=1}^{\mu}\mathbf{p}_{k,l}^{(2)}))-(\mathcal{HW}(\oplus_{k=1}^{\mu} \mathbf{m}_{k,l}^{(1)}) + \mathcal{HW}( \oplus_{k=1}^{\mu}\mathbf{m}_{k,l}^{(2)}))$. \\
    \item \fbox{$\mathcal{S}_{i}\rightarrow\mathfrak{d}_{i,j}$} : If $\max_{l \in [D]}\{s_{l} \} > \tau$: \\
    \vspace{1mm}
    \hspace*{15mm}Send corresponding $id_{i}$ to $\mathfrak{d}_{i,j}$, \\
    \vspace{1mm}
    \hspace*{11mm}Otherwise: Send $\bot$ to $\mathfrak{d}_{i,j}$.
\end{description}    
}}
\caption{Description of $(\mu, \nu)\text{-}\mathsf{MPCFace}$.}
\label{fig:uvmpcface}
\vspace{-3mm}
\end{figure}
%
%\vspace{-10mm}

\;

\noindent \textbf{Communication Overheads From Extensions.} For each construction, the advantage of $(\mu,\nu)\text{-}\mathsf{IDFace}$ over $(\mu,\nu)\text{-}\mathsf{MPCFace}$ is that the former can run $\mathsf{IDFace.Identify}$ ``locally". That is, the remaining servers do nothing during the protocol. 
For this reason, additional communication on the former occurs only for broadcasting $(\mathbf{x}, id)$ during $\mathsf{Enroll}$. Concretely, with transformation before sending, the communication cost for this is $(\mu-1)(2d + |id|)$, where $|id|$ is the bit length of $id$. We note that $|id| =\lceil \log_{2}D \rceil$ suffices to distinguish all enrolled identities, so that the asymptotic communication overhead with respect to $\mu$ and $D$ is $O(\mu\log D)$. For our experimental setting where $d=512$ and the database of size $D=1$M, $|id| = 20 > \log D$ suffices, so $(2d + |id|) \approx 0.26$KB.

\begin{table}[t]
    \centering
    \resizebox{.85\linewidth}{!}{
    \begin{tabular}{c|c|ccc|ccc}
        \hline
         \multirow{2}{*}{$\mu$} & \multirow{2}{*}{Parameters $(\alpha, \beta)$} & \multicolumn{3}{c|}{\textsf{Enroll}} & \multicolumn{3}{c}{\textsf{Identify}} \\ \cline{3-8}
            &      &   $\mathsf{2PCFace}$   & $\mathsf{MPCFace}$ & Overhead & $\mathsf{2PCFace}$ & $\mathsf{MPCFace}$ & Overhead \\ \hline
        \multirow{3}{*}{3} & (341,341)  & \multirow{3}{*}{0.13KB}  & \multirow{3}{*}{0.26KB} & \multirow{3}{*}{+0.13KB} & 81.30MB & 162.60MB & +81.30MB   \\
                            & (341,127) & & &  & 30.28MB & 60.56MB & +30.28MB  \\
                            & (341,63) & & & & 15.02MB & 30.04MB & +15.02MB  \\ \hline
        \multirow{3}{*}{6} & (341,341)  & \multirow{3}{*}{0.13KB}  & \multirow{3}{*}{0.64KB} & \multirow{3}{*}{+0.51KB} & 81.30MB & 406.50MB & +325.20MB  \\
                            & (341,127) & & & & 30.28MB & 151.39MB & +121.11MB  \\
                            & (341,63) & & & & 15.02MB & 75.10MB & +60.08MB   \\ \hline
        \multirow{3}{*}{11} & (341,341)  & \multirow{3}{*}{0.13KB}  & \multirow{3}{*}{1.28KB} & \multirow{3}{*}{+1.15KB} & 81.30MB & 813.01MB & +731.71MB  \\
                            & (341,127) & & &  & 30.28MB & 302.79MB & +272.51MB  \\
                            & (341,63) & & & & 15.02MB & 150.20MB & +135.18MB  \\ \hline
        \multirow{3}{*}{21} & (341,341)  & \multirow{3}{*}{0.13KB}  & \multirow{3}{*}{2.55KB} & \multirow{3}{*}{+2.42KB} & 81.30MB & 1626.02MB & +1544.72MB  \\
                            & (341,127) & & & & 30.28MB & 605.59MB & +575.31MB  \\
                            & (341,63) & & & & 15.02MB & 300.41MB & +285.39MB  \\ \hline
            
    \end{tabular}
    }
    \vspace{-3mm}
    \caption{Communication cost of $(\mu,\nu)\text{-}\mathsf{MPCFace}$ with various number of servers and parameters, along with the overhead compared to $\mathsf{2PCFace}$ under the same parameters. The number of identities is set to 1M. We note that the communication cost overhead of $(\mu,\nu)\text{-}\mathsf{MPCFace}$ in \textsf{Enroll} is identical to $(\mu,\nu)\text{-}\mathsf{IDFace}$}
    \label{tab:comm_cost_multi}
    \vspace{-4mm}
\end{table}

But in the latter, communication occurs between servers in both $\mathsf{Enroll}$ and $\mathsf{Identify}$: In $\mathsf{Enroll}$, the additive share of input $\mathbf{x}$ with $id$ is sent to other servers, and in $\mathsf{Identify}$ the output of $\mathsf{DBAND}$ is merged to a single server. Although the size of the former is identical to that from $(\mu,\nu)\text{-}\mathsf{IDFace}$, the latter is quite large. According to the communication cost formula of $\mathsf{2PCFace}$, we have that the latter carries $(\mu-1)(2D\beta + 2d)$, or asymptotically $O(\mu D)$, communication overhead. As we can infer from the result in Table~\ref{tab:comm_cost}, for the dataset with 1M enrolled identities, this overhead becomes extremely large (from 285MB to 1544MB) even when $\mu = 21$. The communication cost overhead with respect to various parameters and the number of servers $\mu$ are summarized in Table~\ref{tab:comm_cost_multi}.

\section{Extending \textsf{IDFace} to Other Types of Biometrics}\label{sec:supp_domain}
Although we focused on protecting face templates in $\mathsf{IDFace}$, we note that the core contribution of our work is to accelerate the computation of inner products in an encrypted domain while preserving the inner product value as much as possible. That is, our work can be adapted to other biometric authentication systems using different types of biometrics where cosine similarity is used for measuring the matching score, such as fingerprint recognition or speaker recognition.
In this section, we investigate the extensibility of $\mathsf{IDFace}$ to these tasks.

\subsection{Extension 1: Speaker Recognition}

We conducted experiments on evaluating the accuracy of three publicly available speaker recognition models with $\mathsf{IDFace}$, including RawNet3~\cite{jung22rawnet3}, MR-RawNet~\cite{kim24jmr_rawnet}, and ReDimNet~\cite{yakovlev24redimnet}. The source codes and pre-trained parameters of all these models are disclosed by the authors.
In particular, for ReDimNet, we utilized the (\texttt{b6, finetuned}) version in their paper, which showed the best performance among the versions disclosed by the authors. We used the cleaned version of VoxCeleb1 test dataset~\cite{nagrani2020voxceleb}, which is widely used for evaluation in speaker recognition tasks. 
We used full utterances of each audio sample in the benchmark dataset for feature extraction.
For the evaluation metrics, we used Error Equal Rate (EER) and minimum Detection Cost Function (MinDCF).
We note that the latter is used for the annual speaker recognition challenge held by NIST~\cite{NISTSRE24}. According to the official evaluation plan, we selected the parameter setting \texttt{id:1} in the audio track, which is $P_{target}=0.01, C_{FalseAlarm}=C_{Miss}=1$. EER and MinDCF can be represented in terms of FAR and FRR, where $\mathsf{FRR}(\tau) = 1 - \mathsf{TAR}(\tau)$, as follows:
\begin{gather}
\mathsf{EER} = \frac{\mathsf{FRR}(\hat{\tau}) + \mathsf{FAR}(\hat{\tau})}{2}, \ \text{where} \ \hat{\tau} = \arg \min_{\tau} \left| \mathsf{FAR}(\tau) - \mathsf{FRR}(\tau) \right|. \nonumber \\ 
\mathsf{MinDCF} = \min_{\tau} \left( C_{\mathsf{miss}} \cdot \mathsf{FRR}(\tau) \cdot (1 - P_{\text{target}}) + C_{\mathsf{FalseAlarm}} \cdot \mathsf{FAR}(\tau) \cdot P_{\mathsf{target}} \right). \nonumber
\end{gather}
As can be inferred from the definitions, smaller EER and MinDCF imply means the model has more discrimination power, \textit{i.e.}, having a better performance.

For parameter selection of $\mathsf{IDFace}$, we followed the same strategy as our analysis in face recognition task. We checked that the RawNet3 and MR-RawNet models output 256-dimensional templates, whereas 192-dimensinoal for ReDimNet. By following our analysis, we select $\alpha = \left\lfloor\frac{2}{3}d\right\rfloor$, which is 170 and 127 for the former two feature extractors and and the latter one, respectively.
In addition, we also select $\beta$ in accordance with $\alpha$ and our accuracy-efficiency trade-off trick: $\beta = 170, 127, 63$ for the former, and $\beta = 127, 63, 31$ for the latter.

The evaluation results can be found in Tab.~\ref{tab:speak}. From this table, we can observe that the accuracy drops in EER for all cases are less than 1\%, so does in MinDCF except for the case of $\beta = 31$ on ReDimNet. This demonstrates that our $\mathsf{IDFace}$ can be employed for protecting voice templates in large-scale identification tasks without significant accuracy degradation, allowing faster identification even than the case of protecting face templates.

\begin{table}[h]
\begin{minipage}{.49\linewidth}
    \resizebox{\linewidth}{!}{
\begin{tabular}{c|c|c|ccc}
\hline
\multirow{2}{*}{Model} & \multirow{2}{*}{Metric} & \multirow{2}{*}{Plain} & \multicolumn{3}{c}{$\mathsf{IDFace}$, Parameters: ($\alpha=170$)} \\ \cline{4-6} 
 &  &  & \multicolumn{1}{c|}{$\beta=170$} & \multicolumn{1}{c|}{$\beta=127$} & $\beta=63$ \\ \hline
\multirow{2}{*}{RawNet3} & EER(\%) & 0.83 & \multicolumn{1}{c|}{1.19} & \multicolumn{1}{c|}{1.22} & 1.30 \\
 & MinDCF & 0.105 & \multicolumn{1}{c|}{0.151} & \multicolumn{1}{c|}{0.135} & 0.158 \\ \hline
\multirow{2}{*}{MR-RawNet} & EER(\%) & 1.38 & \multicolumn{1}{c|}{1.56} & \multicolumn{1}{c|}{1.59} & 1.89 \\
 & MinDCF & 0.131 & \multicolumn{1}{c|}{0.159} & \multicolumn{1}{c|}{0.175} & 0.183 \\ \hline
\end{tabular}
    }
\end{minipage}
\;\;
\begin{minipage}{.49\linewidth}
    \resizebox{\linewidth}{!}{
   \begin{tabular}{c|c|c|ccc}
\hline
\multirow{2}{*}{Model} & \multirow{2}{*}{Metric} & \multirow{2}{*}{Plain} & \multicolumn{3}{c}{$\mathsf{IDFace}$, Parameters: ($\alpha = 127$)} \\ \cline{4-6} 
 &  &  & \multicolumn{1}{c|}{$\beta=127$} & \multicolumn{1}{c|}{$\beta=63$} & $\beta=31$ \\ \hline
\multirow{2}{*}{ReDimNet} & EER(\%) & 0.40 & \multicolumn{1}{c|}{0.69} & \multicolumn{1}{c|}{0.75} & 1.15 \\
 & MinDCF & 0.033 & \multicolumn{1}{c|}{0.089} & \multicolumn{1}{c|}{0.107} & 0.146 \\ \hline
\end{tabular} 
    }
\end{minipage}
\vspace{-3mm}
\caption{Various speaker verification benchmark results of non-protected template extractor (Plain) and IDFace for various ($\alpha$, $\beta$).}\label{tab:speak}
\vspace{-5mm}
\end{table}

\subsection{Extension 2: Fingerprint Recognition}

 We also conducted experiments on evaluating the accuracy of fingerprint recognition models, DeepPrint~\cite{engelsma2019learning}, for $\mathsf{IDFace}$ using FVC2004 dataset~\cite{maio2004fvc2004}. The pre-trained model is publicly available in \href{https://github.com/tim-rohwedder/fixed-length-fingerprint-extractors}{github} by the authors of \cite{rohwedder2023benchmarking}. The evaluation result can be found in Tab~\ref{tab:finger}. Although the baseline performance (Plain) is significantly lower compared to face and speaker recognition models, it is noteworthy that the accuracy drops in EER are still less than 1\%. Note that there are no publicly available large-scale datasets on fingerprint recognition, unlike face recognition or speaker recognition. Therefore, our pre-trained model is trained by synthetic fingerprints data generated by SFinGE~\cite{cappelli2004sfinge}.

\begin{table}[h]
\centering
\begin{tabular}{c|c|c|c|c|c}
\hline
Model & Metric & Plain & $(\alpha,\beta)=(341,63)$ & $(\alpha,\beta)=(341,127)$ & $(\alpha,\beta)=(341,341)$ \\ \hline
DeepPrint~\cite{engelsma2019learning} & EER(\%) & 5.295 & 5.919 & 5.621 & 5.617 \\ \hline
\end{tabular}
\vspace{-3mm}
\caption{Performance comparison between non-protected fingerprint template extractor (Plain) and IDFace for various $\beta$ using DeepPrint.}\label{tab:finger}

\end{table}

\subsection{Discussion on Accuracy Drop}

We briefly share our discussions about the tendency of accuracy drops appeared in each task.
First of all, in contrast of face recognition and fingerprint recognition tasks, we can observe that the speaker recognition models used in our analysis produces smaller dimension of templates than 512. 
Since the error term in $\epsilon_{\alpha, \theta}$ becomes larger as $d$ becomes smaller according to Proposition~\ref{prop_1_full}, one may expect that the accuracy drop in speaker recognition tasks would be larger than that from the other tasks.
However, we observed that the accuracy drop in speaker recognition models is comparable with that of both face and fingerprint recognition tasks. We suspect that this is because the amount of the change of cosine similarity value from our transformation is solely determined by the ratio between the number of nonzeros $(\alpha)$ and the dimension $d$. Here, if we assume that the error term \textit{randomly} contributes to the inner product value in terms of whether the decision of the recognition model is changed or not due to the transformation, then for a large number of sample pairs, we can expect that the effect of our transformation on the accuracy would be determined by the average error $\epsilon_{\alpha, \theta}$ in a macroscopic viewpoint.
We leave a more thorough analysis on this phenomenon as future work.

We also address another non-trivial phenomenon that appeared in Tab.~\ref{tab:speak}: the amount of accuracy drop in MinDCF is slightly larger than that of EER. To this end, we compared the thresholds corresponding to EER and MinDCF for each speaker recognition model, including RawNet3, MR-RawNet, and ReDimNet, obtaining (0.3376, 0.3383, 0.3519) and (0.4720, 0.4944, 0.4339) for each model in EER and MinDCF, respectively. Since the variation of the cosine similarity at the latter is larger than the former according to Fig.~\ref{fig:kirruk}, we can conclude that this phenomenon does in fact coincide with our theoretical analysis.

\section{Additional Figures}\label{sec:supp_J}
In this section, we provide additional figures that were omitted in the main text due to space constraints. In Figure~\ref{fig:encode}, we visualize how the proposed space-efficient encoding technique works, as introduced in Section~\textcolor{iccvblue}{4.2}. The left hand side shows the procedure $\mathsf{Encode}$ to encode $m$ transformed templates $x_{1}, \ldots, x_{m}$ into two vectors $x^{+}$ and $x^{-}$. The right side shows how we can compute the inner product value between each $m$ template and query with $\mathsf{Encode}$ at once, along with the decoding procedure $\mathsf{Decode}$ to retrieve the inner product results. With additive homomorphic encryption, we can process the above computation with the encoded vector $x^{\dagger}$ encrypted for $\dagger \in \{+, - \}$, as described in $\mathsf{IDFace.IP}_{\mathtt{DB}}$. For more detailed information, we recommend the reader refer to Section~\textcolor{iccvblue}{4.2}.

We also provide the visualization of the application scenario of $\mathsf{IDFace}$. Recall that $\mathsf{IDFace}$ enables securing a face identification system through encrypting the template database, without accompanying significant accuracy or efficiency degradation compared to the identification system without protection. For this reason, we expect that our $\mathsf{IDFace}$ can be used for any face identification system where the user physically presents his/her biometrics into the system, such as face identification systems for airports or entrances of the building. We remark that there already have been several real-world use cases, such as several airports in Europe participating in Star Alliance Biometrics\footnote{https://www.staralliance.com/en/biometrics} and securing building entrances in several facilities, e.g., public schools in the United States\footnote{https://wvpublic.org/four-counties-to-implement-facial-recognition-for-school-safety/}$^{, }$\footnote{https://www.nytimes.com/2020/02/06/business/facial-recognition-schools.html} and casinos\footnote{https://www.econnectglobal.com/blog/how-integrated-resorts-and-casinos-are-leveraging-facial-recognition-software-for-increased-security}$^{, }$\footnote{https://journalrecord.com/2022/10/25/casino-uses-facial-recognition-technology-to-supplement-security/}. For more details about the application scenario and the threat model, we recommend the reader refer to Section~\textcolor{iccvblue}{2.1} (biometric identification system) and Section~\textcolor{iccvblue}{3.2} (biometric identification system with database encryption). In addition, the threat model we considered in $\mathsf{IDFace}$ is also described in Section~\textcolor{iccvblue}{5.3}.

\begin{figure*}[t]
    \centering
    \includegraphics[width=.99\textwidth]{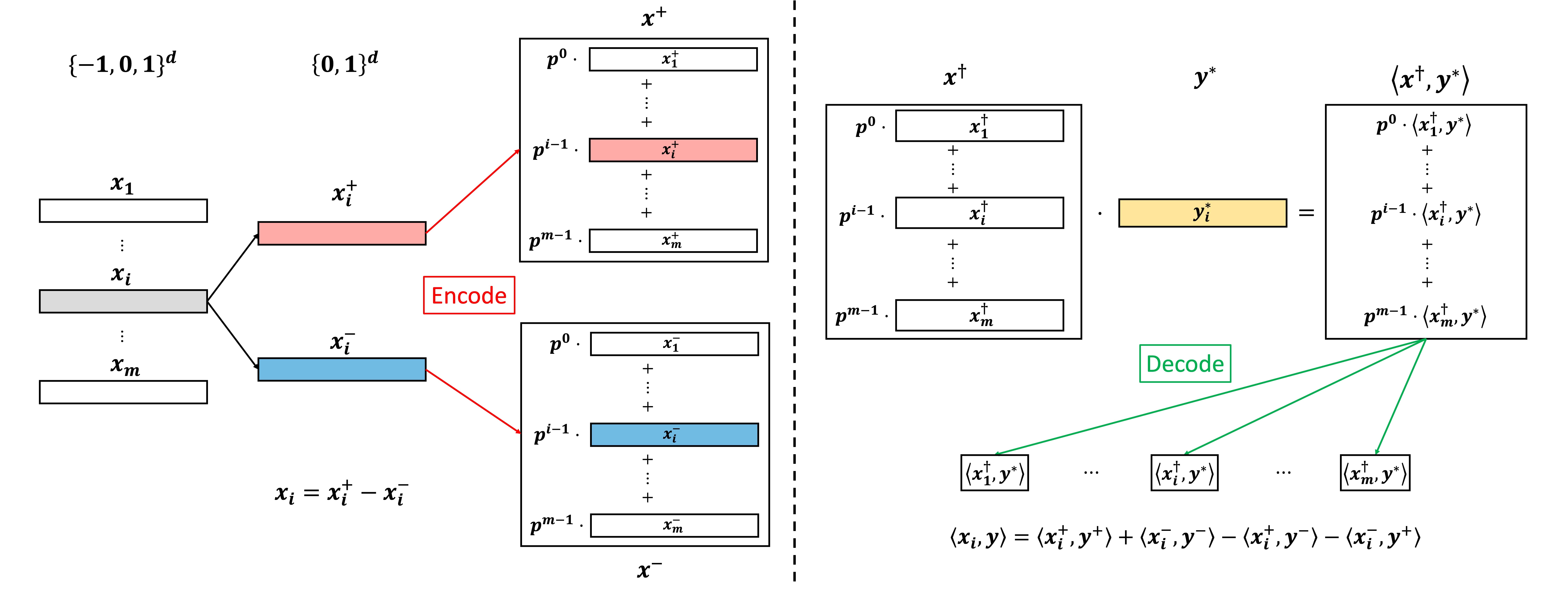}
    \caption{The overview of space efficient encoding technique. More details are provided in Section~\textcolor{iccvblue}{4.2}.}
    \label{fig:encode}
\end{figure*}

\begin{figure*}[t]
    \centering
    \includegraphics[width=.99\textwidth]{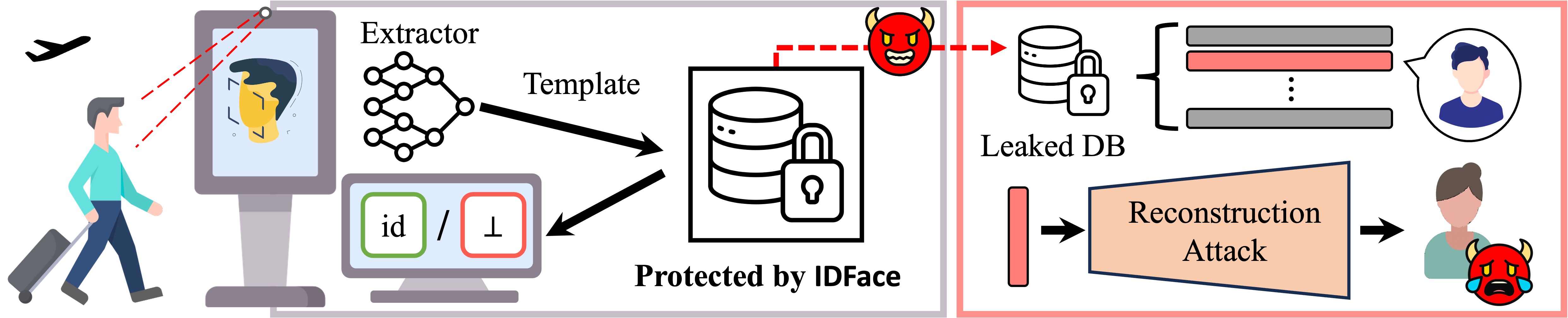}
    \caption{Target application scenario of $\mathsf{IDFace}$.}
    \label{fig:scenario}
\end{figure*}

\end{document}